\documentclass[11pt]{article}     % onecolumn (ditto)
\usepackage{a4wide}
\usepackage{fullpage}
\usepackage{graphicx}
\usepackage{amsmath}
\usepackage{amssymb}
\usepackage{hyperref}
\usepackage{cleveref}
\usepackage{authblk}
\usepackage{booktabs,multirow,array,caption}
\usepackage{pifont}
\usepackage{bigdelim}
\newcommand{\cmark}{\ding{51}}%
\newcommand{\xmark}{\ding{55}}%

\usepackage[vlined,english,boxed,linesnumbered]{algorithm2e} % for styled algorithms
      \IncMargin{1em}  % margin of algorithms
     \SetKwInput{Input}{Input}
     \SetKwInput{Output}{Output}
     \SetKw{Return}{Return}
     \SetKwIF{If}{ElseIf}{Else}{If}{then}{else if}{else}{endif}
     \SetKwFor{For}{For}{do}{done}
     \SetKwFor{for}{For}{}{}
     \SetKwFor{While}{While}{do}{done}
     \SetKwFor{Repeat}{Repeat}{}{}
     \SetKwRepeat{Do}{Do}{while}
      \SetKwIF{Do}{notused}{WhileOfDo}{Do}{}{notused}{while}{done}
     \SetKwFor{Procedure}{Procedure}{}{}
     \SetKwFor{Function}{Function}{}{}
     \SetKw{Return}{Return}%
     \SetKw{return}{return}%
     \SetKwComment{Comment}{/* }{ */}
     \SetSideCommentRight
     \DontPrintSemicolon

\newcommand{\lightparagraph}[1]{\smallskip\textit{#1 }}

\usepackage{xcolor}
\usepackage[normalem]{ulem} % \sout{} for striking out (barré)
\newcommand{\lv}[1]{\textcolor{blue}{LV: #1}}

\newcommand{\todolater}[1]{\textcolor{orange}{}}
\long\def\begtodolater#1\endtodolater{\begingroup\color{orange}\endgroup}
\long\def\begskip#1\endskip{}

\newtheorem{theorem}{Theorem}
\newtheorem{proposition}{Proposition}
\newcommand{\qed}{\hfill \fbox{} \medskip}
\newenvironment{proof}[1][Proof]{\medskip\noindent\emph{#1. }\ignorespaces}{\qed\medskip\par\noindent\ignorespacesafterend}
\newtheorem{lemma}{Lemma}
\newtheorem{corollary}{Corollary}

\newtheorem{claim}{Claim}

\newcommand{\optproblem}[2]{\medskip\fbox{\parbox{0.92\textwidth}{\textsc{#1.} #2}}\medskip}

\DeclareMathOperator{\processcosts}{FinalizeCosts}
\DeclareMathOperator{\removearc}{RemoveArc}
\DeclareMathOperator{\popmin}{PopMin}

\DeclareMathOperator{\reorder}{Reorder}

\DeclareMathOperator{\dijkstra}{DijkstraFromSet}
\newcommand{\intervs}[1]{{\cal I}_{#1}}
\let\bottom=\perp

\begin{document}

\title{Minimum-Cost Temporal Walks 
under Waiting-Time Constraints
in Linear Time\footnote{This work was suported by the French National Research Agency (ANR) through project Multimod with reference number ANR-17-CE22-0016.}}
%\subtitle{Do you have a subtitle?\\ If so, write it here}

%\titlerunning{On The Complexity of Maximising Temporal Reachability via Trip Temporalisation}        % if too long for running head

\author[1]{Filippo Brunelli}
\author[1]{Laurent Viennot}

%\authorrunning{Short form of author list} % if too long for running head

\affil[1]{\small Université Paris Cité, Inria, CNRS, Irif, France}
% not yet \affil[2]{\small DIENS, École normale supérieure, PSL University, Inria Paris, France}

%\date{Received: date / Accepted: date}
% The correct dates will be entered by the editor

\maketitle

\begin{abstract}
In a temporal graph, each edge  is available at specific points in time. Such an availability point is often represented by a ``temporal edge'' that can be traversed from its tail only at a specific departure time, for arriving in its head after a specific travel time. In such a graph, the connectivity from one node to another is naturally captured by the existence of a temporal path where temporal edges can be traversed one after the  other. When imposing constraints on how much time it is possible to wait at a node in-between two temporal edges, it then becomes interesting to consider temporal walks where it is allowed to visit several times the same node, possibly at different times.

%Bentert et al. (Applied Network Science 2020) proposed an algorithm for computing minimum-cost temporal walks from a single source under waiting-time constraints. It requires a transformation of the temporal graph and performs a Dijkstra traversal at each time instant where an edge event occurs. The resulting time complexity is within a logarithmic factor from linear. We propose a single pass algorithm when temporal edges have strictly positive travel times and when they are provided in input by non-decreasing departure time and also by non-decreasing arrival time. Its time complexity is linear with this appropriately sorted input. Interestingly, we show that a logarithmic factor in the time complexity appears to be necessary if the input contains only one ordering of the temporal edges (either by arrival times or departure times).
We study the complexity of computing minimum-cost temporal walks from a single source under waiting-time constraints in a temporal graph, and ask under which conditions this problem can be solved in linear time.  Our main result is a linear time algorithm when the input temporal graph is given by its (classical) space-time representation.
%temporal edges are provided in input by non-decreasing departure time and also by non-decreasing arrival time. 
We use an algebraic framework for manipulating abstract costs, enabling the optimization of a large variety of criteria or even combinations of these. It allows to improve previous results for several criteria such as number of edges or overall waiting time even without waiting constraints. It saves a logarithmic factor for all criteria under waiting constraints.
%This result is somehow optimal: 
Interestingly, we show that a logarithmic factor in the time complexity appears to be necessary with a more basic input consisting of a single ordered list of temporal edges (sorted either by arrival times or departure times). We indeed show equivalence between the space-time representation and a representation with two ordered lists.

\end{abstract}

%\pagebreak

\section{Introduction} \label{sec:intro}

Computing shortest paths is certainly one of the most fundamental problems within algorithmic graph theory, as well as one of the most important subroutine for a large diversity of applications in networks. While its complexity has been extensively covered in the context of static graphs, there is still room for improvement in temporal graphs, where the edge set evolves with time. Temporal graphs arose with the need to better model contexts where the appearance of interactions or connections depends on time, such as epidemic propagation or transport networks.
For example, in a flight network, nodes represent airports while each edge corresponds to a flight and is labeled with a departure time and a travel time. The natural notion of connectivity is then grasped by temporal paths where edges appear in chronological order and can be traversed one after the other. Starting with the work on time-dependent networks~\cite{CookeH1966} and the telephone problem~\cite{Bumby1981}, the discrete time version of temporal graphs we consider here was already investigated in~\cite{Berman1996,Nachtigall1995,PallottinoS1998} and introduced later in various contexts ranging from social interactions to mobile networks and distributed computing (see e.g. \cite{Casteigts2012,Michail2016,Latapy2018}). This classical point-availability model of temporal graph is the following. The availability of an edge $(u,v)$ at time $\tau$ is modeled by a temporal edge $e=(u,v,\tau,\lambda)$. It represents the possibility to traverse the edge from $u$ at time exactly $\tau$ with arrival in $v$ at time $\tau+\lambda$. We refer to $\tau$ and $\tau+\lambda$ as the departure time and arrival time of $e$ respectively, while $\lambda\ge 0$ is called the travel time of $e$. A temporal walk can then be defined as a sequence of temporal edges such that each temporal edge arrives at the departing node of the next one, and the arrival time of each temporal edge is less or equal to the departure time of the next one. The inequality means that it is possible to wait at the node in-between two consecutive temporal edges. We distinguish such a walk from a temporal path, which is a temporal walk visiting at most once any node.

\begin{figure}[t]
    \begin{minipage}[c]{.34\textwidth}
    \includegraphics[width=\textwidth]{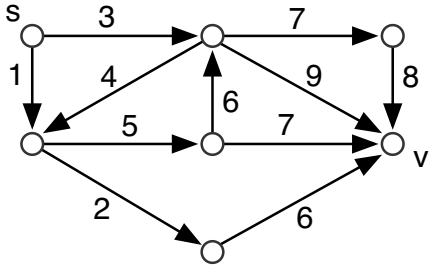}
    \end{minipage}
    \hfill
    \begin{minipage}[c]{.65\textwidth}
    \caption{A temporal graph: each edge is labeled by its departure time and travel times are all one here for simplicity. Several temporal paths or walks from source node $s$ to node $v$ optimize different criteria (we identify those by the consecutive labels of their edges): 1,2,6 has earliest arrival time ($7=6+1$), 3,4,5,7 has shortest duration ($5=7+1-3$), 3,9 has fewest number of edges (2), and 3,4,5,6,7,8 has minimum overall waiting time (0).}
    \label{fig:temp-walks}
    \end{minipage}
    \hfill
\end{figure}

%Computing shortest temporal paths or walks is definitely one of the most central tasks concerning temporal graphs.
While the notion of ``shortest'' path is fairly standard for static graphs, there exists several natural extensions for defining ``shortest'' temporal walks. Indeed, the following natural criteria can be optimized: earliest arrival time, shortest duration,  or fewest number of edges, to name most popular ones. See \Cref{fig:temp-walks} for an example showing the different paths or walks resulting from optimizing these.
Most criteria result in optimal temporal walks which are indeed temporal paths. Minimizing overall waiting time at nodes is a notable exception where walks obviously help compared to paths (see \Cref{fig:temp-walks}).
Indeed, optimizing different criteria might appear as different problems and the single-source optimal path/walk problem has mostly been addressed with different algorithms for different criteria, resulting in different time complexities~(see e.g. \cite{BuiXuanFJ2003,Michail2016,Wu2016}). Most of them assume that the input is given in some specifically sorted format, and their running time is either linear in the number $M$ of temporal edges or within a logarithmic factor at most $\log M$ from it. Recently, generic algorithms allowing to solve the problem for any criteria, or a linear combination of them, were proposed~\cite{BentertHNN2020,BrunelliCV2021} but inherently incur a logarithmic factor. This raises the question of which criteria have a linear time algorithm given a pre-sorted input. This paper addresses it.

\smallskip
\textbf{Related work.}
Interestingly, after numerous works inspired by Dijkstra algorithm (see e.g. \cite{Berman1996,BrodalJ04,Nachtigall1995,PallottinoS1998}), a linear-time algorithm for earliest arrival time was first claimed in~\cite{Wu2014,Wu2016} with a similar algorithm as~\cite{Dibbelt2013,Dibbelt2018} through a single scan of temporal edges ordered by non-decreasing departure time. Both works assume positive travel times, ensuring that temporal paths are strict in the sense that departure times strictly increase along a temporal path. Single-source shortest duration temporal paths are then obtained by basically solving the profile problem, that is computing the earliest arrival time at each node for each possible departure time at the source, which seems more difficult. However, a more intricate version of the scanning algorithm~\cite{Wu2014,Dibbelt2013} solves it in $O(M\log \Delta)$ time where $\Delta\le M$ is the maximum number of temporal edges with same head. A linear time algorithm was later included in~\cite{Wu2016} by taking as input a representation of the temporal graph as a static graph. This representation is similar to the classical ``space-time'' (or ``time-expanded'') graph~\cite{PallottinoS1997,PallottinoS1998,SchulzWW2000,MullerHSWZ04,Michail2016} where each node is split into node events, one for each time where a temporal edge arrives to it or departs from it, and each temporal edge is turned into an arc between two node events.
This ``space-time'' approach is also used for temporal paths with fewest edges (a.k.a. shortest temporal paths)~\cite{Wu2016} but results in an $O(M \log M)$ time complexity as Dijkstra algorithm is used on this static graph which has $\Theta(M)$ vertices and $\Theta(M)$ edges. It was thus unclear whether linear time is possible for fewest edges, and other criteria, such as overall waiting time which seems even harder as it requires to consider walks rather than just paths. Moreover, if linear-time could be explained by the use of a sorted input for earliest arrival time, it left open what property of the space-time representation is enabling linear time for shortest duration.

A further level of difficulty happens when the waiting time at each node is bounded: the computation of temporal paths becomes NP-hard~\cite{CasteigtsHMZ2021}.
Note that such waiting restrictions are natural in several contexts such as epidemic propagation or flight networks. For example, in an epidemic propagation model where nodes represent individuals and temporal edges represent contacts, an upper bound on waiting time allows to take into account the fact that a contaminated individual is contagious only during a time interval, and cannot spread the disease after. 
Despite this hardness result, a recent break-through~\cite{BentertHNN2020} shows that computing single-source optimal walks is still possible under such waiting-time constraints in $O(M\log M)$ time. Similarly to optimizing overall waiting time, such constraints indeed impose to switch from paths to walks for other criteria also. The algorithm is generic as it optimizes a linear combination of all classical criteria.
%Their algorithm relies on first transforming the temporal graph and then performing a Dijkstra computation for each time instant when a temporal edge departs or arrives. More precisely, it first builds an equivalent temporal graph where all edges have zero travel time. This is done by adding a dummy node with appropriate waiting restrictions for each temporal edge. Note that this can considerably increase the number of nodes. Then, it scans  time instants when temporal edges occur in increasing order. For each time instant $t$, a static directed graph is constructed and a Dijkstra computation allows to update the reachability of nodes with temporal edges up to time $t$.
%The algorithm also provides for each such node $v$ an optimal walk from the source node to $v$ according to a linear combination of most classical cost criteria. The $O(M\log M)$ complexity results from both sorting the temporal edges according to time and the use of Dijkstra algorithm on graphs that can have up to $\Theta(M)$ nodes and $\Theta(M)$ edges.
The $\log M$ factor comes from using several calls to Dijkstra algorithm on graphs that can have up to $\Theta(M)$ nodes and $\Theta(M)$ edges. This leaves open whether linear time is possible for any criterion under waiting-time constraints.

Other related work~\cite{BuiXuanFJ2003,DehneOS2012,Simard21} consider a more general model where each temporal edge is available during an interval of time (rather than a point), resulting in higher time complexities (e.g. quadratic for shortest duration) although waiting-time constraints are not considered.

\smallskip
\textbf{Contribution.}
%\fb{We will consider temporal graphs that do not admit cycles of temporal edges with zero travel time and same departure time. We propose a rather simple temporal-edge scanning algorithm that runs in linear time in such graphs, and the input is given as two appropriately sorted lists of temporal edges: one according to arrival times, and the other according to departure times.}
We propose a temporal-edge scanning algorithm for single-source minimum-cost walks that runs in linear time given an acyclic space-time representation. The acyclic assumption means that the temporal graph do not contains a cycle of temporal edges with zero-travel time at any time instant.
It is obviously satisfied when travel times are positive which is the case in many practical settings such as transit networks.
%when travel times are strictly positive, and the input is given as two appropriately sorted lists of temporal edges: one according to arrival times, and the other according to departure times.
The algorithm can handle waiting-time constraints as defined in~\cite{BentertHNN2020}. We use an algebraic definition of cost similarly to~\cite{BrunelliCV2021} and in the spirit of~\cite{Sobrinho2005,Griff2010}. This enables a large variety of cost definitions, including in particular the linear combination considered in~\cite{BentertHNN2020}, or compositions such as shortest-fastest~\cite{Simard21}. This shows that linear time is possible for all criteria given an acyclic space-time representation when this was unknown for fewest edges and overall waiting time. Moreover, this holds even with waiting-time constraints while a logarithmic factor was previously necessitated for all criteria in that context. Our algorithm also solves the profile problem (again in linear time) with waiting-time constraints. No such algorithm was previously claimed, although we suspect that \cite{BentertHNN2020} can be adapted for that, but with a logarithmic slowdown. See \Cref{tab:compl} for a comparison with previous work.
%Assuming positive travel times is related to the concept of strict temporal path where departing times are required to strictly increase along a temporal path. This assumption is quite natural in many practical settings such as transit networks.
%, and . 
% % From a practical perspective, an adequate input is initially computed in $O(M\log M)$ time from any classical representation as usually assumed in practical works~\cite{Dibbelt2018,Wu2016} for better efficiency. Then, our algorithm enables an arbitrary number of queries from different sources in $O(M)$ time each. In particular, this implies a logarithmic speed-up for the all-pair problem. A trivial modification allows to restrict queries to a given interval of time within the lifetime of the temporal graph.

\begin{table}[t]
  \newcommand{\lambdapos}{$\lambda>0$}
  \newcommand{\betainf}{$\beta=\infty$}
  \newcommand{\both}{$\lambda>0$, $\beta=\infty$}
  \newcommand{\none}{--}
    \centering
    
    \scalebox{1}{
    \begin{tabular}{@{} p{3.5cm}|p{5.1cm}ccl @{}}
    \toprule[0.5pt]
    
     \multirow{1}{*}{Criterion} &  \multirow{1}{*}{Time complexity}  &  Model & Waiting restr. & Input\\
%      &   &  $\lambda > 0$ & $\beta = \infty$ & pre-sorted\\
     
    \midrule[0.5pt]
    
      Earliest arrival time
      & $O(M)$ \cite{Dibbelt2018,Wu2016} 
      & \lambdapos & \xmark & pre-sorted \\
      
      Shortest duration
      & $O(M)$ \cite{Wu2016} 
      & \lambdapos & \xmark & space-time \\
      
      Fewest edges
      & $O(M\log \Delta)$ \cite{Wu2016} 
      & \lambdapos & \xmark & pre-sorted \\
 
% The O(n^2M(n+\log M)) is for interval model
%      \multirow{2}{*}{Shortest-fastest}& \rdelim\{{2}{17.5mm}[\begin{tabular}{l}
%           $O(M \log M)$ \cite{Simard21} \\
%           $O(n^2M(n+\log M))$ \cite{Simard21}
%      \end{tabular}]   & \xmark & \xmark & \xmark \\
%       &  & \cmark  & \xmark & \xmark \\   
%      Shortest-fastest & $O(M \log M)$ \cite{Simard21} & \xmark & \xmark & -- \\
  
     Any above & \rdelim\}{3}{*}[ $O(M\log M)$ \cite{BentertHNN2020}] & \multirow{3}{*}{\none} & \multirow{3}{*}{\cmark} & \multirow{3}{*}{any}  \\  
     Overall waiting time & & & & \\
     Linear combination & & & & \\

%      Linear combination & $O(M\log M)$ \cite{BentertHNN2020} & \cmark & \cmark & \xmark \\
      
     Profile & $O(M\log \Delta)$ \cite{Dibbelt2018}
     & \lambdapos & \xmark & pre-sorted\\   
    \midrule[0.5pt]
    
     & {This Paper} 
     &  &\\
     %$\lambda=0$ & $\beta\neq \infty$ \\
    \midrule[0.5pt]
    
     \multirow{2}{*}{Any above}& \rdelim\{{2}{*}[%17.5mm
        \begin{tabular}{l}
        $O(M)$ Algorithm~\ref{alg:gal} \\
        $O(M \log n)$ Algorithm~\ref{alg:0del}
      \end{tabular}]   & acyclic & \cmark & space-time \\
       &  & \none & \cmark & space-time \\   
          
     Overall waiting time & $\Omega(M\log M)$ \Cref{th:lb}
     & \lambdapos & \xmark & pre-sorted\\

          \bottomrule[0.5pt]
    \end{tabular}
    }
    \caption{Best time complexities for solving single-source optimal temporal walks for various criteria in a temporal graph with $M$ temporal edges, $n$ nodes, and where a node has at most $\Delta\le M$ temporal edges entering it. The ``Model'' column indicates if positive travel times are assumed with ``$\lambda > 0$''; ``acyclic'' stands for the more general setting where no cycle of zero-travel-time edges occurs at any time instant; and a dash stands for the general model where such cycles can occur. A check-mark in column ``Waiting restr.'' indicates that waiting-time constraints are supported while a cross indicates that unrestricted waiting is assumed. The ``Input'' column indicates if the input is required to be a ``pre-sorted'' list of temporal edges, or a ``space-time'' representation, or a list of temporal edges in ``any'' order.} \label{tab:compl}
\end{table}

Our algorithm does not work directly on the space-time representation but on two ordered lists of temporal edges, one where edges arriving at any given node are sorted by non-decreasing arrival time and one where edges departing from any given node are sorted by non-decreasing departure time. The first list must also follow a certain property related to the acyclic requirement on the input. We call this representation a ``doubly-sorted'' representation of the temporal graph. Each list can easily be computed in linear-time from a space-time representation through a procedure similar to topological ordering. We indeed show that this doubly-sorted representation is equivalent to the space-time representation in the sense that one can be computed from the other in linear time and space. In particular, when travel times are positive, two lists sorted by non-decreasing arrival time and non-decreasing departure time respectively form such a doubly-sorted representation.
%, shedding light on why the space-time representation could enable linear time for shortest duration.

%Saving a logarithmic factor when the input is sorted is not as simple as one might think as 
Interestingly, a single sorted list of temporal edges is not sufficient for obtaining linear time. We show a lower bound of $\Omega(M\log M)$ with such a ``singly-sorted'' representation % \fb{should we switch to single and double sorted?} may be better english, but there are a lot to change
for algorithms using comparisons only, which is indeed a desirable algorithmic feature in the algebraic approach for supporting a wide variety of abstract costs. %The idea for obtaining this lower bound is to reduce a sorting problem to the computation of one-to-all shortest duration walks in a temporal graph with $O(M)$ temporal edges.
This lower bound holds even if the input temporal edges are required to be given either by non-decreasing arrival times or by non-decreasing departure times.  This shows that our requirement for a ``doubly-sorted'' representation with both orderings is somehow necessary for allowing linear time computation. It also sheds light on why the equivalent space-time representation could enable linear time for shortest duration.
%
%Interestingly, we show that the classical ``space-time'' representation is equivalent to our representation with two sorted lists of temporal edges as each one can be computed from the other in linear time.

%Our algorithm actually works in a more general setting where zero travel time are allowed, but where no cycle of edges with zero travel time occurs at any time instant. Positive travel times obviously imply such acyclicity. 
%Note that assuming positive travel times is quite natural in many practical settings such as transit networks, and that they are related to the concept of strict temporal paths where departing times are required to strictly increase along a temporal path. 
%It also works with a wider range of pairs of orderings of temporal edges as long as they satisfy a certain doubly-sorted requirement related to this acyclicity. We show that the required doubly-sorted representation of the temporal graph is equivalent to the space-time representation in the sense that one can be computed from the other in linear time and space, shedding light on why the space-time representation could enable linear time for shortest duration.

Finally, we show how to handle the setting where cycles of edges with zero travel time can occur. It is then possible to compute for a fixed source an adequate pair of orderings allowing our algorithm to run correctly for that source. This pair can be computed in $O(M\log n)$ time where $n$ is the number of nodes, allowing to reduce the complexity from $O(M\log M)$ to $O(M\log n)$ compared to previous work. Note that this $\log n$ factor seems mandatory as the single-source shortest problem in static graphs has then a trivial reduction to our problem. Note also that $M$ is not bounded with respect to $n$ and can be much larger in practice. 

%During a single scan of temporal edges, our algorithm updates the minimum cost for reaching any edge through a walk from the source using edges scanned so far. 
Our main new technique consists in maintaining at each node a list of intervals spanning a sliding window of outgoing temporal edges. It allows to update in constant time the cost of candidate minimum-cost walks departing in a time interval. These intervals may be split as temporal edges are scanned, and a careful use of the two orderings of temporal edges given as input allows to manage them with linear amortized complexity. We think that this technique is a valuable contribution and could appear useful for other temporal graph problems involving temporal connectivity such as computing temporal betweenness~\cite{BussMNR2020} or delay-robust temporal walks~\cite{FuchsleMNR22}.

%Besides its simplicity compared to~\cite{BentertHNN2020}, our algorithm also offers a slight improvement in terms of complexity. 
%From a theoretical perspective, it shows that it is possible to save a $\log M$ factor for a wide variety of shortest path problems with appropriately sorted input. Such linear time algorithms were previously known only for earliest arrival time and shortest duration, both with unrestricted waiting. We thus extend these results to waiting-time constraints and to a wide variety of optimization criteria including fewest number of edges and overall waiting time, to name a few for which this was unknown. 

The paper is organized as follows. After defining the main notions in Section~\ref{sec:def}, we first present, as a warm-up for handling waiting constraints, a simple linear-time algorithm allowing to compute all temporal edges appearing in any temporal walk from a given source assuming positive travel times. It allows to compute single-source earliest arrival time walks. We then introduce an algebraic definition of cost in \Cref{sec:gal} and present an algorithm solving the following more general problem: given a source node $s$, compute for each temporal edge $e$, the minimum-cost of a temporal walk from $s$ having $e$ as last edge (if such a walk exists). Maintaining a tentative minimum-cost for each edge requires additional data-structures, and the algorithm allows more general ordering of edges as input in the more general acyclic setting. In Section~\ref{sec:solving}, specific algebraic cost structures are proposed to solve the single-source optimal temporal walk problem for most classical criteria, and combinations of them, as well as the profile problem. A lower bound on the time complexity of the minimum overall waiting time problem is then given in Section~\ref{sec:lb}. We show in Section~\ref{sec:representations} the equivalence between the space-time representation of a temporal graph and the doubly-sorted representation required by our algorithm. Finally, we show in Section~\ref{sec:zero} how to extend our algorithm in the setting where cycles of edges with zero travel time can occur.

\section{Preliminary definitions}
\label{sec:def}

A \emph{temporal graph} is a tuple ${G}=(V,{E},\alpha,\beta)$, where $V$ is the set of \emph{nodes}, $E$ is the set of temporal edges and $\alpha,\beta\in [0,+\infty]^V$ are minimum and maximum waiting-times at each node. A \emph{temporal edge} $e$ is a quadruple $(u,v,\tau,\lambda)$, where $u\in V$ is the \emph{tail} of $e$, $v\in V$ is the \emph{head} of $e$, $\tau\in\mathbb{R}$ is the \emph{departure time} of $e$, and $\lambda\in\mathbb{R}_{\geq0}$ is the \emph{travel time} of $e$. We also define the \emph{arrival time} of $e$ as $\tau+\lambda$, and we let $dep(e)=\tau$ and $arr(e)=\tau+\lambda$ denote the departure time and arrival time of $e$ respectively. For the sake of brevity, we often say edge instead of temporal edge.
%We call \emph{strict temporal graph} a temporal graph such that all the temporal edges have (strictly) positive travel time. 
We let $n=|V|$ and $M=|E|$ denote the number of nodes and edges respectively.

Given a temporal graph ${G}=(V,{E},\alpha,\beta)$ a \emph{walk} $Q$ from $u$ to $v$, or a \emph{$uv$-walk} for short, is a sequence of temporal edges $\langle e_1=(u_1,v_1,\tau_1,\lambda_1),\ldots,e_k=(u_k,v_k,\tau_k,\lambda_k)\rangle\subseteq E^k$ such that $u=u_1$, $v_k=v$, and, for each $i$ with $1<i\leq k$, $u_i=v_{i-1}$ and $a_{i-1}+\alpha_{u_i} \le \tau_i \le a_{i-1} + \beta_{u_i}$ where $a_{i-1}=\tau_{i-1}+\lambda_{i-1}$ is the arrival time of $e_{i-1}$.  %We say that $e_k$ \emph{terminates} walk $Q$ and that edge $e$ is an \emph{arrival edge} with respect to node $v$ whenever the head of $e$ is $v$ and there exists an $sv$-walk terminating with $e$.
Note that the waiting time $\tau_i-a_{i-1}$ at node $u_i$ is constrained to be in the interval $[\alpha_{u_i},\beta_{u_i}]$.
We say that \emph{waiting is unrestricted} when $\alpha_v=0$ and $\beta_v=+\infty$ for all $v\in V$. Note that for positive travel times, such a walk is \emph{strict} in the sense that $\tau_{i-1} < \tau_i$ for $1<i\leq k$ as the constraint $a_{i-1}+\alpha_{u_i} \le \tau_i$ implies $\tau_i \ge a_{i-1}=\tau_{i-1}+\lambda_{i-1} > \tau_{i-1}$ for $\lambda_{i-1} > 0$.
The \emph{departing time} $dep(Q)$ of $Q$ is defined as $\tau_1$, while the \emph{arrival time} $arr(Q)$ of $Q$ is defined as $\tau_k+\lambda_k$. We say that a temporal edge $e=(x,y,\tau,\lambda)$ \emph{extends} $Q$ when $x=v_k$ and $arr(Q)+\alpha_x\le \tau\le arr(Q)+\beta_x$. When $e$ extends $Q$, we can indeed define the walk $Q.e=\langle e_1,\ldots, e_k,e\rangle$ from $u$ to $y$. Moreover, we also say that $e$ extends $e_k$ as it indeed extends any walk $Q$ having $e_k$ as last edge. %\fb{Add zero-walk}
More generally, we say that an edge $f=(x,y,\tau,\lambda)$ \emph{half-extends} an edge $e=(u,v,\tau',\lambda')$ when $x=v$ and $arr(e)+\alpha_x\le \tau$. Note that $f$ half-extends $e$ whenever $f$ extends $e$.
We also say that an edge $e$ with head $v$ is \emph{$s$-reachable} when there exists an $sv$-walk ending with edge $e$.
A \emph{zero-walk} is a walk consisting of temporal edges with same departure time and with zero travel time, and going through nodes with zero minimum waiting constraint. More formally, we define a zero-walk as a walk $\langle e_1=(u_1,v_1,\tau_1,\lambda_1),\ldots,e_k=(u_k,v_k,\tau_k,\lambda_k)\rangle$ such that $\tau_i=\tau_j$ for $i,j\in [k]$, $\lambda_i=0$ and $\alpha_{u_i} = 0$ for $i\in [k]$. Such a zero-walk is called a zero-cycle when $u_1 = v_k$. We say that a temporal graph $G$ is \emph{zero-acyclic} if there are no zero-cycle in $G$.

Let us now introduce some orderings of temporal edges with respect to certain temporal criteria. We say that an ordering $E^{ord}$ of all edges is \emph{half-extend-respecting} when
%the edges of any walk $Q$ in $G$ appear in order in $E^{ord}$. Equivalently, $E^{ord}$ is half-extend-respecting when
for any pair $e,f\in E$ of edges such that $f$ half-extends $e$, then $e$ appears before $f$ in $E^{ord}$ which is denoted by $e <_{E^{ord}} f$. We also write $e \leq_{E^{ord}} f$ for $e <_{E^{ord}} f$ or $e = f$.
Note that the edges of any walk $Q$ in $G$ must appear in order in such an ordering $E^{ord}$ as each edge of $Q$ half-extends the edge preceding it.
We will show in Section~\ref{sec:representations} that a temporal graph admits a half-extend-respecting ordering of its edges if and only if it is zero-acyclic.
We say that an ordering $E^{ord}$ of all the temporal edges is \emph{node-departure sorted} if all edges departing from the same node are ordered by non-decreasing departure time in $E^{ord}$, that is we have $e <_{E^{ord}} f$ whenever $e,f \in E$ have same tail and satisfy $dep(e) < dep(f)$. Similarly, we say that an ordering $E^{ord}$ of all the temporal edges is \emph{node-arrival sorted} if all edges arriving to the same node are ordered by non-decreasing arrival time in $E^{ord}$,  that is we have $e <_{E^{ord}} f$ whenever $e,f \in E$ have same head and satisfy $arr(e) < arr(f)$.

Finally, we define the \emph{doubly-sorted representation} of a temporal graph $(V,E,\alpha,\beta)$ as a data-structure with two lists $(E^{dep}, E^{arr})$, containing $|E|$ quadruples each, representing all temporal edges in $E$, where $E^{arr}$ is a node-arrival sorted list, and $E^{dep}$ is a node-departure sorted list. Moreover, we assume that we have implicit pointers between the two lists, that link each quadruple of one list to the quadruple representing the same temporal edge in the other list. We also say that $(E^{dep}, E^{arr})$ is \emph{half-extend-respecting} when $E^{arr}$ is additionally half-extend-respecting.

Without loss of generality, we can restrict our attention to nodes appearing as head or tail of at least one temporal edge and we thus assume $|V|=O(|E|)$. An algorithm is said to be linear in time and space when it runs in $O(|E|)$ time and uses $O(|E|)$ space.
Given a doubly sorted representation $(E^{dep}, E^{arr})$, we also assume that we are given for each node $v$ the list $E_v^{dep}$ of pointers to temporal edges with tail $v$ ordered by non-decreasing departure time, as it can be computed in linear time and space from $E^{dep}$ through bucket sorting. We assume that each list $E^{arr}$, $E^{dep}$, or $E_v^{dep}$ is stored in an array $T$ such that each element $T[i]$ can be accessed directly through its index $i\in [1,|T|]$ in constant time. Given two indexes $i\le j$, we also let $T[i:j]$ denote the sub-array  of elements of $T$ with index in $[i,j]$.

Finally, consider the case of  a temporal graph $G$ with strictly positive travel times. It is obviously zero-acyclic as it contains no edge with zero travel time. Clearly, if $E^{arr}$ (resp. $E^{dep}$) is an ordering of its edges by non-decreasing arrival time (resp. departure time), then $(E^{dep},E^{arr})$ is a doubly-sorted representation of $G$. Moreover, $E^{arr}$ is half-extend-respecting: whenever an edge $f=(v,w,\tau',\lambda')$ half-extends $e=(u,v,\tau,\lambda)$, the arrival time $a=\tau+\lambda$ of $e$ satisfies $a+\alpha_v\le \tau'$ and the arrival time $a'=\tau'+\lambda'$ of $f$ thus satisfies $a < a'$ as $\alpha_v\ge 0$ and $\lambda'>0$. Such a representation $(E^{dep},E^{arr})$ with fully sorted lists is called a \emph{fully doubly-sorted representation}.

\section{Warming up: a simple linear-time algorithm for reachability}
\label{sec:simplealg}

As a warm-up, we first provide a simple algorithm for solving in linear time and space the reachability problem when assuming positive travel times, which is defined as follows.

\optproblem{Singles-Source Reachability Problem}{Given a temporal graph  with waiting constraints $G=(V,E,\alpha,\beta)$ and a source node $s$, compute the set of all temporal edges that are $s$-reachable.}

Notice that this problem generalises the single-source earliest arrival time problem. Indeed, given the set of the $s$-reachable edges it is sufficient to perform a linear scan of such set to identify for each node $v$ the $s$-reachable edge with head $v$ that has lowest arrival time, and which corresponds to the earliest arrival time at $v$. 

In this section, we assume to be given a fully doubly-sorted representation $(E^{dep},E^{arr})$ of a temporal graph with positive travel times. We design an algorithm which mainly consists in scanning linearly edges in $E^{arr}$ while updating the set $A_v$ of $s$-reachable edges terminating $sv$-walks in the temporal graph resulting from the edges read so far. To help identifying edges that will appear in such walks in next iterations, we also mark edges that extend these walks.

We now describe more precisely how edges are scanned and marked as formalized in Algorithm~\ref{alg:eat}. We first build the lists $E_v^{dep}$ of temporal edges with tail $v$ by bucket sorting $E^{dep}$ at Line~\ref{line:prebeg}. We then identify the $s$-reachable edges as follows. We linearly scan $E^{arr}$. In the temporal graph resulting from the temporal edges read up to edge $e=(u,v,\tau,\lambda)\in E^{arr}$, the only walks from $s$ that have not been considered yet must contain $e$, and must have it as last edge as $E^{arr}$ is sorted by non-decreasing arrival time.
If its tail $u$ is $s$, or if $e$ is marked, then we know that there exists a walk from $s$ to its head $v$. In that case, we add edge $e$ to $A_v$ at Line~\ref{line:add}, and we then mark edges that extend $e$, that is edges in $E_v^{dep}$ with departure time in $[a+\alpha_v,a+\beta_v]$, since the arrival time of $e$ is $a=\tau+\lambda$. These edges appear consecutively in $E_v^{dep}$ which is processed linearly as walks from $s$ to $v$ are identified. This process is done in Lines~\ref{line:procbeg}-\ref{line:procend} in Algorithm~\ref{alg:eat}, starting from the index $p_v$ of the last processed edge in $E_v^{dep}$, and such edges $f$ that extend $e$ are marked at Line~\ref{line:mark} before updating $p_v$.
Moreover, we use classical parent pointers to be able to compute an $sv$-walk for each $s$-reachable edge with head $v$. 
Each parent pointer $P[f]$ of an edge $f$ is initially set to a null value $\bottom$ at Line~\ref{line:pp}. Whenever we mark edge $f$, that extend the currently scanned edge $e$, we set the parent pointer of $f$ to $e$.
If $f$ is an $s$-reachable edge at $v$, we can then get an $sv$-walk by following the parent pointer $P[f],P[P[f]],\dots$.

\begin{algorithm}[t]
%\SetAlgoLined
\DontPrintSemicolon %otherise print the semicolon in \lIf
\Input{A doubly-sorted representation $(E^{arr},E^{dep})$ of a temporal graph $G$ with waiting constraints $(\alpha,\beta)$, and a source node $s\in V$.}
\Output{The sets $(A_v)_{v\in V}$ of $s$-reachable edges at each node $v$ sorted by non-decreasing arrival time.}

For each node $v$, generate the list $E^{dep}_v$ by bucket sorting $E^{dep}$. \label{line:prebeg}\\
\For{each node $v$}{
  Set $A_v:=\emptyset$. \Comment{Set of $s$-reachable edges (as a sorted list).}
  Set $p_v:=0$. \Comment{Index of the last processed edge in $E_v^{dep}$.} 
}

% For each node $v$, initialize the arrival time Set $A_v:=\emptyset$.\\
Set all the edges in $E^{arr}$ as unmarked.\\
Set $P[e]:=\bottom$ for each edge $e\in E^{arr}$.\label{line:pp}\label{line:preend}\Comment{Parent of $e$, initially null.}
%\Comment{Scan edges in $E^{arr}$:}
\For{each edge $e = (u,v,\tau,\lambda)$ in $E^{arr}$% by non-decreasing arrival time
}{
    \If{$u=s$ or $e$ is marked}{\label{line:if}\Comment{$e$ is $s$-reachable.}
        $A_v:=A_v\cup\{e\}$\\\label{line:add}%\Comment{Add $a$ to $A_v$.}
        Let $a=\tau + \lambda$ be the arrival time of $e$.\label{line:procbeg}\\
        \Comment{Process further edges from $v$ until dep.\,time $\ge a+\beta_v$:}
        Let $l>p_v$ be the first index of an edge $(v,w,\tau',\lambda') \in E_v^{dep}$ such that $\tau' \geq a + \alpha_v$ (set $l:=|E_v^{dep}| +1$ if no such index exists).\label{line:skip}\label{line:left}\\
        %\Comment{Skip edges from $v$ with dep.\,time $< a+\alpha_v$:}
        %Let $q\ge p_v$ be the last index of an edge $(v,w,\tau',\lambda') \in E_v^{dep}$ such that $\tau' < a + \alpha_v$ (set $q:=p_v$ if no such index exists).\label{line:skip}\\
        %\Comment{Process further edges from $v$ with dep.\,time  $\le a+\beta_v$:}
        %Set $l:=q+1$.\label{line:left}\lv{We could directly use $l$ instead of $q$: Let $l>p_v$ be the first index of an edge...}\\
        Let $r \ge l$ be the last index of an edge $(v,w,\tau',\lambda') \in E_v^{dep}$ such that $\tau'\leq a + \beta_v$ (set $r:=l-1$ if no such index exists).\label{line:right}\\
        \Comment{Mark unmarked edges with dep.\,time in $[a + \alpha_v,a+\beta_v]$:}
        \lIf{$l\le r$}{mark each edge $f\in E_v^{dep}[l:r]$ and set $P[f]:=e$.}\label{line:mark} 
        Set $p_v:=r$.\label{line:rvupdate}\label{line:procend}\\
    }
}

\Return{the sets $(A_v)_{v\in V}$.}

\caption{Computing, for each node $v$, the set $A_v$ of all $s$-reachable edges with head $v$.}
\label{alg:eat}
\end{algorithm}

\begin{theorem}
    Given a fully doubly-sorted representation of a temporal graph with waiting constraints $G=(V,E,\alpha,\beta)$ having positive travel times, and a source node $s\in V$, Algorithm~\ref{alg:eat} computes all $s$-reachable temporal edges in linear time and space.
\end{theorem}

\begin{proof}
\lightparagraph{Correctness.}
Let us denote by %$E_k$ the set of edges in $E^{arr}[1:k]$, and with
$G_k=(V,E^{arr}[1:k],\alpha,\beta)$ the temporal graph induced by the first $k$ temporal edges in $E^{arr}$.
We will prove, by induction on $k$, the following two invariants: 

\begin{description}
    \item [($I^1_k$)] For every node $v$, $A_v$ contains all $s$-reachable edges with head $v$ in $G_k$.
    \item [($I^2_k$)] The marked edges are all the edges in $E$ that extend a walk from $s$ in $G_k$.
\end{description}

The correctness of the algorithm will follow from invariant $(I^1_k)$ for $k=|E|$. The invariants are satisfied for $k=0$ since there are no edges in $G_0$ while the sets $(A_v)_{v\in V}$ of $s$-reachable edges are initially empty and no edge is initially marked.

Now suppose that the two invariants hold for $k-1$, with $k\geq 1$, and let us prove that they still hold for $k$ after scanning the $k$th edge $e_k=(u,v,\tau,\lambda)$ in $E^{arr}$. To prove $(I^1_k)$ and $(I^2_k)$, we first show that the condition of the if statement at Line~\ref{line:if} is met when $e_k$ is an $s$-reachable edge in $G_k$. It is obviously the case when $u=s$ as $\langle e_k\rangle$ is in $G_k$, or when $e_k$ was previously marked, as Invariant $(I^2_{k-1})$ then implies that it extends a walk $Q$ from $s$ in $G_{k-1}$ and that $Q.e_k$ is a walk in $G_k$. The converse also holds: if $e_k$ is an edge of a walk $Q$ from $s$ in $G_k$, then either it is the first edge and we have $u=s$ or the sequence $Q'$ of edges before $e_k$ in $Q$ is a walk in $G_{k-1}$  and $(I^2_{k-1})$ implies that it is marked. %\fb{hard to rephrase as an iff., it becomes heavy}

Note that when $e_k$ appears in a walk $Q$ of $G_k$, it must be the last edge of $Q$ as $E^{arr}$ is sorted by non-decreasing arrival time and edges have positive travel time. This allows to prove $(I^1_k)$: as we assume $(I^1_{k-1})$, we just have to consider walks from $s$ that are in $G_k$ but not in $G_{k-1}$, that is those containing $e_k$. Since all these walks have $e_k$ as last edge, and $e_k$ is the only edge added to $A_v$ when such walks exist, we can conclude that $(I^1_k)$ holds.

Similarly, to prove $(I^2_k)$ when $(I^2_{k-1})$ holds, we just have to consider the edges extending a walk $Q$ from $s$ which is in $G_k$ but not in $G_{k-1}$. As discussed above, when such a walk $Q$ exists, $e_k$ is its last edge and the condition of the if statement Line~\ref{line:if} holds. Edges extending such a walk $Q$ are thus those extending $e_k$, that is all edges $f\in E_v^{dep}$ such that $a+\alpha_v\le dep(f) \le a+\beta_v$. Note that the ordering of $E_v^{dep}$ implies that these edges are consecutive in $E_v^{dep}$. If no such edges exist, let $l'$ and $r'$ designate the first and last indexes respectively where they are placed in $E_v^{dep}$. To prove $(I^2_k)$, it thus suffices to prove that all edges in $E_v^{dep}[l':r']$ are marked after scanning $e_k$ and that only edges in $E_v^{dep}[l':r']$ are marked during the iteration for $e_k$ (if no such edges exist we prove that we mark no edges).
Consider the values $l$ and $r$ computed at Lines~\ref{line:skip} and~\ref{line:right} respectively. If no edge $f$ extends $e_k$, then we get $r=l-1$ and no edge is marked. Now, we assume that such edges exist and that $l'$ and $r'$ are well defined. 
First assume $l\le r$ and thus that $l$ was not set to $|E_v^{dep}|+1$. The choice of $l,r$ then imply $a+\alpha_v\le dep(E_v^{dep}[l])$ and $dep(E_v^{dep}[r])\le a +\beta_v$. We thus have $l'\le l\le r\le r'$ and all marked edges at Line~\ref{line:mark} are in $E_v^{dep}[l':r']$. Moreover, the choice of $r$ indeed then implies $r=r'$. We still need to prove that edges in $E_v^{dep}[l':l-1]$ have already been marked.
Otherwise, when $r=l-1$, no edge is marked. This occurs when $p_v\ge r'$ and we then have $l=p_v+1$. % and $r=p_v$.
In both cases, it remains to prove that all edges $f\in E_v^{dep}[l':\min\{l-1,r'\}]$ have already been marked. This interval is non empty when $l'\le l-1$ and thus $p_v=l-1$ by the choice of $l$. We thus have $p_v\ge \min\{l-1,r'\}$. Let $i$ be the index of $f$ in $E_v^{dep}$ and consider the iteration $j<k$ when $p_v$ was updated form a value smaller than $i$ to a value $r''\ge i$ where $l''$ and $r''$ denote the indexes computed for variables $l$ and $r$ respectively during the $j$-th iteration for edge $e_j\in E^{arr}$. 
Since $E^{arr}$ is sorted by non-decreasing arrival time, the arrival time $a'$ of $e_j$ satisfies $a'\le a$ and we thus have $dep(f)\ge a+\alpha_v\ge a'+\alpha_v$. The choice of index $l$ at Line~\ref{line:skip} in that iteration thus guarantees that the index $l''$ must satisfy $l''\le i$. We thus have $l''\le i\le r''$ and $f$ was marked at Line~\ref{line:mark} during the $j$th iteration. This completes the proof of $(I^2_k)$.
% Moreover, the choice of $r$ also implies $r'\le r$, even when $r=l-1$\fb{ when $r=l-1$ what does it mean to say $r'\le r$? who is $r'$?}. Thus, we just have to prove that any edge $f=(v,w,\tau',\lambda')\in E_v^{dep}[l':q]$ was marked previously. When $q>p_v$, the edge $E_v^{dep}[q]$ has departure time $\tau'<a+\alpha_v$ and the ordering of $E_v^{dep}$ implies $l'> q$, which means that the sub-array $E_v^{dep}[l':q]$ is empty. Let us now consider the case $q=p_v$ and the iteration $j<k$ when $f$ was processed, that is when the variable $p_v$ was first updated to a value greater or equal to the index $i$ of $f$ in $E_v^{dep}$. While scanning the $j$th edge $e_j\in E^{arr}$, we have updated $p_v$ from a value $p_v<i$ to a value $r''\ge i$. Since $E^{arr}$ is sorted by non-decreasing arrival time, the arrival time $a'$ of $e_j$ satisfies $a'\le a$ and we thus have $\tau'\ge a+\alpha_v\ge a'+\alpha_v$. The choice of index $q$ at Line~\ref{line:skip} in this iteration guarantees that the index $l''$ set in variable $l$ at Line~\ref{line:left} must satisfy $l''\le i$. We thus have $l''\le i\le r''$ and $f$ was marked at Line~\ref{line:mark} during the $j$th iteration. This completes the proof of $(I^2_k)$.

%\lv{The following paragraph, seems to over-explain, we have to reduce to the minimum argumentation.}

We finally prove that the parent pointers allow us to compute for each $s$-reachable edge $f=(u,v,\tau,\lambda)$ with head $v$ an $sv$-walk  ending with $f$. If $f\in A_v$ and it is not marked, then $P[f]=\bottom$, and we must have $u=s$ as $f$ was added to $A_v$. In this case, $\langle f \rangle$ is an $sv$-walk itself. Now consider the case $f\in A_v$ and $f$ is marked. Consider the iteration $k$ where $f$ was marked. By $(I^2_{k-1})$ and $(I^2_k)$, $f$ extends a walk from $s$ ending with $e_k$, where $e_k$ is the edge scanned at iteration $k$, and $P[f]$ was then set to $e_k$.
% Thus, let us consider the case in which we are interested in an $sv$-walk that ends with $f$ and $tail(f)\neq s$. As we proved, the marked edges are exactly the edges that extends a walk from $s$, and thus they are also the last edge of the walk from $s$ that have at least two edges.
% If an edge $f$ is marked at $k$th iteration while edge $e_k$ is being scanned, this means that there exists a walk from $s$ terminating with $e_k$ that can be extended by edge $f$. Indeed, the parent pointer of $f$ will be set to $e_k$. \lv{I don't follow here, we just want to argue that $P[f]$ is the last edge of an $su$-walk that $f$ extends.}
This guarantees by a simple induction that, if $P[f]\neq \bottom$, by following the parent pointers in classical manner, namely $P[f], P[P[f]], \dots$, until $\bottom$ is found, it is possible to obtain a walk terminating with edge $f$. 

\lightparagraph{Complexity analysis.}
The preprocessing of $E^{dep}$ and the initialization from Line~\ref{line:prebeg} to Line~\ref{line:preend} clearly take linear time. The main for loop scans each temporal edge $e=(u,v,\tau,\lambda)$ in $E^{arr}$ exactly once. For each iteration there are three operations that may require non-constant time: the computation of $l$ and $r$ at Lines~\ref{line:skip} and~\ref{line:right}, and marking edges in $E^{dep}_v[l,r]$ at Line~\ref{line:mark}. They all take $O(r-p_v)$ time as $l$ and $r$ can be found by scanning edges in $E^{dep}_v$ from $p_v+1$.
%
%Note that adding $a$ to set $A_v$ at Line~\ref{line:add} can be done in constant time: since arrival times are added to $A_v$ in non-decreasing order, testing whether $a$ is already in $A_v$ amounts to compare $a$ to the last value added to $A_v$.
Thanks to the update of the index $p_v$ to $r$, each edge in $E^{dep}_v$ is processed at most once for a total amortized cost of $O(|E^{dep}_v|)$.  Overall, this leads to a time complexity of $O(|E| + \sum_{v \in V}|E^{dep}_v|)=O(|E|)$.
%Since $\sum_{v \in V}|E^{dep}_v| = |E|$, the total time complexity of the algorithm is $O(|E| + n)$.
Algorithm~\ref{alg:eat} thus runs in linear time.
%
%space complexity
Finally, let us notice that for all nodes $v$, the set $A_v$ has size bounded by the number of temporal edges with head $v$. We thus have $\sum_{v\in V}|A_v|\le |E|$,
% NO $|A_v| \leq |E^{dep}_v|$, for each node $v$. 
and the space complexity of Algorithm~\ref{alg:eat} is linear.
\end{proof}

\section{Single-source all-reachable-edge minimum-cost walks}
\label{sec:gal}

To solve the problem of computing minimum-cost walks from a single source $s$, we will consider a more general problem consisting in computing at each destination $v$, and for each possible $s$-reachable edge $e$ with head $v$, an $sv$-walk with minimum cost among all $sv$-walks ending with $e$. We first introduce an algebraic framework for associating costs to edges and walks.

\subsection{General cost structure for walks}

%We now introduce an algebraic framework for defining costs of edges and walks.
We integrate a temporal graph $G=(V,E,\alpha,\beta)$ with an algebraic \emph{cost structure} $({C},\gamma,\oplus,\preceq)$, where ${C}$ is the set of possible \emph{cost values}, $\gamma$ is a \emph{cost function} $\gamma:E\rightarrow{C}$, $\oplus$ is a \emph{cost combination function} $\oplus:{C}\times{C}\rightarrow{C}$, and $\preceq$ is a \emph{cost total order} $\preceq\ \subseteq\ {C}\times{C}$. We also define the relation $\prec$ between the elements of ${C}$ as $a \prec b$ if and only if $a \preceq b$ and $a \neq b$.
For any walk $Q=\langle e_1,\ldots,e_k\rangle$, the \emph{cost function} of $Q$ is recursively defined as follows: $\gamma_{Q} = \gamma_{\langle e_1,\ldots,e_{k-1}\rangle} \oplus\gamma(e_k)$, with $\gamma_{\langle e_1\rangle}=\gamma(e_1)$. In other words, the costs combine along the walk according to the cost combination function. 
The cost structure is supposed to satisfy the following \emph{right-isotonicity property}~\cite{BrunelliCV2021,Sobrinho2005,Griff2010} (\emph{isotonicity} for short): 
\begin{align}\tag{isotonicity}\label{eq:isotonicity}  
\mbox{For any } c_1,c_2,c\in C \mbox{ such that } c_1\preceq c_2, \mbox{ we have } c_1\oplus c\preceq c_2\oplus c. 
\end{align} 
This property guarantees that if several walks are extended by a given temporal edge $e$, then the best cost is obtained by extending the walk $Q^*$ with minimum cost: as for any other walk $Q$ we have $\gamma_{Q^*}\preceq \gamma_{Q}$, we get $\gamma_{Q^*.e}\preceq \gamma_{Q.e}$ by the isotonicity property and the cost function definition. However, a prefix of a minimum-cost walk is not necessarily a minimum-cost walk.
% Due to waiting-time constraints, Lemma~1 in \cite{BrunelliCV2021} does not hold.

%Given a temporal graph $G=(V,E,\alpha,\beta)$ with cost structure $({C},\gamma,\oplus,\preceq)$, and a source node $s\in V$, we let $A_v$ denote the set of all $s$-reachable edges with head $v$. 
We define the single-source all-reachable-edge minimum-cost problem as follows.

\optproblem{single-source all-reachable-edge minimum-cost problem}{Given a temporal graph $G=(V,E,\alpha,\beta)$ with cost structure $({C},\gamma,\oplus,\preceq)$, and a source node $s\in V$, compute for each destination $v\in V$ and each possible $s$-reachable edge $e$ with head $v$ the minimum cost of any $sv$-walk ending with edge $e$. 
}

Letting $A_v$ denote the set of all $s$-reachable edges with head $v$, it consists in computing for each node $v$ all pairs $(e,c)$ such that $e\in A_v$ and $c=\min\{\gamma_Q : Q \mbox{ is an $sv$-walk ending with edge $e$}\}$. We will denote with $A'_v$ the list of such pairs $(e,c)$ ordered by non-decreasing arrival time of the edges. In this section we consider this problem in the case of zero-acyclic temporal graphs. 

\subsection{Solving the single-source all-reachable-edge minimum-cost problem}

We can now state our main theorem about the complexity of the above problem.

\begin{theorem}\label{th:gal}
Given a half-extend-respecting doubly-sorted representation $(E^{dep},E^{arr})$ of a zero-acyclic temporal graph $G=(V,E,\alpha,\beta)$ with cost structure $({C},\gamma,\oplus,\preceq)$ satisfying isotonicity, and a source node $s$, the single-source all-reachable-edge minimum-cost problem can be solved in linear time and space.
\end{theorem}

We will see in Section~\ref{sec:representations} that the half-extend-respecting requirement can be dropped as a doubly-sorted representation $(E^{dep},E^{arr})$ such that $E^{arr}$ is half-extend-respecting can easily be computed in linear time and space from any doubly-sorted representation of a zero-acyclic graph.

\bigskip

To prove the above theorem, we can design an algorithm that scans linearly edges in $E^{arr}$. The general idea is to maintain for each unscanned edge $f$ the minimum cost of any walk $Q$ from $s$ in the partial temporal graph induced by edges scanned so far such that $f$ extends $Q$. We can update these costs each time a new edge $e\in E^{arr}$ is scanned relying on the property that the edges of any walk are scanned in order. In particular, the cost that has been associated to $e$ itself allows to infer easily the minimum cost of a walk from $s$ ending with $e$. 

However, we propose a more general algorithm that works on any doubly-sorted representation $(E^{dep},E^{arr})$ to enable the wider setting of Section~\ref{sec:zero}. The counterpart is that it considers only certain walks that are well ordered with respect to $E^{arr}$ and $E^{dep}$ in the following sense. We say that a walk $Q$ is \emph{$(E^{dep},E^{arr})$-respected} when for each pair $e$ and $f$ of consecutive edges in $Q$, and for each edge $e' \in E$ having same tail as $f$ and satisfying $f\leq_{E^{dep}} e'$, we have $e <_{E^{arr}}e'$. In particular, we get $e<_{E^{arr}} f$ for $e'=f$, and the edges of $Q$ must appear in order in $E^{arr}$. It also implies that edges $e'$ after $f$ in $E_v^{dep}$ will be scanned after $e$ by our algorithm, a property we will use for getting efficient updates. We will see later that all walks are $(E^{dep},E^{arr})$-respected when $E^{arr}$ is half-extend-respecting and $E^{dep}$ is node-departure sorted.

We now introduce two notions related to $(E^{dep},E^{arr})$-respected walks. An $(s,E^{dep},E^{arr})$-reachable edge is defined as an $s$-reachable edge that ends an $(E^{dep},E^{arr})$-respected walk from $s$.
Moreover, given a subset $E'$ of edges, and an edge $f\in E$ with tail $v$, we define its \emph{best extendable cost} with respect to $E'$ as the minimum cost of an $sv$-walk $Q$ that $f$ extends in the partial graph induced by $E'\cup\{f\}$ and such that $Q.f$ is $(E^{dep},E^{arr})$-respected.

Our algorithm maintains the best extendable costs of all edges with respect to the prefix of edges scanned so far as follows.
Initially, all best extendable costs are undefined as expressed by a special value $\bottom$. Then each time an edge $e\in E^{arr}$ is scanned, it is sufficient to update the costs of edges $f$ that extend $e$ and such that $\langle e,f\rangle$ is $(E^{dep},E^{arr})$-respected as detailed in the proof of Lemma~\ref{lem:kthiter}. The main difficulty is to perform this update in constant amortized time although a large number of edges $f$ may extend $e$. For that purpose, $E^{dep}$ is first bucket sorted according to tails, and edges from a node $v$ that extend the same walks from $s$ to $v$ are grouped into intervals of the array $E_v^{dep}$ of edges from $v$. These intervals are stored in a doubly linked list $\intervs{v}$ of quadruples where each interval $(l,r,c,e)\in \intervs{v}$ represents the association of edges in $E_v^{dep}[l:r]$ to best extendable cost $c$ and parent edge $e$ where $e$ is an edge they all extend and such that there exists an $(E^{dep},E^{arr})$-respected walk from $s$ having cost $c$ and ending with $e$. We also maintain the overall interval $(l_v,r_v)$ spanned by $\intervs{v}$. Note that this interval is considered to be empty when $r_v<l_v$.
%Note that this data-structure allows to update the cost associated to all edges of an interval in constant time.
Algorithm~\ref{alg:gal} describes how to update these intervals each time an edge $e$ with head $v$ is scanned. It relies on the fact that intervals of $\intervs{v}$ are consecutive in $E_v^{dep}$ and are also ordered by non-decreasing associated costs. When the scan of $E^{arr}$ has progressed sufficiently, the best extendable cost of some edges will not change anymore and it is then stored directly in an array $B$ through the procedure $\processcosts(v,j)$ which erases intervals of $\intervs{v}$ up to index $j$ and stores the cost associated to the corresponding edges in $B$ as detailed in Algorithm~\ref{alg:fixcosts}.
%The overall approach of our algorithm is to perform two main tasks during a single scan of $E^{arr}$. First, it identifies the edges that extend walks from $s$ in the partial temporal graph resulting from the edges scanned so far. For such an edge $f$ with tail $v$, we define its \emph{best extendable cost} as the minimum cost of an $sv$-walk that $f$ extends in this partial graph. The algorithm associates to such edges their \emph{best extendable cost} as more edges are scanned in $E^{arr}$ and more walks are considered in this growing partial graph. We will refer to this first task as \emph{update costs}. Second, it infers the minimum cost $c$ of walks from $s$ ending with the edge $e=(u,v,\tau,\lambda)$ currently scanned from the best extendable cost of $e$. This task will be called \emph{compute best cost}. 
At the end of the scan, our algorithm returns the lists $(A'_v)_{v\in V}$ %\fb{we are free to put $A_v$ just if we do not put the easy algorithm} 
which contain the minimum cost associated to all possible $(s,E^{dep},E^{arr})$-reachable edges. During the execution, we build parent pointers, that allow to represent, for each such edge $e\in A'_v$, an $sv$-walk with minimum cost ending with edge $e$ by associating to each edge $f$ the edge $P[f]$ preceding it in such a walk.
%Algorithm~\ref{alg:gal} gives a detailed description of the algorithm and Algorithm~\ref{alg:fixcosts} describes the  procedure $\processcosts(v,j)$ which stores in $B[f]$ the best extendable cost of edges $f$ with tail $v$ and index at most $j$ in $E_v^{dep}$ when we have progressed sufficiently in the scanning of $E^{arr}$ to ensure that we have considered all walks that such edges $f$ extend in $G$.

%new version with index in ProcessCosts:
\begin{algorithm}[ht]
%\SetAlgoLined
\DontPrintSemicolon %otherise print the semicolon in \lIf
\Input{A doubly-sorted representation $(E^{dep},E^{arr})$ of a temporal graph $G$ with waiting-time constraints $(\alpha,\beta)$ and cost structure $({C}, \gamma,\oplus,\preceq)$ satisfying isotonocity, and a source node $s$.}
\Output{Minimum cost of an $(E^{dep},E^{arr})$-respected $sv$-walk for each node $v$ and for each $(s,E^{dep},E^{arr})$-reachable edge $e$ with head $v$.}

For each node $v$, generate the list $E^{dep}_v$ by bucket sorting $E^{dep}$. \label{ling:prebegpar}\\
\For{each node $v$}{
  Set $A'_v:=\emptyset$. \Comment{List of pairs of $s$-reachable edge and cost.}
%  Set $p_v:=0$. \Comment{Index of the last processed edge in $E_v^{dep}$.}
%  Set $r_v:=0$. \Comment{Index of the last right bound of an interval in $\intervs{v}$.}
  Set $\intervs{v}:=\emptyset$. \Comment{Doubly linked list of consecutive intervals of $E_v^{dep}$.}
  Set $(l_v,r_v):=(1,0)$. \Comment{Overall interval of $E_v^{dep}$ spanned by $\intervs{v}$.}
}
%Set all the edges in $E^{arr}$ as unmarked.\\
Set best extendable cost $B[e]:=\bottom$ and parent pointer $P[e]:=\bottom$ for each edge $e\in E^{arr}$.\label{ling:preendpar}\\

\For{each edge $e = (u,v,\tau,\lambda)$ in $E^{arr}$\label{ling:begoutfor}}{
    Let $i$ be the index of $e$ in $E_u^{dep}$.\label{ling:begbest}\\
    $\processcosts(u,i)$ \label{ling:callonu} \Comment{Obtain $B[e]$ in particular.}
    \If{$u=s$ or $B[e]\neq \bottom$\label{ling:edgecondition}}{
        \Comment{Get the minimum cost $c$ of a walk having $e$ as last edge:} 
        \lIf{$u=s$ and ($B[e]=\bottom$ or $\gamma(e) \prec B[e]\oplus \gamma(e)$)\label{ling:costsource}}{$c:=\gamma(e)$ and $P[e]:=e$}
        \lElse{$c:=B[e]\oplus\gamma(e)$.}\label{ling:cost}
        Append $(e,c)$ to $A'_v$.\label{ling:add}\\
        \Comment{Find the interval $(l,r)$ of edges in $E_v^{dep}$ that extend $e$:}
        Let $a=\tau + \lambda$ be the arrival time of $e$.\label{ling:begupdate}\\
        Let $l \geq l_v$ be the first index of an edge $(v,w',\tau',\lambda') \in E_v^{dep}$ such that $\tau' \ge a + \alpha_v$ (set $l:=|E_v^{dep}|+1$ if no such index exists).\label{ling:left}\\
        Let $r \geq r_v$ be the last index of an edge $(v,w,\tau',\lambda') \in E_v^{dep}$ such that $\tau'\leq a + \beta_v$ (set $r:=r_v$ if no such index exists).\label{ling:setrvprime}\label{ling:right}\\
        %Let $j \geq l_v$ be the last index of an edge $(v,w',\tau',\lambda') \in E_v^{dep}$ such that $\tau' < a + \alpha_v$ (set $j:=l_v-1$ if no such index exists).\\
        \Comment{Remove from $\intervs{v}$ intervals preceding $l$ and set $l_v:=l$:%of edges with dep.\,time $< a+\alpha_v$:
        }
        $\processcosts(v,l-1)$\label{ling:callonv}\label{ling:begprocess}\\
        \Comment{Remove from $\intervs{v}$ intervals with cost greater than $c$:}
        Set $l_c:=\max\{l,r_v+1\}$. \Comment{First index in $(l,r)$ after $\intervs{v}$.} \label{ling:setl}
        \While{$\intervs{v}\not=\emptyset$ has last interval $I'=(l',r',c',e')$ satisfying $c\prec c'$ \label{ling:whileinterv}}{
          Remove $I'$ from $\intervs{v}$ and update $l_c:=l'$.\label{ling:delinterv}
        }
        \Comment{Associate cost $c$ and parent $e$ to edges in $E_v^{dep}[l_c:r]$:}
        \lIf{$l_c \le r$}{
          append interval $I=(l_c,r,c,e)$ to $\intervs{v}$. \label{ling:addinterv}% \Comment{%$I$ Associate $c$ to $E_v^{dep}[l_c:r]$.}
        }
        Set $r_v := r$.\label{ling:endupdate}
        %Set $p_v := r$.\label{ling:setrv}\label{ling:endprocess}\\
    }
    \label{ling:endoutfor}
}
\Return{the lists $(A'_{v})_{v\in V}$.}
\caption{Computing, for each node $v$ and each $(s,E^{dep},E^{arr})$-reachable edge $e$ with head $v$, the minimum cost of any $(E^{dep},E^{arr})$-respected $sv$-walk ending with $e$.}
\label{alg:gal}
\end{algorithm}

\begin{algorithm}[ht]
\Procedure{$\processcosts(v,j)$}{
  \While{the first interval $I=(l,r,c,e)$ in $\intervs{v}$ satisfies $l \leq j$ \label{linp:beginterv}}{
    Let $l' = \min \{r,j\}$.\\
    \lFor{each edge $f=(v,w,\tau,\lambda)$ in $E_v^{dep}[l:l']$}{$B[f]:=c$ and $P[f]:=e$.} \label{linp:setB}
    \leIf{$l'=r$}{remove $I$ from $\intervs{v}$}{update $I:=(j+1,r,c,e)$.} \label{linp:left-check}\label{linp:endinterv}
  }
  %\lIf{$\intervs{v}=\emptyset$}{
    Set $l_v:=j+1$. \label{linp:pv-check}
  %}
}
\caption{Set the best extendable cost and parent of edges in $E_v^{dep}[l_v:j]$ and remove corresponding intervals from $\intervs{v}$.
%from $v$ from the $p_v$-th edge to the $i$-th.%Set further best extendable costs of edges from $v$ with departing time less than $a+\alpha_v$.}
}
\label{alg:fixcosts}
\end{algorithm}

Let us denote by $G_k=(V,E^{arr}[1:k],\alpha,\beta)$ the temporal graph induced by the first $k$ temporal edges in $E^{arr}$. We now describe more precisely how the algorithm proceeds when the $k$th edge $e_k=E^{arr}[k]=(u,v,\tau,\lambda)$ is scanned.

It first evaluates the minimum cost of an $(E^{dep},E^{arr})$-respected $sv$-walk ending with $e_k$ at Lines~\ref{ling:begbest}-\ref{ling:add} of Algorithm~\ref{alg:gal}. The call to $\processcosts(u,i)$ at Line~\ref{ling:callonu} allows to definitely set the best extendable cost associated to edges up to $e_k$ in $E_u^{dep}$ and store this value for $e_k$ in $B[e_k]$ with an associated parent edge $P[e_k]$. We can thus test at Line~\ref{ling:edgecondition} if $e_k$ is an $(s,E^{dep},E^{arr})$-reachable edge: it is the last edge of an $(E^{dep},E^{arr})$-respected $sv$-walk when $u=s$ or $e_k$ has best extendable cost $B[e_k]$ different from $\bottom$. In the positive case, we generally obtain the minimum cost $c$ of such a walk by combining $B[e_k]$ with $\gamma(e_k)$ through the $\oplus$ combination function of the cost structure as $c=B[e_k]\oplus\gamma(e_k)$. However, edges from the source have to be handled with special care: if $u=s$ then the $sv$-walk $\langle e_k\rangle$ which has cost $\gamma_{\langle e_k\rangle}=\gamma(e_k)$ must also be taken into account. Note that a single edge walk is always $(E^{dep},E^{arr})$-respected. Lines~\ref{ling:costsource}-\ref{ling:cost} set $c$ accordingly and update $A'_v$ at Line\ref{ling:add}. When the minimum cost $c$ is obtained for walk $\langle e_k\rangle$, we set the parent pointer $P[e_k]:=e_k$ allowing to detect that $e_k$ is the first edge of the walk.

Once the minimum cost $c$ is computed, we update the best extendable cost of edges $f$ extending $e_k$ and such that $\langle e_k,f\rangle$ is $(E^{dep},E^{arr})$-respected at Lines~\ref{ling:begupdate}-\ref{ling:endupdate}.
The reason for considering only these edges is that the $(E^{dep},E^{arr})$-respected walks that are in $G_k$ and not in $G_{k-1}$ must contain edge $e_k$. Moreover $e_k$ must be the last edge of such a walk $Q$ as its edges must appear in order in $E^{arr}$ when $Q$ is $(E^{dep},E^{arr})$-respected.
%It consists in maintaining for each edge $f=(v,w,\tau',\lambda')$ its best extendable cost which is the minimum cost of an $sv$-walk $Q$ in $G_k$ such that $f$ extends $Q$. Edges from $v$ extending the same walks are grouped into consecutive intervals of $E_v^{dep}$ which are stored in a list $\intervs{v}$.
The general idea is to update $\intervs{v}$ so that its intervals span exactly these edges. The edges $f$ extending $e_k$ must have departure time within $a+\alpha_v$ and $a+\beta_v$ where $a=arr(e_k)$ is the arrival time of $e_k$. As $E_v^{dep}$ is node-departure sorted, they indeed correspond to a sub-array $E_v^{dep}[l:r]$ of edges where $l$ and $r$ are the indexes computed respectively at Lines~\ref{ling:left} and~\ref{ling:right}. Recall that $(l_v,r_v)$ is the overall interval spanned by $\intervs{v}$ after the last update at $v$. The reason for forcing $l\ge l_v$ at Line~\ref{ling:left} is that we consider only edges $f$ such that $\langle e_k,f\rangle$ is $(E^{dep},E^{arr})$-respected as detailed in the proof of Lemma~\ref{lem:kthiter}.
%
%This data-structure allows to update in constant time the best extendable cost of all the edges in an interval. We also store in $(l_v,r_v)$ the overall interval spanned by $\intervs{v}$. Note that values $(l_v,r_v)$ such that $r_v<l_v$ indicate that $\intervs{v}$ is empty. We update $\intervs{v}$ and $(l_v,r_v)$ each time the scanned edge $e_k$ has head $v$ and we have obtained the minimum cost $c$ of any $sv$-walk ending with $e_k$. We then associate this cost $c$ to the appropriate interval of edges in $E_v^{dep}$ that extend $e_k$, that is those with departure time in $[a+\alpha_v,a+\beta_v]$ where $a=arr(e_k)$. The ordering of $E_v^{dep}$ by non-decreasing departure time ensures that these edges are indeed consecutive in $E_v^{dep}$. More precisely they are in $E_v^{dep}[l:r]$ where $l$ and $r$ are the indexes computed respectively at Lines~\ref{ling:left} and~\ref{ling:right} relying on the assumptions $l\ge l_v$ and $r\ge r_v$ which can be deduced from the fact that $E^{arr}$ is node-arrival sorted and half-extend-respecting as shown in the proof of Lemma~\ref{lem:cor}.
Some edges may already belong to some previously constructed intervals when $l_v<l$. We first remove intervals of edges $f$ with index less than $l$ as they do not extend $e_k$ or $\langle e_k,f\rangle$ is not $(E^{dep},E^{arr})$-respected. This is done through the call to $\processcosts(v,l-1)$ at Line~\ref{ling:begprocess}. Note that $l_v=l$ after this call. All edges in $E_v^{dep}[l_v:r_v]$ now extend $e_k$ and belong to some interval of $\intervs{v}$. The remaining edges extending $e_k$ are thus in $E_v^{dep}[l_c:r]$ where $l_c=\max\{l,r_v+1\}$ is set at Line~\ref{ling:setl} and we aim at creating an interval $(l_c,r,c,e_k)$ for associating these edges to cost $c$ and parent $e_k$.
However, we first remove intervals associated to a cost $c'$ greater than $c$ at Lines~\ref{ling:whileinterv}-\ref{ling:delinterv}. The reason is that a better cost is obtained by extending $e_k$ for them.
For that, we use the key property that intervals in $\intervs{v}$ are all consecutive and their associated costs are non-decreasing. The intervals with cost greater than $c$ are thus at the end of $\intervs{v}$. While removing such intervals, the left bound $l_c$ is updated to include the corresponding edges in the interval $(l_c,r)$ of $E_v^{dep}$.
%The list $\intervs{v}$ of intervals is stored in a doubly linked list so that it is possible to add and remove an interval both at the beginning or at the end of $\intervs{v}$ in constant time.\lv{move to complexity analysis?}
Finally, the edges in $E_v^{dep}[l_c,r]$ are associated to cost $c$ and parent $e_k$ at Lines~\ref{ling:addinterv}-\ref{ling:endupdate} by adding interval $(l_c,r,c,e_k)$ to $\intervs{v}$ and updating $r_v$ accordingly. As all intervals with cost greater than $c$ have been removed, we maintain the fact $\intervs{v}$ is sorted by non-decreasing costs. The computation of $l_c$ also ensures that all intervals in $\intervs{v}$ remain consecutive. Note that we have $(l_v,r_v)=(l,r)$ at the end of the iteration.

Finally, for each node $v$, and each possible $(s,E^{dep},E^{arr})$-reachable edge $e$ with head $v$, the parent pointers provide a representation of an $(E^{dep},E^{arr})$-respected $sv$-walk $Q$ ending with $e$ and having minimum-cost. The walk $Q$ can be retrieved in $O(|Q|)$ time by computing $P[e],P[P[e]],\ldots$ until reaching the first edge $f$ such that $P[f]=f$.

%The time complexity of Algorithm~\ref{alg:gal} is linear. The main reasons are that each edge $e \in E^{arr}$ is scanned once and each edge $f\in E_v^{dep}$ is scanned at most twice: once for adding it in an interval and once when setting its best extendable cost; furthermore, for a given lists $E_v^{dep}$, we construct at most $|E_v^{dep}|$ intervals which are added to $\intervs{v}$ once and later removed only once. 

\bigskip

The correctness of the algorithm mainly follows from the following lemma.

\begin{lemma}\label{lem:kthiter}
  After the $k$th iteration of Algorithm~\ref{alg:gal}, if an edge $f$ with tail $v$ is associated to cost $c$, either through an interval $(l,r,c,e)\in\intervs{v}$ containing the index of $f$ in $E_v^{dep}$ or by the value $c=B[f]$ when $B[f]\not=\bottom$, then $c$ is the best extendable cost of $f$ with respect to the subset $E^{arr}[1:k]$ inducing graph $G_k$.
  Moreover, the edge $e_k$ scanned at the $k$th iteration gets associated to cost $c$ in $A'_v$ if and only if it is an $(s,E^{dep},E^{arr})$-reachable edge and $c$ is the minimum cost of any $(E^{dep},E^{arr})$-respected $sv$-walk ending with $e_k$.
\end{lemma}

Note that exactly one of the three following cases occurs: $f$ has no associated cost, or $f$ is in an interval of $\intervs{v}$, or we have $B[f]\not=\bottom$. This is due to the way we maintain the value $l_v$ which is always the leftmost bound of an interval of $\intervs{v}$: once a value $B[f]$ is set for an edge $f$ with tail $v$ by a call to $\processcosts(v,j)$, $l_v$ is updated to a value greater than $j$, and $f$ cannot appear in an interval of $\intervs{v}$ anymore.

\begin{proof} 
  We prove the statement by induction on $k$. As there are no edges and no walks in $G_0$ and no edge has initially an associated cost, the statement holds for $k=0$. Assume that the statement holds for $k-1$ and let us prove it for $k$. Recall that the best extendable cost of an edge $f$ with respect to $E^{arr}[1:k]$ is the minimum cost of an $sv$-walk $Q$ in $G_k$ that $f$ extends and such that $Q.f$ is $(E^{dep},E^{arr})$-respected (where $v$ denotes the tail of $f$). By the induction hypothesis, we must update the cost associated to $f$ only when there is such a walk $Q$ with cost $c$ which is in $G_k$ but not in $G_{k-1}$ and such that $c$ is lower than the cost associated to $f$ after the previous iteration. This can occur only when $Q$ contains the edge $e_k=E^{arr}[k]=(u,v,\tau,\lambda)$ which is scanned at the $k$th iteration. Moreover, as $Q.f$ is $(E^{dep},E^{arr})$-respected, so is $Q$. Its edges thus appear in order in $E^{arr}$ and $e_k$ must be its last edge. It is thus sufficient to consider the minimum cost $c^*$ of an $(E^{dep},E^{arr})$-respected $sv$-walk $Q$ ending with $e_k$ and compare it the cost associated to edges $f$ that extend $e_k$ and such that $Q.f$ is $(E^{dep},E^{arr})$-respected. 

  We first show that the value $c$ computed at Lines~\ref{ling:costsource}-\ref{ling:cost} is indeed $c^*$. Consider an $(E^{dep},E^{arr})$-respected $sv$-walk $Q$ ending with $e_k$ and having cost $c^*$. In the case where $Q$ has at least two edges, let $Q'$ denote the prefix of $Q$ excluding $e_k$. The induction hypothesis and the call to $\processcosts(u,i)$ at Line~\ref{ling:callonu} then ensure that $B[e_k]$ is set to the minimum cost of such a walk $Q'$ that $e_k$ extends and such that $Q'.e_k$ is $(E^{dep},E^{arr})$-respected. By isotonicity, this implies that $B[e_k]\oplus \gamma(e_k)$ is the minimum cost of an $(E^{dep},E^{arr})$-respected $sv$-walk ending with $e_k$ and having at least two edges. In the case where $Q$ has one edge, we have $Q=\langle e_k\rangle$, and $e_k$ must be an edge from $s$ and the cost of $Q$ is $\gamma(e_k)$. In both cases, the test at Line~\ref{ling:edgecondition} passes when $e_k$ is an $(s,E^{dep},E^{arr})$-reachable edge and the computation of $c$ at Lines~\ref{ling:costsource}-\ref{ling:cost} sets $c$ to the minimum of $\gamma(e_k)$ and $B[e_k]\oplus \gamma(e_k)$ when both cases occur, ensuring that $c=c^*$ is the minimum cost of any $(E^{dep},E^{arr})$-respected $sv$-walk ending with $e_k$. Moreover, Line~\ref{ling:add} then ensures that $e_k$ gets associated to cost $c=c^*$ in $A'_v$.

  We now show that the set $F$ of edges $f$ that extend $e_k$ and such that $Q.f$ is $(E^{dep},E^{arr})$-respected is precisely $E_v^{dep}[l:r]$ where $(l,r)$ are the values computed at Lines~\ref{ling:left}-\ref{ling:right}. As these edges extend $e_k$, their departure time lies within $arr(e_k)+\alpha_v$ and $arr(e_k)+\beta_v$ which correspond to an interval $(l',r')$ of $E_v^{dep}$ as $E_v^{dep}$ is node-departure sorted. We further restrict our attention to those edges $f$ such that $Q.f$ is $(E^{dep},E^{arr})$-respected or equivalently $\langle e_k,f\rangle$ is $(E^{dep},E^{arr})$-respected when $Q$ is assumed to be $(E^{dep},E^{arr})$-respected. That is we should consider only edges $f$ so that no edge $e'$ with tail $v$ satisfies both $f\leq_{E^{dep}} e'$ and $e' \leq_{E^{arr}} e_k$. Let $l''$ be the highest index in $E_v^{dep}$ of such an edge $e'$.
% We thus need to show $(l,r)=(\max\{l',l''\}, r')$.
  The call to $\processcosts(u',i')$ at Line~\ref{ling:callonu}, where $u'=v$ is the tail of such an edge $e'$ scanned before $e_k$ and $i'\le l''$ is the index of $e'$ in $E_v^{dep}$, ensures that $l_v$ is at least $i'+1$. When $u=v$, $e_k$ is itself such an edge $e'$, and the call to $\processcosts(u,i)$ at Line~\ref{ling:callonu} where $i\le l''$ is the index of $e_k$ ensures that $l_v$ is at least $i+1$. We thus have $l_v\ge l''+1$. Note also that $l_v$ was updated to $l''+1$ at most in these calls from Line~\ref{ling:callonu}. Moreover, each call to $\processcosts(v,j)$ at Line~\ref{ling:callonv} for an edge $e'<_{E^{arr}}e_k$ with head $v$ was made for an arrival time $arr(e')\le arr(e_k)$ as $E^{arr}$ is node-arrival sorted. This ensures that the argument $j$ of such a call was at most $\max\{l''+1,l'\}$. Similarly, the update of $r_v$ at the end of the corresponding iteration was at most $r'$.
  We thus conclude that we have $l''+1\le l_v\le\max\{l',l''+1\}$ and $r_v\le r'$ at the beginning of the $k$th iteration. The computation of $l$ and $r$ at Lines~\ref{ling:begbest}-\ref{ling:add} thus implies $l=\max\{l',l''+1\}$ and $r=r'$ and the interval $(l,r)$ of $E_v^{dep}$ indeed corresponds to edges of $F$.

  We finally show that each edge $f\in F$ gets associated to its best extendable cost with respect to $E^{arr}[1:k]$. This mainly relies on the induction hypothesis and the fact that $c=c^*$ is the minimum cost of a walk $Q$ in $G_k$ and not in $G_{k-1}$ that $f$ extends and such that $Q.f$ is $(E^{dep},E^{arr})$-respected. Among those edges $f\in F$ which are already associated with a cost $c'$, the removal of intervals at Lines~\ref{ling:setl}-\ref{ling:delinterv} ensures that we modify their associated cost only when $c'$ is greater than $c=c^*$. This relies on the property that $\intervs{v}$ is sorted by non-decreasing cost which is an invariant of the algorithm as the eventual interval $(l_c,r,c,e)$ added at the end of $\intervs{v}$ at Lines~\ref{ling:addinterv} has cost $c$ which is greater or equal to the cost of remaining intervals. The induction hypothesis and the optimality of $c$ ensure that the best extendable cost of these edges is $c$. The update of bound $l_c$ at Line~\ref{ling:delinterv} ensures that these edges get associated to cost $c$. All edges in $F$ that were not previously associated to a cost are those in interval $(\max\{l''+1,l',r_v+1\},r')=(\max\{l,r_v+1\},r)$ which is included in $(l_c,r)$ as $l_c$ is initialized to $\max\{l,r_v+1\}$ at Line~\ref{ling:setl} and can only decrease by the updates at Line~\ref{ling:delinterv}. These edges also get associated to $c$ through interval $(l_c,r,c,e)$, and it is their best extendable cost by optimality of $c$. Finally, all edges in $F$ that were associated to a cost $c'\prec c$ remain associated to the same cost which is their best extendable cost by the induction hypothesis.
\end{proof}

We can now state the following.

\begin{proposition}\label{prop:gal}
  Given a doubly-sorted representation $(E^{dep},E^{arr})$ of a temporal graph $G=(V,E,\alpha,\beta)$ with cost structure $({C},\gamma,\oplus,\preceq)$ satisfying isotonicity, and a source node $s$, Algorithm~\ref{alg:gal} computes in linear time and space, for each node $v$ and each $(s,E^{dep},E^{arr})$-reachable edge $e$ with head $v$, the minimum cost of any $(E^{dep},E^{arr})$-respected $sv$-walk ending with $e$.
\end{proposition}

\begin{proof}
  The correctness of the algorithm follows directly from \Cref{lem:kthiter}. The reason is that any $(E^{dep},E^{arr})$-respected $sv$-walk $Q$ ending with an edge $e$ must have edges appearing in order in $E^{arr}$ so that if $k$ is the index of $e$ in $E^{arr}$, all edges of $Q$ are in $E^{arr}[1:k]$, and $Q$ is also a walk in $G_k$.

  Let us turn to the complexity analysis. Each edge $e\in E^{arr}$ is scanned only once. For all nodes $v$, each edge $f\in E_v^{dep}$ with index $i$ is finalized at most once: the first time $\processcosts(v,j)$ is called with a value $j\ge i$. The update of $l_v$ to $j+1$ in $\processcosts(v,j)$ ensures that $f$ is never finalized again. Computing the value of $l$ at Line~\ref{ling:left} takes $O(l-l_v)$ time, which thanks to the update of $l_v$ to $l$ in the call to $\processcosts(v,l-1)$  results in amortized time of $O(|E^{dep}_v|)$. Similarly, computing the value of $r$ takes $O(r-r_v)$ time, which thanks to the update of $r_v$ to $r$ results in amortized time of $O(|E^{dep}_v|)$. In addition, at most one interval is created at each iteration and later removed. The number of times we modify the left bound of an interval is bounded by the number of times we udpate $l_v$ which is $|E_v^{dep}|$ at most. As $\sum_{v\in V}|E_v^{dep}|=|E|$, Algorithm~\ref{alg:gal} runs in linear time assuming that operations with $\oplus$ and $\preceq$ can be computed in constant time. Finally, let us notice that for all nodes $v$, $\intervs{v}$ contains at most $|E_v^{dep}|$ intervals and the set $A'_v$ has size bounded by the number of temporal edges with head $v$. We thus have $\sum_{v\in V}|\intervs{v}|\le |E|$ and $\sum_{v\in V}|A'_v|\le |E|$. The space complexity of Algorithm~\ref{alg:gal} is thus linear.
\end{proof}

\bigskip

We now turn back to the zero-acyclic case with the following lemma.

\begin{lemma}\label{lem:all}
    Let $G$ be a zero-acyclic temporal graph and let $(E^{dep},E^{arr})$ be a half-extend-respecting doubly-sorted representation of $G$. Then any walk in $G$ is $(E^{dep},E^{arr})$-respected and any $s$-reachable edge is an $(s,E^{dep},E^{arr})$-reachable edge.
\end{lemma}

\begin{proof}
  Consider a walk $Q$ in $G$ and consider two consecutive edges $e,f$ of $Q$. We have to prove that for any edge $e'$ with same tail $v$ as $f$ and satisfying $f\leq_{E^{dep}} e'$, we have $e<_{E^{arr}}e'$. First, we have $arr(e) + \alpha_v \leq dep(f)$ as $f$ extends $e$. Second, $f\leq_{E^{dep}} e'$ implies $dep(f)\le dep(e')$ as $E^{dep}$ is node-departure sorted. Combining both inequalities, we get $arr(e) + \alpha_v \leq dep(e')$, that is $e'$ half-extends $e$. We must thus have $e<_{E^{arr}} e'$ as $E^{arr}$ is half-extend-respecting.
\end{proof}

Theorem~\ref{th:gal} is a direct consequence of the above lemma and Proposition~\ref{prop:gal}.

\section{Solving classical optimal temporal walks problems}
\label{sec:solving}

\paragraph*{Single-source fewest-edges walks}
%\label{sub:fewestedges}
%The one-to-all fewest-edges walk problem consist in computing from a given source a temporal walk to each node $v$ with minimum number of edges. Indeed,
As a very basic example, optimizing the number of edges in Algorithm~\ref{alg:gal} is straightforward: it suffices to consider the cost structure $(\mathbb{N},\gamma,+,\le)$ associated to integers ordered as usual and where each edge $e$ has cost $\gamma(e)=1$, the combination function being addition. It obviously satisfies isotonicity. Using Theorem~\ref{th:gal} 
%and the above extension of Algorithm~\ref{alg:gal}, we thus deduce 
thus implies that the single-source fewest-edges walk problem, that is computing an $sv$-walk with minimum number of edges for all nodes $v$, can be solved in linear time and space.

\paragraph*{Shortest-fastest walks}
A walk with shortest duration is also called a fastest walk, and a fastest walk having a minimum number of edges is called a shortest-fastest walk. For finding such walks, we define a cost structure $({C},\gamma,\oplus,\preceq)$ where $C=\mathbb{R}\times\mathbb{N}$. Given an edge $e=(u,v,\tau,\lambda)$, we define its cost $\delta(e)\in \mathbb{R}\times\mathbb{N}$ as $\gamma(e)=(\tau,1)$. We define the cost combination function $\oplus$ by $(\tau,k)\oplus (\tau',k')=(\tau,k+k')$. A walk departing at time $\tau$ and having $k$ edges thus has cost $(\tau, k)$. We define the cost total order by $(\tau,k)\preceq (\tau',k')$ when $\tau > \tau'$ or $\tau=\tau'$ and $k\le k'$. Among two walks, the one with latest departure is thus always preferred, and among several walks with same departure time, one with fewest edges is always preferred.
Given a source $s$, Algorithm~\ref{alg:gal} now outputs for each destination $v$ the set $A'_v$ of all pairs $(e,c)$ such that $e$ is an $s$-reachable edge with head $v$ and $c=(\tau,k)$ is the minimum cost of an $sv$-walk ending with $e$. Note that our cost definition implies that $\tau$ is the latest departure time of an $sv$ walk ending with $e$ and $k$ is the minimum number of edges among walks with departure time $\tau$ and last edge $e$. We thus obtain the shortest duration of an $sv$-walk as $D^*=\min_{(e,(\tau,k))\in A'_v} arr(e) - \tau$. We then obtain the minimum number of edges in a fastest $sv$-walk as $k^*=\min_{(e,(\tau,k))\in A'_v : arr(e)-\tau = D^*} k$. The edge $e^*$ for which we get the minimum value allows to obtain, through parent pointers, a walk having duration $D^*$ and $k^*$ edges, that is a shortest-fastest walk. Theorem~\ref{th:gal} 
%and the above extension of Algorithm~\ref{alg:gal}, we thus deduce 
thus implies that the single-source shortest-fastest walk problem can be solved in linear time and space.

\paragraph*{Linear combination of classical criteria}
%\label{sub:lincomb}
To exemplify the generality of the algebraic approach, we now give an example of cost structure allowing Algorithm~\ref{alg:gal} to compute optimal temporal walks for the linear combination of criteria used in~\cite{BentertHNN2020}.
% (with an impressive dexterity in handling the complex cost updates required along their algorithm with such an exhaustive combination).
Our formalism enables more modularity as all complex updates required by such an exhaustive combination are then encapsulated in operations $\oplus$ and $\prec$.
Given a walk $Q=\langle e_1=(v_0,v_1,\tau_1,\lambda_1),\ldots,e_k=(v_{k-1},v_k,\tau_k,\lambda_k)\rangle$, we consider the following criteria that we usually seek to minimize:

\smallskip
\begin{tabular}{lll}
    (1) & $\tau_k+\lambda_k$ &   arrival time (or foremost) \\
    (2) & $-\tau_1$ & departure time (or reverse-foremost)\\
    (3) & $\tau_k+\lambda_k-\tau_1$ & duration (or fastest)\\
    (4) & $\sum_{i=1}^k\lambda_i$ & total travel time (or shortest)\\
    (5) & $\sum_{i=1}^k c(e_i)$ & total cost (each edge $e\in E$ is associated to a cost $c(e)\in \mathbb{R}$)\\ %_{\ge 0}$)\\
    (6) & $k$ & number of edges (or fewest-edges)\\
    (7) & $\sum_{i=1}^{k-1}\tau_{i+1} - (\tau_i+\lambda_i)$ & total waiting time\\
\end{tabular}
\smallskip

Given $\delta_1,\ldots,\delta_7\in \mathbb{R}$, %_{\ge 0}$,
the \emph{linear combined cost} of $Q$ %for scalars $(\delta_1,\ldots,\delta_7)$
is defined in~\cite{BentertHNN2020} as:
\begin{multline*}
lin(Q) = \delta_1(\tau_k+\lambda_k) 
+ \delta_2(-\tau_1)
+ \delta_3(\tau_k+\lambda_k-\tau_1) 
\\
+ \delta_4(\sum_{i=1}^k\lambda_i)
+ \delta_5(\sum_{i=1}^k c(e_i))
+ \delta_6 \,\, k
+ \delta_7(\sum_{i=1}^{k-1}\tau_{i+1} - (\tau_i+\lambda_i)).
\end{multline*}

It is simply a linear combination of all classical criteria.
Note that we do not need to assume non-negativity of costs or scalars $\delta_1,\ldots,\delta_7$, enabling a more general framework than~\cite{BentertHNN2020}. % see section zero delay
To optimize such a combined cost, we define the cost structure $({C},\gamma,\oplus,\preceq)$ where $C=\mathbb{R}\times\mathbb{R}$. Given an edge $e=(u,v,\tau,\lambda)$, we define its combined cost $\delta(e)\in \mathbb{R}$ and its cost $\gamma(e)\in \mathbb{R}\times\mathbb{R}$ as:
$$
\delta(e) = (\delta_4-\delta_7)\lambda_i + \delta_5 c(e_i) + \delta_6
\mbox{ and }
\gamma(e) = (\tau,\delta(e)).
$$
Observe that they are linked to the linear combined cost of $Q$ by:
$$
lin(Q) = (\delta_1+\delta_3+\delta_7) \, arr(Q)  - (\delta_2+\delta_3+\delta_7) \, dep(Q) + \sum_{i=1}^k \delta(e_i).
$$
Recall that $arr(Q)=\tau_k+\lambda_k$ and $dep(Q)=\tau_1$ are the arrival time and the departure time of $Q$ respectively.
We define the cost combination function $\oplus$ by
$$
(\tau,\Delta)\oplus(\tau',\Delta') = (\tau,\Delta+\Delta').
$$
This definition implies that the cost of $Q$ is then $\gamma_Q=(\tau_1, \sum_{i=1}^k \delta(e_i))=(\tau,\Delta)$ with $\tau=dep(Q)$ and $\Delta=\sum_{i=1}^k \delta(e_i)$.
We finally define the cost total order $\preceq$ by
$$
(\tau,\Delta)\preceq (\tau',\Delta') \text{ when }  - (\delta_2+\delta_3+\delta_7)\tau + \Delta\le - (\delta_2+\delta_3+\delta_7)\tau' + \Delta'.
$$
This order is related to the minimization of $lin(Q)$ for a fixed arrival time $a$:
for all $sv$-walks $Q$ such that $arr(Q)=a$, minimizing $lin(Q)$ is equivalent to minimizing $- (\delta_2+\delta_3+\delta_7) \, dep(Q) + \sum_{i=1}^k \delta(e_i)=- (\delta_2+\delta_3+\delta_7)\tau + \Delta$ where $(\tau,\Delta)=\gamma_Q$ is the cost of $Q$. A walk $Q$ with minimum cost according to $\preceq$ thus has minimum value for $lin(Q)$ among all walks with same arrival time.

Note that the cost structure satisfies the isotonicity property: for any costs $(\tau_1,\Delta_1),(\tau_2,\Delta_2)$, $(\tau,\Delta)\in \mathbb{R}\times\mathbb{R}$, we have $(\tau_1,\Delta_1)\oplus(\tau,\Delta)=(\tau_1,\Delta_1+\Delta)$
and $(\tau_2,\Delta_2)\oplus(\tau,\Delta)=(\tau_2,\Delta_2+\Delta)$.
If $(\tau_1,\Delta_1)\preceq (\tau_2,\Delta_2)$, then we have
$- (\delta_2+\delta_3+\delta_7)\tau_1 + \Delta_1\le - (\delta_2+\delta_3+\delta_7)\tau_2 + \Delta_2$. By adding $\Delta$ on both sides of the inequality, we obtain $(\tau_1,\Delta_1)\oplus(\tau,\Delta)\preceq (\tau_2,\Delta_2)\oplus(\tau,\Delta)$.

Now running Algorithm~\ref{alg:gal} with this cost structure from a source node $s$ allows to compute for each destination $v$ the set $A_v'$ of all pairs  $(e,c)$ such that $e$ is an  $s$-reachable edge with head $v$ and $c=(\tau,\Delta)$ is the minimum cost of any $sv$-walk end with $e$ according to our cost structure. The minimum linear combination cost of an $sv$-walk can then be obtained through a linear scan of $A_v'$ as:
$$
\min \{ lin(Q) : Q \mbox{ is an $sv$ walk} \} = \min_{(e,(\tau,\Delta))\in A_v'} (\delta_1+\delta_3+\delta_7) \, arr(e) - (\delta_2+\delta_3+\delta_7) \, \tau + \Delta.
$$
This is due to the fact that for a given arrival time $arr(e)$, minimizing $lin(Q)$ is equivalent to minimizing $\gamma_Q$ according to $\preceq$, as discussed above, and that $A_v'$ contains a pair for all $s$-reachable edges with head $v$.
Using the above cost structure, we thus obtain the following corollary.

\begin{corollary}\label{cor:lincomb}
Given a half-extend-respecting doubly-sorted representation $(E^{dep},E^{arr})$ of a zero-acyclic temporal graph $G=(V,E,\alpha,\beta)$, a source node $s$, and $\delta_1,\ldots,\delta_7\in \mathbb{R}$,
the single-source minimum-combined-cost walk problem, that is computing  for all nodes $v$ an $sv$-walk with minimum linear combined cost for $(\delta_1,\ldots,\delta_7)$, can be solved in linear time and space.
\end{corollary}

\paragraph*{Minimum-overall-waiting-time walks.}
%
%In particular, the \emph{one-to-all minimum-waiting walk problem}, that is computing an $sv$-walk with minimum total waiting time for all nodes $v$, can be solved in linear time and space by setting $\delta_1=\cdots=\delta_6=0$ and $\delta_7=1$.
%
Setting $\delta_7=1$ and $\delta_i=0$ for $i\not=7$, the above corollary implies in particular the existence of an algorithms that, given a half-extend-respecting doubly-sorted representation $(E^{dep},E^{arr})$ of a zero-acyclic temporal graph $G=(V,E,\alpha,\beta)$, and a source node $s$, can find  in linear time and space
%\begin{description}
%\item[one-to-all fewest-edges walks:]  $sv$-walks with minimum number of edges for all nodes $v$,
%\item[
\emph{single-source minimum-overall-waiting-time walks}, i.e.  $sv$-walks with minimum overall waiting time for all $v$.
%\end{description} 

%% \begin{corollary}\label{cor:fewestedges}
%% Given a doubly-sorted representation $(E^{arr},E^{dep})$ of a strict temporal graph $G=(V,E,\alpha,\beta)$, a source node $s$, and $\delta_1,\ldots,\delta_7\in \mathbb{R}$,
%% the one-to-all minimum-combined-cost walk problem, that is computing  for all nodes $v$ an $sv$-walk with minimum linear combined cost for $(\delta_1,\ldots,\delta_7)$, can be solved in linear time and space.
%% \end{corollary} 

\paragraph*{Profiles.} The profile function from a source node $s$ to a destination node $v$ associate to any starting time $t$ the earliest arrival time of a walk from $s$ to $v$ departing at time $t$ or later. The single-source profile problem consists in computing, for a given source $s$ and for every node $v$, a representation of the profile function from $s$ to $v$. %TODO check \textcolor{blue}
{A classical representation consists in listing all pairs $(d,a)$ where $a$ is a possible arrival time at $v$ and $d$ is the latest departure time of a walk arriving at time $a$ or before, such that $a$ is also the earliest arrival time at $v$ when departing from $s$ at time $d$ or later.} %\fb{I think we have to say something more about $a$. If $a$ is just possible arrival time there might not exists a piecewise-constant non-decreasing function passing through all the points.}
The profile function is indeed the only piecewise-constant non-decreasing function passing through these points. Note that this classical representation was introduced in the setting where waiting is unrestricted.

This extends naturally to waiting constraints when waiting at the source is unrestricted at starting time, i.e. we impose the waiting constraint only between two consecutive edges of a walk, and consider any walk from $s$ departing at time $t'\ge t$ as a valid walk when starting at time $t$ from $s$.
By setting $\delta_2=1$ and $\delta_i=0$ for $i\not=2$ in the cost structure for a linear combination of classical criteria described above,  Algorithm~\ref{alg:gal} then outputs for each $v\in V$ and for each $s$-reachable edge with head $v$ the pair $(e,d)$, where $d$ is the latest departure time of a walk from $s$ to $v$ ending with edge $e$. These pairs are ordered in $A'_v$ by non-decreasing arrival time of $e$. A linear scan allows to replace each pair $(e,d)$ by $(d, arr(e))$. Finally, it suffices to remove pairs that are Pareto dominated, that is those pairs $(d,a)$ for which a pair $(d',a')$ satisfies either $a'<a$ and $d'\ge d$ or $a'\le a$ and $d' > d$. As the pairs are sorted by non-decreasing arrival time, a single scan in reverse order allows to filter out those dominated pairs by comparing each pair with the last non dominated pair. The list of pairs $(a,d)$ we obtain this way is ordered both by increasing arrival time and by increasing departure time. One can easily see that it is a representation of the profile function. Note that filtering out dominated pairs relies on the assumption that waiting is unrestricted at the source $s$ when starting from it. 
%\lv{Filtering out dominated pairs relies on unrestricted waiting at source! We should put a remark about this and that it can be adressed by using a symetric version of our algorithm. Proposal:}

%\fb{Maybe should make a more clear distinction of resitrcited/unrestricted waiting at the source}
We now consider the setting where waiting at $s$ should be bounded by $\beta_s$ also when starting from it, that is a walk from $s$ departing at time $t'$ is considered as valid walk when starting at time $t$ from $s$, only if we have $t+\alpha_s\le t'\le t+\beta_s$.
Note that the profile function might not be non-decreasing in this setting. However, we can still solve the single-destination profile problem as follows. We run our algorithm on the reverse temporal graph where time is reversed, and which is obtained by turning each edge $(u,v,\tau,\lambda)$ into $(u,v,-(t+\lambda),\lambda)$.
This is equivalent to running a symmetric version of our algorithm for solving the single-destination problem by scanning edges backwards and computing for each departing edge at a node the earliest arrival time at the destination (which corresponds to the latest departure time in the reverse temporal graph).
For a given destination $x$, this allows to obtain for each node $v$, and for each edge $e$ departing from $v$, the earliest arrival time $a$ of walks from $v$ to $x$ starting with $e$. As this list is sorted by departure time, we can obtain in linear time the pairs $(d,a)$ where $a$ is the minimum arrival time among edges departing at time $d$. Each pair $(d,a)$ provides the earliest arrival time when starting in interval $(d-\beta_v,d-\alpha_v)$ and using an edge departing at time $d$.
A representation of the profile function from $v$ to $x$ can be obtained by keeping for each window of time covered by multiple overlapping intervals, the lowest earliest arrival time corresponding to such intervals. It can be computed  by merging the list of the left bounds of these intervals with the list of their right bounds: scanning the resulting list, while maintaining a queue of currently open left bounds with their associated arrival time, allows to compute for each consecutive interval of starting times, the earliest arrival time.
We omit the details of how to get efficiently the minimum arrival time associated to open intervals in the queue. 
%\lv{Check how to get the minimum arrival time in the queue in constant amortized time.} 
%LV the following is WRONG!!!! : we still consider waiting constraint at nodes rather than the source.
%Finally, we can solve the single-source profile problem by running this symmetric algorithm on the reverse temporal graph where time is reversed and which is obtained by turning each edge $(u,v,\tau,\lambda)$ into $(u,v,-t-\lambda,\lambda)$.

\begtodolater

\subsection{Single-source shortest-duration walks}
\label{sub:shortestduration}

The duration of a walk $Q$ is defined as the time $arr(Q)-dep(Q)$ elapsed between its departure and its arrival. It is one of the classical criteria for evaluating the quality of walks, shortest duration being usually preferred.
As an example, we now show how to obtain shortest-duration walks with Algorithm~\ref{alg:gal} and obtain the following result as a corollary of Theorem~\ref{th:gal}.

\begin{corollary}\label{cor:shortestduration}
Given a doubly-sorted representation $(E^{arr},E^{dep})$ of a strict temporal graph $G=(V,E,\alpha,\beta)$, and a source node $s$,
the single-source shortest-duration walk problem, that is computing a shortest duration $sv$-walk for all nodes $v$, can be solved in linear time and space.
\end{corollary}

For a given destination $v$ and a given arrival time $a$, the duration of any $sv$-walk $Q$ arriving at time $a$ is $a-dep(Q)$. It is thus natural to consider $-dep(Q)$ as the cost of $Q$. This is indeed possible with the following cost structure $(C,\gamma,\oplus,\preceq)$ with $C=\mathbb{R}$, for each edge $e=(u,v,\tau,\lambda)\in E$ its cost is defined as $\gamma(e)=-\tau$, the combination function is defined by $t\oplus t'=t$ for all $t,t'\in\mathbb{R}$, and $\preceq$ is the usual $\le$ total order for reals. Note that the $\oplus$ definition implies that the cost of a walk $Q=\langle e_1,\ldots,e_k\rangle$ depends only on its first edge $e_1=(u_1,v_1,\tau_1,\lambda_1)$ and is given by $\gamma_Q=\gamma(e_1)=-\tau_1=-dep(Q)$.
The isotonicity property is obviously satisfied as for any $t_1,t_2,t\in \mathbb{R}$, we have $t_1\oplus t=t_1$ and $t_2\oplus t=t_2$.

Running Algorithm~\ref{alg:gal} using parent pointers as explained in Section~\ref{sub:parent}, thus results in a set $A'_v$ for each node $v$ where every possible arrival time $a$ is associated to a walk $Q$ minimum cost, that is minimizing $-dep(Q)$, or equivalently maximizing $dep(Q)$. Indeed, $Q$ is thus an $sv$-walk arriving at time $a$ with last departure time. Its duration $arr(Q)-dep(Q)=a-dep(Q)$ is thus guaranteed to be minimum among all $sv$-walks with arrival time $a$. We can finally scan all possible arrival times $a$ in $A_v'$ to obtain an $sv$-walk with overall shortest duration. Corollary~\ref{cor:shortestduration} is thus a consequence of Theorem~\ref{th:gal}.

\smallskip
Interestingly, when waiting is unrestricted, that is $\beta_v=+\infty$ for all $v\in V$, it is always possible to find for any shortest duration $sv$-walk $Q$, a temporal path $P$ from $s$ to $v$ with same duration, where a temporal path from $s$ to $v$ is an $sv$-walk visiting each node at most once.
It indeed suffices to wait instead of following a loop in $Q$, and $Q$ can easily be transformed into a temporal path by removing its loops (when assuming unrestricted waiting).
However, Algorithm~\ref{alg:gal} may still produce walks with loops. If temporal paths are desired, it is possible to post-process the output or use a modified cost structure preventing such loops. A possible such modification consists in preferring walks with less edges among those with same departure time, and is left to the reader.

\subsection{Pareto optimal walks}
\label{sub:pareto}

We now show how our algorithm generalizes the Pareto set computation described in~\cite{BrunelliCV2021}. We first recall the definition of the Pareto problem.
% Pareto
%Given a temporal graph $G=(V,E)$ with a cost structure $\mathcal{C}=({C},\gamma,\oplus,\preceq)$ over $E$, 
We say that a pair $(a_1,c_1)\in\mathbb{R}\times{C}$ \emph{dominates} a pair $(a_2,c_2)\in\mathbb{R}\times{C}$ if $a_1 < a_2$ and $c_1\preceq c_2$, or $a_1\leq a_2$ and $c_1\prec c_2$. 
%\fb{CHECK}\lv{the sequel} %Moreover, for any two nodes $u,v\in V$ and for any $t\in\mathbb{R}$, let $\allwalks uv(t)$ denotes the set of all walks from $u$ to $v$ with starting time at least $t$.  
%Given a time $t_0\in\mathbb{R}$, a walk $Q\in\allwalks sd(t_0)$ is \emph{Pareto $t_0$-optimal} among all walks in $\allwalks sd(t_0)$, if there is no walk $Q'\in\allwalks sd(t_0)$ such that $(arr(Q),\gamma_{Q})$ dominates $(arr(Q'),\gamma_{Q'})$. 
Consider a strict temporal graph $G=(V,E,\alpha,\beta)$ with cost structure $({C},\gamma,\oplus,\preceq)$.
An $sv$-walk $Q$ is \emph{Pareto optimal} if there is no $sv$-walk $Q'$ such that $(arr(Q'),\gamma_{Q'})$ dominates  $(arr(Q'),\gamma_{Q'})$.
The \emph{single-source Pareto} problem is then defined as follows. Given  a source node $s\in V$, compute, for each destination $v\in V$, the set $P_v$ containing all pairs $(a,c)\in\mathbb{R}\times{C}$ for which there exists a Pareto optimal walk $Q$ such that $a=arr(Q)$ and $c=\gamma_{Q}$.
%\lv{In the special case where $\alpha_v=0$ and $\beta_v=\infty$ for all $v$, we call it the \emph{unconstrained Pareto problem}.}
The problem defined in~\cite{BrunelliCV2021} is slightly more general as a time $t_0$ is also given, and only walks departing no earlier than $t_0$ are considered. An instance of this more general problem can easily be reduced to our version by discarding edges departing before $t_0$.

%Algorithm~\ref{alg:gal} indeed solves the one-to-all Pareto problem. Given a source $s$, consider the list $A_v'$ returned for node $v$. It contains all pairs $(a,c)$ such some $sv$-walks arrive at time $a$ and such that $c$ is the minimum cost among them.
The single-source Pareto problem can easily be solved by Algorithm~\ref{alg:gal} as it computes for each destination $v$ the set $A_v'$ of all pairs $(a,c)$ such that some $sv$-walks arrive at time $a$ and such that $c$ is the minimum cost among them. It thus suffices to remove from $A'_v$ dominated pairs to obtain $P_v$. As $A_v'$ is ordered by non-decreasing value of arrival time, this can easily be done in linear time.

\lv{Note that when waiting is unrestricted, we indeed have $A_v'=P_V$. (Are you sure?)}

\endtodolater

\section{Lower bound for the single-source optimal walk problem}
% \section{Lower bound for one-to-all shortest-duration walks}
\label{sec:lb}
% We now show that computing optimal temporal walks from an unsorted input cannot be performed in linear time by a \emph{comparison-only} algorithm, that is an algorithm using only comparisons for testing if an edge extends another. This holds even if the algorithm requires \emph{sorted arrivals as input} \lv{TODO: change to \emph{arrival-sorted representation}}, that is when it is required that the input contains the list of temporal edges sorted by non-decreasing arrival times. \lv{Use later? : We similarly say that it requires \emph{sorted departures as input} when the input must contain the list of temporal edges sorted by non-decreasing departure times.}

We now show that computing optimal temporal walks in linear time somehow requires both orderings needed by our algorithm. More precisely, we define an arrival-sorted representation (resp. a departure-sorted representation) of a temporal graph as a list of its temporal edges sorted by non-decreasing arrival times (resp. non-decreasing departure times). 
\todolater{Compare with node-arrival (resp. node-departure) sorted.}
We say that an algorithm is \emph{comparison-only} when it uses only comparisons for deciding whether an edge extends another one, or for deciding which walk has minimum cost among several walks. We show that any comparison-only algorithm optimizing general costs that can encompass overall waiting time, and taking as input either a departure-sorted representation or an arrival-sorted representation, must be slower than linear time by a logarithmic factor at least for some inputs.

% \begin{theorem}\label{th:lb}
% There exists a family of strict temporal graphs $(G_n)_{n\ge 0}$ with $O(n)$ nodes and $O(n)$ temporal edges with a marked source node, such that any comparison-only and sorted-arrival-restricted algorithm computing one-to-all shortest-duration walks in $G_n$ from the marked node has time complexity $\Omega(n\log n)$. Furthermore, this bound holds also for any such algorithm requiring sorted arrivals as input.
% \end{theorem}

%\fb{define minimum waiting walk}

\begin{theorem}\label{th:lb}
For each integral $n$ there exists a family of instances $\mathcal{I}_n$ (resp. $\mathcal{I}'_n$) of temporal graphs with unrestricted waiting and strictly positive travel times, given as departure-sorted representations (resp. arrival-sorted representations) with $O(n)$ nodes and $O(n)$ temporal edges, such that any comparison-only  deterministic algorithm computing single-source minimum-\-overall-\-waiting-\-time walks from instances in $\mathcal{I}_n$ (resp. $\mathcal{I}'_n$) has time complexity $\Omega(n\log n)$. Moreover, for any comparison-only randomized algorithm computing single-source minimum-\-overall-\-waiting-\-time walks from instances in $\mathcal{I}_n$ (resp. $\mathcal{I}'_n$), there exists an instance in $\mathcal{I}_n$ (resp. $\mathcal{I}'_n$) for which the expected running time is $\Omega(n\log n)$.
\end{theorem}

% The proof follows from the reduction from a sorting problem to the one-to-all shortest-duration walks problem. 
% Note that the result can be extended to randomized algorithms by using Yao's principle.
% The construction does not use any waiting-time constraints, and the walks of the solution are indeed temporal paths.

%Note that the result holds both for deterministic and randomized algorithms, by using Yao's principle \fb{here and in the proof?}.
%Moreover, the construction does not use any waiting-time constraints, and the walks of the solution are indeed temporal paths.
Recall that unrestricted waiting is equivalent to $\alpha_u=0$ and $\beta_u=+\infty$ for all $u$ as in the classical setting for temporal graphs.
This result also holds when restricting the temporal graph model to integer times and $O(n)$ lifetime, that is when the union of departure and arrival times of all edges is included in an interval of length $O(n)$.

\begin{proof}
The first step of the proof is to build $\mathcal{I}_n$ and $\mathcal{I}'_n$.
Let us fix $2n$ integer times $\tau_1,\ldots,\tau_n$ and $t_1,\ldots,t_n$ such that $0<\tau_1<t_1<\tau_2<t_2<\cdots<\tau_n<t_n<3n$. Consider the set of vertices $V = \{s,u,v_1, \dots, v_n\}$.
For any two permutation $\pi$ and $\pi'$ of $[n]=\{1,\ldots,n\}$, we define the temporal edges $E_{\pi} = \{ e_i^{\pi} = (s,u, -i, \tau_{\pi(i)} + i) : i=1,\dots,n\}$, and $F_{\pi'} = \{f_j^{\pi'} = (u,v_j,t_{\pi'(j)},t_n+j-t_{\pi'(j)}) : j =1, \dots, n \}$. We can now define the temporal graphs $G_{\pi,\pi'} = (V, E_{\pi} \cup F_{\pi'}, \alpha, \beta)$, where $\alpha_x = 0$ and $\beta_x = +\infty$ for all $x \in V$. 

The family of instances $\mathcal{I}_n$ is given by the temporal graphs $\bigcup_{\pi} G_{\pi,id}$ given as departure-sorted representations, where $s$ is marked as the source node and $id$ denotes the identity permutation. Notice that the ordering of its temporal edges $E_{\pi} \cup F_{id}$ by non-decreasing departure time, is $(e_1^{\pi}, \dots, e_n^{\pi}, f_1^{id}, \dots, f_n^{id})$ for any $\pi$.
Similarly, the family of instances $\mathcal{I}'_n$ is given by the temporal graphs $\bigcup_{\pi'} G_{id,\pi'}$ given as arrival-sorted representations. Note also that the ordering of its temporal edges $E_{id} \cup F_{\pi'}$ by non-decreasing arrival time, is $(e_1^{id}, \dots, e_n^{id}, f_1^{\pi'}, \dots, f_n^{\pi'})$ for any $\pi'$. 

Now suppose that we are given a deterministic algorithm $A$ for computing minimum-\-overall-\-waiting-\-time walks from $s$ to all nodes. In the temporal graph $G_{\pi,id}$ the possible temporal walks from $s$ to $v_j$ are given by $\langle e_i^{\pi},f_j^{id} \rangle$ such that $\tau_{\pi(i)} \leq t_j$, and the overall waiting time of such walk is $t_j - \tau_{\pi(i)}$. The minimum overall waiting time is thus obtained for the largest $\tau_{\pi(i)} \leq t_j$, that is for $\pi(i) = j$. This means that, there is a one to one correspondence between the outputs of $A$ and the permutations of $[n]$.  In particular, if algorithm $A$ is correct then there are at least $n!$ possible different outputs. Since $A$ is a deterministic comparison only algorithm and the input instance order $(e_1^{\pi}, \dots, e_n^{\pi}, f_1^{id}, \dots, f_n^{id})$ does not depend on $\pi$, two executions of $A$ with same comparisons lead to the same output. This means that, if we denote with $c$ the maximum number of comparisons made by $A$, there are at most $2^c$ different outputs. The correctness of $A$ thus implies $2^c \geq n!$. We then get $c \geq n \ln{n} -n$ and conclude that the time complexity of $A$ is $\Omega(n \log n)$.

More precisely, consider the decision tree corresponding to each time comparison. We have just argued that this tree has depth $n(\ln n - 1)$ at least. Consider the permutations where the execution terminates after $\frac{n}{2}\ln n$ comparisons only. As the subtree corresponding to such executions has at most $n^{n/2}$ leaves, there are at most $n^{n/2}$ such permutations. On instances built according to other permutations, algorithm $A$  requires at least $\frac{n}{2}\ln n$ comparisons. With uniform distribution over the inputs in $\mathcal{I}_n$, the average complexity of $A$ is thus at least $\frac{n! - n^{n/2}}{n!}\frac{n}{2}\ln n \ge  (1 - \exp(n-\frac{n}{2}\ln n))\frac{n}{2}\ln n = \Omega(n\ln n)$.
Yao's principle then implies that for any randomized algorithm solving the single-source minimum-\-overall-\-waiting-\-time walk problem, there exists an instance in $\mathcal{I}_n$ on which its average running time is $\Omega(n\log n)$.

Similarly, in a temporal graph $G_{id,\pi'}$ the possible temporal walks from $s$ to $v_j$ are given by $\langle e_i^{id},f_j^{\pi'} \rangle$ such that $\tau_{i} \leq t_{\pi'(j)}$, and the overall waiting time of such walk is $t_{\pi'(j)} - \tau_{i}$. The minimum overall waiting time is thus obtained for the largest $\tau_{i} \leq t_{\pi'(j)}$, that is for $i = \pi'(j)$. This, again, means that, there is a one to one correspondence between the outputs of $A$ and permutations of $[n]$ and we can conclude similarly to the previous case.
\end{proof}
\todolater{Put the remark for shortest duration}

% \todolater{What if the edges are ordered by dep time?}

\section{Equivalence between space-time and doubly-sorted representations}
\label{sec:representations}

We now show that the doubly sorted representation is equivalent to the classical ``space-time'' representation~\cite{PallottinoS1997}. The latter consists in transforming a temporal graph into a static graph by introducing a copy of each node for each possible time instant. Each temporal edge is then turned into a static edge from the two corresponding copies of its tail and head. We consider here a variant where we introduce copies of a node only for time instants corresponding to a departure time of an edge from that node, or an arrival time of an edge to that node, following the approach of~\cite{SchulzWW2000}.

Formally, given a temporal graph $G=(V,E)$, its \emph{space-time  representation} is a directed graph $D=(W,F^c \cup F^w)$, where:
\vspace{-.5\baselineskip}
\begin{itemize}
\addtolength{\itemsep}{-.5\baselineskip}
    \item The nodes in $W$ are labeled nodes $v_{\tau}$, where $v\in V$ refers to a node of $G$ and $\tau$ is a time label. 
    More precisely, $v_{\tau} \in W$ if and only if there exists a temporal edge in $E$ with tail $v$ and departure time $\tau$ or a temporal edge with head $v$ and arrival time $\tau$. We will also refer to such nodes as \emph{copies of $v$}. Let us denote with $Pred^w(v_\tau)$ the copy of $v$ in $W$ with maximum time label less than $\tau$, if it exists.
  \item We distinguish two types of arcs $F^c$ and $F^w$ called connection arcs and waiting arcs respectively. The set $F^c$ contains an arc $(u_\tau,v_{\tau+\lambda})$ for each temporal edge $e=(u,v,\tau,\lambda)\in E$. These arcs represent a temporal connection between nodes in $V$ and are called \emph{connection arcs}. Note that each arc $(v_{\tau}, w_{\nu})$ in $F^c$ satisfies $\tau \leq \nu$, since travel times are non-negative.
    The set $F^w$ is defined to contain an arc $(Pred^w(v_\tau),v_\tau)$ for each $v\in V$ and for each copy $v_\tau$ of $v$ such that $Pred^w(v_\tau)$ is defined. These arcs represent the possibility to wait at a node in $v \in V$ during a walk in $G$ and are called \emph{waiting arcs}. Note that each arc $(v_\tau,v_\nu)$ in $F^w$ satisfies $\tau < \nu$.
    As we allow temporal edges  which are self loops (i.e. edges $(v,v,\tau,\lambda)$), there might exist two copies of an arc in $D$, one in $F^c$ and one in $F^w$. Formally, $D$ is thus a directed multigraph as we distinguish arcs in $F^c$ from those in $F^w$ and assume $F^c\cap F^w=\emptyset$. % which belong both to $F^c$ and $F^w$. 
\end{itemize}

%\fb{Split the following into two propositions for the sake of clairty of the proof}
%Now we show that there is some sort of equivalence between the two representations. More formally:
We can now state the following equivalence.

\begin{proposition}\label{prop:eq}
Let $G=(V,E,\alpha,\beta)$ be a temporal graph.
\vspace{-.5\baselineskip}
\begin{itemize}
\addtolength{\itemsep}{-.5\baselineskip}
    \item[A.] If $G$ is zero-acyclic and a space-time representation of $G$ is given, it is possible to compute in linear time and space a doubly-sorted representation $(E^{dep},E^{arr})$ of $G$ such that $E^{dep}$ and $E^{arr}$ are both half-extend-respecting.
    \item[B.] Given a doubly-sorted representation $(E^{dep},E^{arr})$ of $G$ it is possible to compute in linear time and space a space-time representation of $G$.
\end{itemize}
\end{proposition}

\begin{proof}
In order to prove Proposition~\ref{prop:eq}.A, we design an algorithm that computes a node-arrival-sorted half-extend-respecting list from a space-time representation of a zero-acyclic temporal graph. It is inspired by Kahn's algorithm for computing a topological ordering of a directed acyclic graph~\cite{Kahn62}.

We first define a notion of extending for arcs in $D$. Given two arcs $f_1,f_2$ in $F^c\cup F^w$, we say that $f_2$ \emph{arc-extends} $f_1$ when the head $v_\nu$ of $f_1$ is also the tail of $f_2$ and we have $f_1\in F^w$ or $f_2\in F^w$ or $\alpha_v=0$ ($v$ is the node whose copy $v_\nu$ is the head of $f_1$). In particular, when $f_1$ and $f_2$ are both connection arcs, we must have $\alpha_v=0$.
Note that for every pair $e_1,e_2$ of temporal edges corresponding respectively to two connection arcs $f_1,f_2$ in $F^c$, and such that $e_2$ extends $e_1$, there must exist a path $P$ in $D$ starting with $f_1$, ending with $f_2$, and containing possibly intermediate waiting arcs in $F^w$. When $arr(e_1)=dep(e_2)$, $P$ contains only $f_1$ and $f_2$, and the minimum waiting time $\alpha_v$ at the head $v$ of $e_1$ must be zero since $e_2$ extends $e_1$. In all cases, each arc in $P$ arc-extends the preceding one according to our new definition.  We say that an ordering $F^{ord}$ of the arcs of $D$ is \emph{arc-extend-respecting} when we have $f_1<_{F^{ord}} f_2$ whenever $f_2$ arc-extends $f_1$ for any pair of arcs $f_1,f_2\in F^c\cup F^w$. It is thus sufficient to produce an arc-extend-respecting ordering $F^{ord}$ of $F^c\cup F^w$ to obtain a half-extend-respecting ordering of the temporal edges according the respective positions of their corresponding arcs in $F^{ord}$.

The main idea of the algorithm is to produce such an ordering by iteratively removing an arc from $D$ so that no other remaining arc arc-extends it. Each time an arc is removed, it is prepended to the list $F^{ord}$ which is initially empty. When a node $v_\nu$ has out-degree zero, we can safely remove all the arcs entering it. Repeating this would suffice when $D$ is acyclic. However, it may contain a cycle. This can only occur when all nodes in the cycle have same time label $\tau$ as each arc $(u_\tau,v_\nu)$ satisfies $\tau\le\nu$. This implies that all the arcs of the cycle must be in $F^c$. In such a case, the zero-acyclicity of $G$ ensures that at least one node $v_\nu$ of the cycle is a copy a node $v$ of the temporal graph with waiting-time constraint $\alpha_v>0$ as otherwise this cycle would correspond to a zero-cycle in $G$. When no node has out-degree zero, the algorithm thus selects any node $v_\nu$ having no out-arc in $F^w$ and satisfying $\alpha_v>0$, and then removes all its in-arcs that are in $F^c$. Note that all out-arcs of $v_\nu$ are then in $F^c$ and none of them arc-extends these in-arcs by the choice of $v_\nu$ such that $\alpha_v\not=0$.
Such a node must exist when $G$ is zero-acyclic and no remaining node has out-degree zero as there must then exist a cycle among nodes with maximum time label $\tau$ while no remaining waiting arc can lead to a copy with time label greater than $\tau$. As the algorithm can always progress, it terminates when all arcs have been removed.  See Algorithm~\ref{alg:sttds} for a formal description, where $\delta_{out}^{c}(v_{\tau})$ is the number of out-neighbors of $v_\tau$ through an arc in $F^c$, $\delta_{out}^{w}(v_{\tau})$ is the number of out-neighbors of $v_\tau$ through an arc in $F^w$. Furthermore, we denote with $N_{in}^{c}(v_{\tau})$ the set of in-neighbors $w_\nu$ of $v_\tau$ such that $(w_\nu,v_\tau) \in F^c$.

\begin{algorithm}[t]
\DontPrintSemicolon
\Input{A space-time representation $D=(W, F^c \cup F^w)$ of a zero-acyclic temporal graph $G=(V,E)$ with waiting-time constraints $(\alpha,\beta)$, given as adjacency lists $N_{in}^{c}(v_{\tau}), Pred^w(v_\tau)$ for each $v_{\tau}\in W$.}
\Output{An arc-extend-respecting ordered list $F^{ord}$ of the arcs in $D$.}

Compute $\delta_{out}^{c}(v_{\tau})$ and $\delta_{out}^{w}(v_{\tau})$ for each node $v_{\tau} \in W$.\label{lins:prebeg}

Compute the set $S$ of nodes $v_{\tau}$ such that $\delta_{out}^{c}(v_{\tau}) = 0$ and $  \delta_{out}^{w}(v_{\tau})= 0$.\label{lins:sets} \\

Compute the set $S'$ of nodes $v_{\tau}$ such that $\delta_{out}^{w}(v_{\tau}) = 0$ and $\alpha_v > 0$. \label{lins:setsp}\\
Set $F^{ord}:=\emptyset$.\label{lins:preend} \Comment{Arc-extend-respecting ordered list.}

\While{$S \cup S' \neq \emptyset$}{
    \eIf{$S \neq \emptyset$}{
        Select any $v_{\tau}\in S$.\\
        \lFor{each node $u_{\nu}$ in $N_{in}^{c}(v_{\tau}) \cup Pred^w(v_\tau)$}{
          $\removearc(u_{\nu},v_{\tau})$.
        }
        Remove $v_{\tau}$ from $D$.
    }{
      Select any $v_{\tau}\in S'$.\\
      \lFor{each node $u_{\nu}$ in $N_{in}^{c}(v_{\tau})$}{
        $\removearc(u_{\nu},v_{\tau})$.
      }
    }

    % Choose any $v_{\tau}\in S\cup S'$. \label{lins:choose} \\
    % \For{each node $u_{\nu}$ in $N_{in}^{c}(v_{\tau})$}{
    %     prepend  $(u,v, \nu,\tau -\nu)$ to $E^{arr}$.\\
    %     Remove $(u_{\nu},v_{\tau})$ from $D$.\\
    %     Update accordingly the degrees of $u_{\nu}$, set $S$ and set $S'$.
    % }
    % \If{$v_{\tau} \in S$}{
    %     Remove $(Pred^w(v_\tau),v_\tau)$ from $D$.\\
    %     Update accordingly the degrees of $Pred^w(v_\tau)$, set $S$ and set $S'$.\\
    %     Remove $v_{\tau}$ from $D$.
    % }
}
\Return{$F^{ord}$}

\medskip
\Procedure{$\removearc(u_{\nu},v_{\tau})$}{
  Prepend $(u_{\nu},v_{\tau})$ to $F^{ord}$ and remove it from $D$.\\
  Update accordingly the degrees of $u_{\nu}$ and the sets $S,S'$.
}

\caption{Computing an arc-extend-respecting arc ordering of a space-time representation of a zero-acyclic temporal graph.}
\label{alg:sttds}
\end{algorithm}

This algorithm runs in linear time by maintaining $\delta_{out}^{c}(v_{\tau})$, $\delta_{out}^{w}(v_{\tau})$, and $N_{in}^{c}(v_{\tau})$ for each node $v_\tau$. This allows to maintain the set $S$ of nodes with out-degree zero, and the set $S'$ of nodes $v_\nu$ having no out-edges in $F^w$ and satisfying $\alpha_v>0$. A node in $S$ or $S'$ can then be selected in constant time. Each arc removal is performed with constant-time updates of out-degrees, in-neighbors, sets $S$ and $S'$. As each node is considered at most twice, the overall execution thus takes linear time.

Each time an arc is prepended to $F^{ord}$ by the algorithm, all arcs that arc-extend it must have already been removed and are thus already in $F^{ord}$. The resulting ordering $F^{ord}$ is thus arc-extend-respecting. We then obtain a half-extend-respecting ordering $E^{arr}$ of the temporal edges by removing waiting arcs from $F^{ord}$ and replacing each remaining connection arc $(u_\tau,v_\nu)$ by its corresponding temporal edge $(u,v,\tau,\nu-\tau)$. Note that $E^{arr}$ is also node-arrival sorted. This is due to the fact that for any node $v$ and any time labels $\tau$ associated to $v$, arcs entering the copy $v_\tau$ cannot be removed as long as it has an out-edge in $F^w$, and this out-edge is removed only when its next copy $v_\nu$ has out-degree zero. This implies that all arcs entering a copy $v_\mu$ with $\mu > \tau$ are prepended to $F^{ord}$ before all arcs entering $v_\tau$.

A similar algorithm can be symmetrically designed to obtain an ordering $E^{dep}$ which is half-extend-respecting and node-departure sorted. When producing a temporal edge in $E^{arr}$ and in $E^{dep}$, we can associate it to the index of its corresponding arc in $F^c$ so that we can easily construct pointers linking each edge in $E^{arr}$ to its copy in $E^{dep}$.

\medskip
  
On the other side, to prove Proposition~\ref{prop:eq}.B, we can use the following procedure. We  obtain for each $v \in V$, the lists $E^{dep}_v$ and $E^{arr}_v$ through bucket sorting. We merge these two lists into a single list $E^{event}_v$ for each node $v$. A linear scan of $E^{event}_v$ then produces all the sorted copies of $v$ in $W$, and associates to each edge having $v$ as head or tail the corresponding copy of $v$. A linear scan of all temporal edges then allows to construct $F^c$ in linear time. Finally, we construct $F^w$ by scanning $E^{event}_v$ for each node $v$.
\end{proof}

%\medskip
Interestingly, the above result implies that any zero-acyclic temporal graph admits a half-extend-respecting ordering. Conversely, the existence of a half-extend-respecting ordering obviously prevents the presence of zero-cycles and implies zero-acyclicity. We thus obtain the following statement.

\begin{proposition}\label{prop:zeroacyclic}
A temporal graph $G=(V,E,\alpha,\beta)$ is zero-acyclic if and only if there exists an half-extend-respecting ordering of its edges.
\end{proposition}

Note that for a given set $E$ of temporal edges, zero-acyclicity depends only  on which nodes $v$ have non-zero minimum waiting time $\alpha_v$ as our algorithm for computing an half-extend-respecting ordering depends only on this.

\medskip
Proposition~\ref{prop:eq} also implies that a half-extend-respecting doubly-sorted representation $(E^{dep},E^{arr})$ can be computed in linear time and space from any doubly-sorted representation of a zero-acyclic temporal graph by constructing its space-time representation as an intermediate step. Interestingly, our single-source minimum-cost walk algorithm can be adapted to take as input either any doubly-sorted representation or a space-time representation. We thus obtain the following corollary of Theorem~\ref{th:gal}.

\begin{corollary}\label{cor:gal}
Given either a doubly-sorted representation or a space-time representation of a zero-acyclic temporal graph $G=(V,E,\alpha,\beta)$ with cost structure $({C},\gamma,\oplus,\preceq)$ satisfying isotonicity, and a source node $s\in V$, the single-source all-reachable-edge minimum-cost problem can be solved in linear time and space.
\end{corollary}

\section{Handling zero travel-times}
\label{sec:zero}

We now show how our approach can be adapted to handle instances containing zero-cycles. Note that a temporal graph containing a zero-cycle fails to admit any half-extend-respecting ordering of its edges.
We will give an algorithm based on Algorithm \ref{alg:gal}, that solves the single-source all-reachable-edge minimum-cost problem in $O(|E| \log |V|)$ in this setting, where $E$ is the set of temporal edges and $V$ is the set of nodes. It still requires a doubly-sorted representation as input. The key point of this algorithm is to compute an ordering of the edges that can be handled correctly by Algorithm \ref{alg:gal} for a given source as long as the temporal graph satisfies the following property generalizing non-negative weights.

Consider a temporal graph $G=(V,E,\alpha,\beta)$ with a cost structure $\mathcal{C}=({C},\gamma,\oplus,\preceq)$. %, and let us denote $\lambda(e)$ the travel time of any temporal edge $e\in E$.
A cost $d\in C$ is said to be \emph{$\mathcal{C}$-non-negative} when it satisfies $c \preceq c \oplus d$ for all $c\in C$. This can be seen as a generalization of classical non-negativity.
We consider instances satisfying the following \emph{right-absorption} property (\emph{absorption} for short) which is a restriction to zero travel time edges of a property similarly considered in~\cite{Griff2010}: 
% \begin{equation}\tag{right absorption}\label{eq:absorption}  
%     \forall e \in E, \text{ if } \lambda(e) = 0 \implies \forall c \in C,\, c \preceq c \oplus \gamma(e). 
% \end{equation}
\begin{equation}\tag{absorption}\label{eq:absorption}
  \begin{split}
    \text{for any } e=(u,v,\tau,\lambda) \in E & \text{ such that } \lambda = 0 \text{ and } \alpha_u=\alpha_v=0, \\
    & \gamma(e) \text{ is $C$-non-negative, that is }
    c \preceq c \oplus \gamma(e) \text{ for all } c\in C.
  \end{split}
    %\text{for any } e=(u,v,\tau,\lambda) \in E \text{ such that } \lambda = 0 \text{ and } \alpha_u=0, \text{ for any } c \in C,\,  c \preceq c \oplus \gamma(e). 
\end{equation}
Notice that under absorption and isotonicity, there cannot exist any zero-walk $Q$ in $G$ such that $c \oplus \gamma_{Q} \prec c$, for any cost $c\in C$. Indeed, this absorption property captures a property similar to non-negativity of weights in classical shortest path computation. 
%\fb{Here I can just skip the problematic sentence of walks with no cycles.}
%In particular, this means that, in our problem, we can restrict to optimal walks $Q$ in $G$ that do not contain sub-sequences of temporal edges that are cyclic walks in $G_{\tau}$ for some $\tau$. \fb{rewrite handling source case}

Let us define the following property on a doubly sorted representation of a temporal graph. Let $G=(V,E,\alpha,\beta)$ be a temporal graph, $(E^{dep},E^{arr})$ be a doubly-sorted representation of $G$ and $s\in V$ be a source node. We say that $(E^{dep},E^{arr})$ is \emph{$s$-optimal-respecting} if for each $s$-reachable edge $e$ there exists a $(E^{dep},E^{arr})$-respected walk $Q$ from $s$ having last edge $e$ and with cost $c$, where $c$ is the minimum cost of any walk from $s$ ending with $e$. We can now state the following result.

% Let us define the following property on temporal edge subsets. Let $G=(V,E,\alpha,\beta)$ be a temporal graph, given a subset $E'$ of the temporal edges $E$, let us denote with $G'=(V,E',\alpha,\beta)$. For a given source $s \in V$ we say that $E'$ is $s$-optimal if for any $v \in V$, and any $sv$-walk $Q$ in $G$ with cost $c$, there exists an $sv$-walk $Q'$ in $G'$ with cost $c' \preceq c$. We can now state the following result.

% Let us define the following property on temporal edge orderings.  Given a temporal graph $G=(V,E,\alpha,\beta)$ and an ordered list $E^{ord}$ of the temporal edges, we say that a $uv$-walk $Q$ is \emph{well-ordered} in $E^{ord}$ if the edges of the walk appear in order in $E^{ord}$.
% %In particular, when $E^{ord}$ is walk-respecting, all walks are well-ordered in $E^{ord}$.
% For a given source node $s$, we say that $E^{ord}$ is \emph{$s$-optimal-walk-respecting} if for any $v \in V$ and any $s$-reachable edge $e\in A_v$ at $v$, there exists a minimum-cost $sv$-walk ending with $e$, that is well-ordered in $E^{ord}$. We can now state the following result.

\begin{theorem}\label{th:zero}
Given a doubly-sorted representation $(E^{arr},E^{dep})$ of a temporal graph $G=(V,E,\alpha,\beta)$ with cost structure $({C},\gamma,\oplus,\preceq)$ satisfying isotonicity and absorption, and a source node $s$,
the single-source all-reachable-edge minimum-cost problem can be solved in $O(|E| \log |V|)$ time and space.
\end{theorem}

%\lv{Sort $E^{arr}$}
To prove this, we design an algorithm that reorders edges with zero travel time and same arrival time in both $E^{arr}$ and $E^{dep}$, producing a different doubly-sorted representation $(\Bar{E}^{arr},\Bar{E}^{dep})$ while performing a linear scan like in  Algorithm~\ref{alg:gal}. The doubly-sorted representation obtained through the reordering is $s$-optimal-respecting.
%To prove this, we design an algorithm that consists in a linear scan of a reordered set $E^{arr}'$ of the temporal edges in $E^{arr}$ and operates as in Algorithm~\ref{alg:gal}. %with special care for sequences of edges with zero travel time and same arrival time.
%The ordered subset $E'$ which is scanned is $s$-optimal, half-extend-respecting and node-arrival sorted.

To identify efficiently such sequences of zero travel time edges, we first sort $E^{arr}$ according to non-decreasing arrival times as follows. We obtain through bucket sorting the lists $E_v^{arr}$ of edges with head $v$ ordered by non-decreasing arrival time for all nodes $v$, and then merge them back using a priority queue in $O(|E|\log |V|)$ time. 
We obtain the ordered list $A^{arr}$ of all the arrival times of the edges by scanning $E^{arr}$. We can now refine the ordering of the edges with same arrival time in the following way. By bucket sorting, we separate edges with same arrival time into four blocks and then merge them back together one after the other. The first block are those edges having positive travel time, the second the edges with zero travel time and positive minimum waiting-time at the tail, the third the edges with zero travel time and zero minimum waiting-time both at the source and the head, and finally the fourth are the edges with zero travel time and positive minimum waiting-time at the head. The reason for this separation is that an edge arriving at time $a$ can be extended by a zero-walk at time $a$ only if its head has zero minimum waiting-time, and similarly, a zero-walk at time $a$ can be extend by an edge departing at time $a$ only if its tail has zero minimum waiting-time. More precisely, the notation we use to indicate the four blocks is the following. We denote by $E^{arr}_{a,\lambda>0}$ the sub-array of edges of $E^{arr}$ with arrival time $a$ and with positive travel time.
Among the edges with arrival time $a$ and with zero travel time we distinguish the following three blocks. We denote by $E^{arr}_{a,\lambda= 0,\alpha_{tail}>0}$ the sub-array of edges $e=(u,v,\tau,0)$ such that $\alpha_u > 0$, by $E^{arr}_{a,\lambda= 0,\alpha=0}$ the sub-array of edges $e=(u,v,\tau,0)$ that $\alpha_u = \alpha_v = 0$, and finally by $E^{arr}_{a,\lambda= 0,\alpha_{head}>0}$ the sub-array of edges $e=(u,v,\tau,0)$ such that $\alpha_v > 0$. In particular we refer to $E^{arr}_{a,\lambda= 0,\alpha=0}$ as a \emph{zero-block}. Before scanning each zero-block, its edges are reordered as described next.

Our goal is to identify certain minimum-cost walks that contain edges in the zero-block, and preserve the ordering of their edges. We represent the edges of the zero-block through a weighted static graph; indeed the the time labels and the waiting constraints in this case play a marginal role, since all the edges in the zero-block have same departure time, zero travel time and no minimum waiting constraints on their head and tail. In particular, walks in the digraph correspond to walks in the temporal graph.
We first identify edges that terminate minimum-cost walks containing exactly one edge in the zero-block. The heads of such edges will serve as sources in the static graph, and are associated to the cost of the aforementioned walks. From these sources, with their initial associated starting costs, we run Dijkstra algorithm~\cite{Dijkstra59} and build a shortest path forest from them. We then reorder the edges of the zero-block in the following way: first the edges terminating minimum cost walks to the sources, then edges corresponding to arcs in the shortest path forest so as to preserve path order for all paths in the forest, and finally the remaining edges of the zero-block. Moreover, we will also partially reorder the edges in $E^{dep}$: among the edges with same departure time we will put first those edges that we identified preceding the sources.
%We generate a weighted static graph that takes into account the minimum cost of walks found so far and the edges in that block. 
%First of all we identify which are the first edges of the block that will appear in walks with minimum cost. The head of such edges will be a suitable set of sources to run Dijkstra algorithm~\cite{Dijkstra59} from and build a shortest path forest from these nodes. 
%We reorder the edges of the block in the following way: first the edges we identified preceding the sources, then edges according to the positions of their corresponding arcs in the shortest path forest so as to preserve path order for all paths in the forest, and finally the remaining edges of the block.
%We run Dijkstra algorithm~\cite{Dijkstra59} on this graph from an appropriate set of source nodes, build a shortest path forest from these nodes, and we reorder the edges of the block according to the positions of their corresponding arcs in the shortest path forest so as to preserve path order for all paths in the forest.  \fb{this previous sentence can be a bit misleading because we add also first the parent of the roots}
%Moreover, we will also partially reorder the edges in $E^{dep}$: among the edges with same departure time we will put first those edges that we identified preceding the sources.

The algorithm then scans the reordered zero-block and the following edges again as in Algorithm~\ref{alg:gal} up to the next zero-block. We will prove that the two reordered lists obtained are indeed a doubly-sorted representation which is $s$-optimal-respecting.

%We will prove that the reordered subset of edges scanned is $s$-optimal and that the resulting global ordering is half-extend-respecting and node-arrival sorted.
% For each $sv$ minimum cost walk, there exists an $sv$ walk with same arrival time and same cost that is well-ordered in the resulting global order of the edges. 
% This order of the edges will guarantee that they are scanned accordingly with the order they appear in the optimal walks sequences.
% More precisely, we denote by $E^{arr}_{a,\lambda>0}$ the sub-array of edges of $E^{arr}$ with arrival time $a$ and with positive travel time and by $E^{arr}_{a,\lambda=0}$ the edges with zero travel time.
% %
% We denote by $E^{arr}_{a,\lambda= 0,\alpha_{tail}>0}$ the sub-array of edges $e=(u,v,\tau,0)$ in $E^{arr}_{a,\lambda=0}$ such that $\alpha_u > 0$ and by $E^{arr}_{a,\lambda= 0,\alpha_{head}>0}$ the sub-array of edges $e=(u,v,\tau,0)$ in $E^{arr}_{a,\lambda=0}$ such that $\alpha_v > 0$. Moreover, we denote by $E^{arr}_{a,\lambda= 0,\alpha=0}$ the \emph{block} of edges $e=(u,v,\tau,0)$ in $E^{arr}_{a, \lambda=0}$ such that $\alpha_u = \alpha_v = 0$.
% %We denote by $E^{arr}_{a,\lambda= 0,\alpha=0}$ the block of edges of $E^{arr}$ with arrival time $a$, zero travel time and zero minimum waiting. Similarly, we denote by $E^{arr}_{a,\lambda= 0,\alpha>0}$ the sub-array of edges of $E^{arr}$ with arrival time $a$, zero travel time and non-zero minimum waiting. 
% Finally, we denote with $A^{arr}$ the ordered list of all arrival times of edges in $E^{arr}$. See Algorithm~\ref{alg:0del} for a formal description.

\begin{algorithm}[H]
\DontPrintSemicolon %otherise print the semicolon in \lIf
\Input{A doubly-sorted representation $(E^{dep},E^{arr})$ of a temporal graph $G$ with waiting-time constraints $(\alpha,\beta)$ and cost structure $({C}, \gamma,\oplus,\preceq)$ satisfying isotonocity and absorption, and a source node $s$.}
\Output{Minimum cost of an $sv$-walk for each node $v$ and for each $s$-reachable edge $e$ with head $v$.}

Sort $E^{arr}$ by non-decreasing arrival time.\\
Scan $E^{arr}$ to compute $A^{arr}$ and blocks $E^{arr}_{a,\lambda>0}$, $E^{arr}_{a,\lambda=0,\alpha_{tail}>0}$, $E^{arr}_{a,\lambda=0,\alpha=0}$, $E^{arr}_{a,\lambda=0,\alpha_{head}>0}$.\\ Rebuild $E^{arr}$ by concatenating these blocks for increasing $a\in A^{arr}$.\label{ling:ord}\\

% Sort $E^{arr}$ by non-decreasing arrival time.\\
% Scan $E^{arr}$ to compute the lists $A^{arr}$, $E^{arr}_{a,\lambda>0}$, $E^{arr}_{a,\lambda=0,\alpha_{tail}>0}$, $E^{arr}_{a,\lambda=0,\alpha=0}$, $E^{arr}_{a,\lambda=0,\alpha_{head}>0}$. \\
%\textcolor{blue}{Concatenate $E^{arr}_{a,\lambda>0}$, $E^{arr}_{a,\lambda=0,\alpha_{tail}>0}$, $E^{arr}_{a,\lambda=0,\alpha=0}$, $E^{arr}_{a,\lambda=0,\alpha_{head}>0}$ for $a\in A^{arr}$ into $E^{arr}$.}\\

Initialize variables as in Algorithm~\ref{alg:gal} (Lines \ref{ling:prebegpar} - \ref{ling:preendpar}).\\
% For each node $v$, generate the list $E^{dep}_v$ by bucket sorting $E^{dep}$. \label{line:prebeg}\\
% \For{each node $v$}{
%   Set $A_v:=\emptyset$. \Comment{List of pairs of arrival time and cost.}
%   Set $r_v:=0$. \Comment{Index of the last processed edge in $E_v^{dep}$.}
%   Set $\intervs{v}:=\emptyset$. \Comment{Doubly linked list of consecutive intervals of $E_v^{dep}$.}
% }
% Set all the edges in $E^{arr}$ as unmarked.\\
% Initialize the best walk cost $B[e]$ of each edge $e\in E^{arr}$ to $\bottom$.\\ 

\For{each arrival time $a$ in $A^{arr}$}{ \label{ling:begscan}
\lFor{$e\in E^{arr}_{a,\lambda>0}$}{scan $e$ as in Algorithm \ref{alg:gal} (line \ref{ling:begoutfor} - \ref{ling:endoutfor}).}
\lFor{$e\in E^{arr}_{a,\lambda=0, \alpha_{tail}>0}$}{scan $e$ as in Algorithm \ref{alg:gal} (line \ref{ling:begoutfor} - \ref{ling:endoutfor}).} 
$E^{ord} := Reorder(E^{arr}_{a,\lambda=0,\alpha=0},a)$.\\
Replace $E^{arr}_{a,\lambda=0,\alpha=0}$ by $E^{ord}$ in $E^{arr}$.\\
\lFor{$e\in E^{ord}$}{scan $e$ as in Algorithm \ref{alg:gal} (line \ref{ling:begoutfor} - \ref{ling:endoutfor}).} \lFor{$e\in E_{a,\lambda=0,\alpha_{head}>0}^{arr}$}{scan $e$ as in Algorithm \ref{alg:gal} (line \ref{ling:begoutfor} - \ref{ling:endoutfor}).} \label{ling:endscan}
}
\Return{the sets $(A'_{v})_{v\in V}$}
\caption{Adaption of Algorithm~\ref{alg:gal} to handle edges with zero travel time.}
\label{alg:0del}
\end{algorithm}

\begin{algorithm}[ht]
\Function{$\reorder(E^{arr}_{\tau,\lambda=0,\alpha=0},\tau)$}{

   \Comment{Identify nodes appearing in edges of $E^{arr}_{\tau,\lambda=0,\alpha=0}$:}
   Let $V'_{tail}:=\{ u : \exists (u,v,\tau,0)\in E^{arr}_{\tau,\lambda=0,\alpha=0}\}$ and $V'_{head}:=\{v : \exists (u,v,\tau,0)\in E^{arr}_{\tau,\lambda=0,\alpha=0}\}$.\label{linr:nodes}\\
    %   Let $V'_{tail}=\{ u : \exists (u,v,\tau,0)\in E^{arr}_{\tau,\lambda=0}\}$.\\
    %   Let $V'_{head}=\{ v : \exists (u,v,\tau,0)\in E^{arr}_{\tau,\lambda=0}\}$.\\
    %   Let $V' := V'_{tail} \cup V'_{head}$\\
   
    \Comment{Set $B_{tail}[u]$ to the best extendable cost of edges from $u$ in the zero-block.}
    \For{each node $u\in V'_{tail}$}{\label{linr:btail-beg}
    Let $I=(l,r,c,e)$ be the first interval in $\intervs{u}$ containing an edge $e$ with $dep(e) \geq \tau$.\label{linr:findint}\\
    Let $i$ be the index of the first edge $e$ in $I$ such that $dep(e)=\tau$.\label{linr:findindex}\\
    Set $p[u]:=i$ and $B_{tail}[u]:=c$.\label{linr:btail-end}
    %Initialize $B'[u]=B_{tail}[u]:=\bottom$.\\
    %Let $l[u]$ bet the index of the first edge $e$ in $E_u^{dep}$ such that $dep(e) \geq \tau$.
    }
    \Comment{Set $B'[v]$ to the minimum cost of an $sv$-walk ending with exactly one edge in the zero-block.}
    \lFor{each node $v\in V'_{head}$}{\label{linr:beg-reach-cost}
    Initialize $B'[v] = \bottom$.}
    Initialize a set $S:=\emptyset$ of reachable nodes.\\
    \For{each edge $e=(u,v,\tau,0)\in E^{arr}_{\tau,\lambda=0,\alpha=0}$}{
      %\lIf{$e$ has an associated cost $c'$ according to $\intervs{u}$}{$B_{tail}[u] := c'$}
      %\lElse{$B_{tail}[u] := \bottom$}
      \If{$u=s$ or $B_{tail}[u]\neq \bottom$}{ \label{linr:if}
        $S:=S\cup\{v\}$ \label{linr:set} \Comment{$v$ can be reached from $s$.}
        \lIf{$u=s$ and ($B_{tail}[u]=\bottom$ or $\gamma(e)\prec B_{tail}[u]\oplus\gamma(e)$)}{$c:=\gamma(e)$} \label{linr:c1}
        \lElse{$c:=B_{tail}[u]\oplus\gamma(e)$} \label{linr:c2}
        %\Comment{Minimum cost for reaching $v$ through $e$.}
        \If{$B'[v]=\bottom$ or $c\prec B'[v]$}{ \label{linr:ifbprime}
          $B'[v]:=c$ \label{linr:bprime}\\ %\Comment{Minimum cost of an $sv$-walk ending with exactly one edge in the zero-block.}
          $P'[v]:=e$ \label{linr:endS} \Comment{Last edge of a walk defining $B'[v]$.}
        }
      }
    }

    \Comment{Define a weighted static digraph $D=(V'_{tail}\cup V'_{head},A')$.}
   $A':=\{(u,v,\gamma(e)) : e=(u,v,\tau,0)\in E^{arr}_{\tau,\lambda=0,\alpha=0}\}$ \label{linr:begdigr}\\
   Construct a weighted digraph $D$ with vertex set $V'_{tail} \cup V'_{head}$ and arc set $A'$.\label{linr:enddigr}\\
    
    \Comment{Reorder the sets of edges $E^{arr}_{\tau,\lambda=0,\alpha=0}$ and $E^{dep}$.} 
    Compute a minimum-cost forest $F:=\dijkstra(D, S, B')$.\label{linr:dijkstra}\label{linr:reorder-beg}\\
    $E^{ord}:=\emptyset$\\
    \For{$v\in S$ such that $F[v]=\bottom$}{
      Let $e=(u,v,\tau,0):=P'[v]$ and append $e$ to $E^{ord}$.\\
      Swap in $E_u^{dep}$ edge $e$ with the edge that has index $p[u]$ and set
      $p[u] := p[u] + 1$.\label{linr:swap}\label{linr:dep}\\
      Compute a BFS ordering $F_v$ of the arcs of the tree rooted at $v$ in $F$. \label{linr:begbfs} \\
      For each arc $(u,v,c)\in F_v$, append the associated edge $(u,v,\tau,0)$ to $E^{ord}$.   \label{linr:endbfs}
    }
    Append to $E^{ord}$ the remaining edges in $E^{arr}_{\tau,\lambda=0,\alpha=0}$. \label{linr:remain}\label{linr:reorder-end} \\
    \Return{$E^{ord}$.}
}

\caption{Reorder the temporal edges in a zero-block so that a sufficiently big set of minimum-cost walk are $(E^{dep},E^{arr})$-respected.}
\label{alg:order}
\end{algorithm}

\begin{algorithm}[ht]
  \Function{$\dijkstra((V',A'), S, B')$}{
    \lFor{$v\in V'$}{initialize a key $K[v]:=\bottom$ and a parent pointer $F[v]:=\bottom$.}
    Initialize a Fibonacci heap $H:=\emptyset$.\\
    \lFor{$v\in S$}{add $v$ to $H$ with key $K[v]:=B'[v]$.} \label{lind:sources}
    \While{$H\not=\emptyset$}{
      $u:=\popmin(H)$ \label{lind:pop}\\
      \For{$(u,v,c)\in A'$}{
        \If{$K[v]=\bottom$ or $K[u]\oplus c\prec K[v]$}{ \label{lind:upd1}
          $K[v]:=K[u]\oplus c$\\
          $F[v]:=(u,v,c)$\\
          \lIf{$v\notin H$}{add $v$ to $H$ with key $K[v]$}\label{lind:add}
          \lElse{decrease key of $v$ in $H$ to $K[v]$.} \label{lind:upd2}\label{lind:decr}
        }
      }
    }
    \return{$F$.}
}
\caption{Algebraic version of Dijkstra from a set $S$ of sources with initial costs $B'$ in a digraph $(V',A')$.}
\label{alg:dijkstra}
\end{algorithm}

%\lv{Check line numbers bellow!}

%Let us make three observations on the behaviour of the $\reorder()$ function.
We now describe more precisely the call to $\reorder(E^{arr}_{\tau,\lambda=0,\alpha=0},\tau)$ which is formally described in Algorithm~\ref{alg:order}.
We first identify the set $V' = V'_{tail} \cup V'_{head}$ of nodes appearing as tail or head of edges in the zero-block $E^{arr}_{\tau,\lambda=0,\alpha=0}$ at Line~\ref{linr:nodes}. 
We notice that edges in the zero-block sharing the same tail have the same associated cost according to Algorithm~\ref{alg:gal}: as they have same departure time $\tau$ they extend the same set of walks and they belong to the same interval. This allows us to assign to each node $u$ in $V'_{tail}$  the cost $B_{tail}[u]$ associated to edges of the zero-block with tail $u$ at the moment of the call to $\reorder()$, namely just before scanning the edges of the zero-block (see Lines~\ref{linr:btail-beg}-\ref{linr:btail-end}). 
%Then we scan the edges of the block $E^{arr}_{\tau,\lambda=0,\alpha=0}$. For each edge $e$, we assign to its tail $u$ the cost $B_{tail}[u]$ associated to edge $e$ in the moment of the call to $\reorder$, namely before the edges of the block will be scanned.
%We then retrieve, for each one the nodes $u$ in $V'_{tail}$, the cost $B_{tail}[u]$ that is associated to them in the moment of the call to $\reorder$, namely before the edges of the block will be scanned. 
This is the minimum cost of any $su$-walk composed of edges scanned so far, excluding in particular walks containing edges in $E^{arr}_{\tau,\lambda=0,\alpha=0}$, that can be extended with edges departing at time $\tau$.
We also associate to each $u$ in $V'_{tail}$ the index of the first edge in $E_u^{dep}$ that has departure time $\tau$ at Lines~\ref{linr:findindex}-\ref{linr:btail-end}. Notice that this is not necessarily an edge in the zero-block.

Next, we compute for each edge $e$ in the block the minimum cost to reach its head considering both the case when $e$ extends a walk from the source, and the case when starts itself a walk from the source at Lines~\ref{linr:beg-reach-cost}-\ref{linr:endS}. In particular, by keeping the minimum among the edges with same head, we compute for each node $v$ in $V'_{head}$ the minimum cost $B'[v]$ of $sv$-walks that contain one, and only one, edge of the zero-block and terminate with it. We store in $P'[v]$ the last edge of such a walk with minimum cost at Line~\ref{linr:endS}. All nodes $v\in V'_{head}$ that can be reached by such an $sv$-walk are stored in a set of sources $S$ at Line~\ref{linr:set}.
%Similarly as in Algorithm~\ref{alg:gal}, when $e$ ends an $sv$-walk, we compute the minimum cost $c$ of such a walk at Lines~\ref{linr:if}-\ref{linr:c2}. We then update the minimum cost $B'[v]$ of an $sv$-walk reaching $v$ through an edge in $e\in E^{arr}_{\tau,\lambda=0}$, and store in $P'[v]$ the corresponding edge $e$ at Lines~\ref{linr:ifbprime}-\ref{linr:endS}. All nodes $v\in V'$ that can be reached by such an $sv$-walk are stored in a set $S$ at Line~\ref{linr:set}.

We then build a weighted directed graph $D$ with node set $V'_{tail}\cup V'_{head}$ and arc set $A'$ which is defined by associating an arc $(u,v,\gamma(e))$ to each edge $e = (u,v,\tau,0) \in E^{arr}_{\tau,\lambda=0,\alpha=0}$ at Lines~\ref{linr:begdigr}-\ref{linr:enddigr}. 

Finally, we reorder the edges in $E^{arr}_{\tau,\lambda=0,\alpha=0}$ and $E^{dep}$ at Lines~\ref{linr:reorder-beg}-\ref{linr:reorder-end} based on a minimum-cost forest $F$ from $S$ in $D$ which is computed at Line~\ref{linr:dijkstra} through an algebraic version of Dijkstra algorithm which is described in detail in the next paragraph. Note that this Dijkstra computation uses the fact that all arcs of $D$ have $\mathcal{C}$-non-negative costs. The forest $F$ is given by parent pointers where $F[v]$ provides for each node $v$ an arc allowing to reach $v$ from $S$ through a path with minimum cost in $D$. The pointer $F[v]$ has value $\bottom$ if $v$ is the root of a tree or if it is not reachable from $S$ in $D$. For each tree $T$ rooted at a node $v$ in $F$, we use a BFS ordering of $T$ to make sure at Lines~\ref{linr:begbfs}-\ref{linr:endbfs} that any path in $T$ corresponds to a walk whose edges appear in order in $E^{ord}$, and after $P'[v]$.
%that is a sequence of half-respecting-edges in $E^{ord}_{\tau,\lambda=0,\alpha=0}$ that appear after $P'[v]$. 
% \lv{I would remove this (over explanation): This is why $P'[v]$ is added before the edges corresponding to the tree rooted at $v$ as edge $P'[v]$ is the last edge of an $sv$-walk with arrival time $\tau$ that can be extended with such a walk as arcs in $D$ correspond to edges $(x,y,\tau,0)$ such that $\alpha_x=0$ and $\alpha_y=0$ .}
Note that such a walk extends edge $P'[v]$ which is added first.
Note also that the edge $P'[v]$ is added only once to $E^{ord}$: since it has head $v$ which is a root, it cannot be associated to an arc of $F$ (we have $F[v]=\bottom$). Additionally, we move edges $P'[v]$ with tail $u$ in $E_u^{dep}$ at Line~\ref{linr:dep} so that they appear in $E_u^{dep}$ before other edges of the zero-block.
%we swap edge $P'[v]$ with the first edge in $E_u^{dep}$ that has departure time $\tau$ and has not be swapped already during the procedure, where $u$ is the tail of $P'[v]$. 
This guarantees that any walk resulting from concatenating an edge $P'[v]$ and a walk corresponding to a path in the tree rooted at $v$ will be $(\Bar{E}^{arr},\Bar{E}^{dep})$-respected.
%at the end of the execution, all the edges $P'[v]$ for $v\in V'$ that share the same tail $u$, will appear as the first ones in $E_u^{dep}$ among the edges with departure time $\tau$. %We also perform the same swap between the two edges in $E^{dep}$, in order to guarantee consistency between the lists $E^{dep}$ and $E_u^{dep}$.
%Finally, we drop edges in $E^{arr}_{\tau,\lambda=0,\alpha=0}$ corresponding to arc in $D$ that do not appear in the forest $F$ or is not $P'[v]$ for a root $v$. Indeed, these edges cannot enable a walk with a strictly lower cost than the walks composed by the edges already added. 
%
%for each edge $e = (u,v,\tau,0) \in E^{arr}_{\tau,\lambda=0}$, arc $(u^1,v^2,\gamma(e))$ and $(u^2,v^2,\gamma(e))$ cannot be both in $T$: they have  same head $v^2$ which has only one entering arc in $T$ as $T$ is an out-branching from $s^1$. This means that $E^T$ corresponds to a list of edges in $E^{arr}_{\tau,\lambda=0}$ without any repetition.

For the sake of completeness, we also include an algebraic version of Dijkstra algorithm as detailed in Algorithm~\ref{alg:dijkstra}. It is similar to the version of~\cite{Griff2010} with slightly different hypothesis and the mild generalization of computing minimum-cost paths in a weighted digraph $(V',A')$ from a set $S$ of sources where each source $v\in S$ is associated to an initial cost $B'[v]$. More precisely, any walk $W=\langle (u_1,v_1,c_1),\ldots, (u_k,v_k,c_k)\rangle$ of $k\ge 1$ arcs from a source $u_1\in S$ in the digraph is associated to a cost $\gamma^D(B',W)=(\cdots(B'[u_1]\oplus c_1)\cdots\oplus c_{k-1})\oplus c_k$. A node $v$ is said to be \emph{reachable} from $S$ in $(V',A')$ if there exists a walk from a node $u\in S$ to $v$.
We consider that all nodes $v\in S$ are reachable through an empty path with cost $B'[v]$. A \emph{minimum-cost walk} from $S$ to $v\in V'$ is defined as a walk $W$ from any node $u\in S$ to $v$ such that $\gamma^D(B',W)\preceq \gamma^D(B',W')$ for any walk $W'$ from $u'\in S$ to $v$. The following elementary lemma allows us to focus on paths rather than walks.

\begin{lemma}\label{lem:paths}
    Let $W$ be a walk in $D$ from $u\in S$ to $v\in V'$, then there exists a path $P$ in $D$ from $u$ to $v$ such that $\gamma^D(B',P) \preceq \gamma^D(B',W)$. 
\end{lemma}
\begin{proof}
We show that we can iteratively remove any cylce from $W$ without increasing the cost of the walk. Let us decompose $W$ into $W_1.W_c.W_2$, where $W_c$ is a cycle. Then we have $\gamma^D(B',W_1) \preceq \gamma^D(B',W_1.W_c)$ by absorption. We then obtain $\gamma^D(B',W_1.W_2) \preceq \gamma^D(B',W_1.W_c.W_2)=\gamma^D(B',W)$ by isotonocity.
\end{proof}

%Notice that, due to to the non-negativity of the costs of the arcs and the isotonicty property, for each minimum-cost walk from $S$ to $v$ there exists a path from $S$ to $v$ with the same cost, as it possible to remove any loops without increasing the cost of the walk\fb{ I dont think we need to be more formal}. 

%If $V''$ denotes the set of nodes that are reachable from $S$ in $D$, 
We define a \emph{minimum-cost forest} $F$ from $S$ as a union of minimum-cost paths from $S$ to all reachable nodes from $S$, that forms a forest, that is where each node $v$ has at most one entering arc $F[v]$. Such a forest can be computed through Dijkstra algorithm. Recall that it consists in visiting nodes according to a non-decreasing cost order. More precisely, each node $v$ is associated to a key $K[v]$ storing the minimum-cost of a path reaching $v$ from $S$ that has been identified so far. Initially, only nodes in $S$ have a defined key which is initialized according to $B'$, see Line~\ref{lind:sources} in Algorithm~\ref{alg:dijkstra}. The next node $u$ to visit, that is a node with minimum key, can be found efficiently through a Fibonacci heap $H$ at Line~\ref{lind:pop}. We can then update the keys of each out-neighbor $v$ of $u$ according to Lines~\ref{lind:upd1}-\ref{lind:upd2}.
The correctness of Algorithm~\ref{alg:dijkstra} follows from the following lemma.

\begin{lemma}\label{lem:dijkstra}
  Given a weighted digraph $(V',A')$ and a cost structure $({C},\gamma,\oplus,\preceq)$ satisfying isotonicity, suppose that for every arc $(u,v,c)\in A'$ the cost $c$ is $C$-non-negative. Given a set $S\subseteq V'$ associated with initial costs $B'[v]\in C$ for $v\in S$, Algorithm~\ref{alg:dijkstra} returns a minimum-cost forest $F$ from $S$.
\end{lemma}

Before proving this lemma, note that $C$-non-negativity is not required for initial costs $B'[v]$ of nodes $v\in S$.
Despite its algebraic abstraction, the proof is nevertheless similar to the one found in algorithm textbooks~\cite{CormenLRS2001} for the classical version.  

\begin{proof}
  As usual with Dijkstra algorithm, we can prove by induction that the nodes are popped from the heap $H$ at Line~\ref{lind:pop} by non-decreasing order of keys. The reason is that all nodes $v$ remaining in $H$ when we pop $u$ have a key $K[v]$ satisfying $K[u]\preceq K[v]$ by the correctness of the heap operations which rely on the fact that $\preceq$ is a total order. Second, each node $v$ added to the heap at Line~\ref{lind:add}, or whose key is decreased at Line~\ref{lind:decr}, has key $K[v]=K[u]\oplus c$ and we have $K[u]\preceq K[u]\oplus c=K[v]$ as $c$ is $C$-non-negative. This non-decreasing order of popped keys together with the $C$-non-negativity of arc costs also imply that once a node has been popped, it is never re-inserted in the heap later. As the parent pointer $F[v]$ of a node always correspond to an arc $(u,v,c)$ such that $u$ has been popped before $v$, $F$ cannot induce any cycle and is indeed a forest. Moreover, following recursively the pointer $F[u]$ as long as $F[u]\not=\bottom$, we obtain a path $P_F^v$ with nodes ordered according to popping order. Note that the first node of $P_F^v$ must have been inserted in $H$ initially and is thus in $S$. The cost $\gamma^D(B',P_F^v)$ of $P_F^v$ is thus defined and it equals the value of $K[v]$ when $v$ is popped (this directly results from the mutual updates of $K[v]$ and $F[v]$). 
  
  Suppose for the sake of contradiction that there exist nodes $v\in V'$ which are reachable from $S$ and such that $F$ does not contain a minimum-cost path from $S$ to $v$. Without loss of generality, we can choose such a node $v$ so that no other such node is popped from $H$ before $v$. Consider a minimum cost path $P$ from $S$ to $v$. Either $P$ is non-empty and we let $u\in S$ denote the tail of its first arc, or we have $v\in S$ and no path from $S$ to $v$ has cost less than $B'[v]$, in which case we set $u:=v$. As $u\in S$ implies that $u$ is initially added to $H$, it must be popped at some point.
  
  First assume that $P$ is empty. We then have $v=u\in S$, and $v$ has been added to $H$, implying that $v$ is popped at some point. The path $P_F^v$ from $S$ to $v$ in $F$ has cost $\gamma^D(B',P_F^v)=K[v]$. As the key of $u=v$ can only decrease, we get $K[v]\preceq B'[v]$ and thus $\gamma^D(B',P_F^v)=B'[v]$ as no path from $S$ to $v$ has cost less than $B'[v]$ in that case. This is in contradiction with our hypothesis on $v$. 
  
  From now on, we assume that $P$ is non-empty.
  Suppose additionally that $v$ is popped before $u$. This implies that the path $P_F^v$ from $S$ to $v$ in $F$ has cost $\gamma^D(B',P_F^v)=K[v]\preceq K[u]$. As the key of $u$ can only decrease, we have $K[u]\preceq B'[u]$. Since the arcs of $P$ have $C$-non-negative costs, we have $B'[u]\preceq \gamma^D(B',P)$ and we get $\gamma^D(B',P_F^v)\preceq \gamma^D(B',P)$, contradicting again the choice of $v$.
 
  Otherwise, we can consider the last node $u'$ in $P$ such that all nodes from $u$ to $u'$ in $P$ have been popped before $v$. Let $(u',v',c)$ be the arc following $u'$ in $P$. The update of $K[v']$ according to arc $(u',v',c)$ at Lines~\ref{lind:upd1}-\ref{lind:upd2} implies $K[v']\preceq K[u']\oplus c$ and $v'\in H$. In particular, $v'$ will be popped at some point. Our choice of $v$ implies $\gamma^D(B',P_F^{u'})\preceq \gamma^D(B',P[u:u'])$ where $P[u:u']$ denotes the subpath of $P$ from $u$ to $u'$. As discussed previously, we have $K[u']=\gamma^D(B',P_F^{u'})$ when $u'$ is popped, and we thus get $K[v']\preceq \gamma^D(B',P[u:u'])\oplus c=\gamma^D(B',P[u:v'])$ by isotonicity. In the case $v'=v$, we thus get $\gamma^D(B',P_F^{v})\preceq \gamma^D(B',P)$. In the case $v'\not=v$, $v$ is popped before $v'$ according to the choice of $u'$, implying $\gamma^D(B',P_F^{v})=K[v]\preceq K[v']\preceq \gamma^D(B',P[u:v'])\preceq \gamma^D(B',P)$ as the arcs of $P[v':v]$ have $C$-non-negative costs. In all cases, we get $\gamma^D(B',P_F^v)\preceq \gamma^D(B',P)$, in contradiction with our hypothesis on $v$.
\end{proof}

We now prove Theorem~\ref{th:zero}.

\begin{proof}[Proof of \Cref{th:zero}]

Our proof mainly relies on the correctness and the complexity of Algorithm~\ref{alg:0del} thanks to the following observation.

\begin{claim}\label{cla:algequiv}
The execution of Algorithm~\ref{alg:0del} corresponds to an execution of Algorithm~\ref{alg:gal} with input $(\Bar{E}^{dep},\Bar{E}^{arr})$ where $\Bar{E}^{dep}$ and $\Bar{E}^{arr}$ are the orderings resulting from the calls to Algorithm~\ref{alg:order}.
\end{claim}

Recall that,
after line~\ref{ling:ord}, the list $E^{arr}$ is a sequence of blocks $E^{arr}_{a,\lambda>0}$, $E^{arr}_{a,\lambda=0,\alpha_{tail}>0}$, $E^{arr}_{a,\lambda=0,\alpha=0}$, $E^{arr}_{a,\lambda=0,\alpha_{head}>0}$ by increasing value of $a$, and the ordering $\Bar{E}^{arr}$ is obtained from it by replacing each zero-block $E^{arr}_{a,\lambda=0,\alpha=0}$ with the local order $E^{ord}$ computed through the call to $\reorder(E^{arr}_{a,\lambda=0,\alpha=0})$. The ordering $\Bar{E}^{dep}$ is obtained as a concatenation of the lists $\Bar{E}_u^{dep}$, where $\Bar{E}_u^{dep}$ is the list obtained from $E_u^{dep}$ through the swaps made at Line~\ref{linr:swap} by Algorithm~\ref{alg:order}. 

We prove Claim~\ref{cla:algequiv} through the following observations. The zero-blocks of edges $E^{arr}_{\tau,\lambda=0,\alpha=0}$ are reordered before any of its edges has been scanned, and edges are indeed scanned according to the ordering of $\Bar{E}^{arr}$. Concerning $\Bar{E}^{dep}$, we first note that none of the edges in a zero-block $E^{arr}_{\tau,\lambda=0,\alpha=0}$ has been finalized when $\reorder(E^{arr}_{\tau,\lambda=0,\alpha=0})$ is called for the following two reasons.  First, all edges with head $u$ scanned so far have arrival time at most $\tau$, and since $\alpha_u=0$ edges with departure time greater or equal to $\tau$ in $E_u^{dep}$ were not considered by any call to $\processcosts()$ at Line~\ref{ling:callonv} of Algorithm~\ref{alg:gal}. Second, $E^{arr}_{a,\lambda=0,\alpha_{tail}>0}$ does not contain any edge with tail $u$ since $\alpha_u=0$, and all edges from $u$ that have been scanned so far have departure time less than $\tau$. They thus appear before edges from $u$ in $E^{arr}_{\tau,\lambda=0,\alpha=0}$ since $E_u^{dep}$ is sorted by non-decreasing departure time, and these edges have not been finalized by any call at Line~\ref{ling:callonu} of Algorithm~\ref{alg:gal} either. 
%And no edge from $u$ in  appear before them in $E_u^{dep}$ and no. This is because $E_u^{dep}$ is sorted by non-decreasing departure time and such edges belong to a block preceding $E^{arr}_{\tau,\lambda=0,\alpha=0}$ in $E^{arr}$.
Finally, the swaps in $E_u^{dep}$ concern edges with same departure time $\tau$ and same tail $u$. This means that they can extend exactly the same set of walks and thus belong to the same interval in $\intervs{u}$. We thus have exactly the same intervals in $\intervs{u}$ for each node $u$ as if the algorithm had been run with input $(\Bar{E}^{dep},\Bar{E}^{arr})$ from the beginning. As the subsequent processing occurs with edges ordered according to $(\Bar{E}^{dep},\Bar{E}^{arr})$ until the next-zero block, this concludes the proof of Claim~\ref{cla:algequiv}

\bigskip
\noindent
\textit{Correctness.}

The core of the proof of correctness consists in showing that the reordered lists $(\Bar{E}^{dep},\Bar{E}^{arr})$ are an $s$-optimal-respecting doubly-sorted representation as we can then conclude by Proposition~\ref{prop:gal}.

First note, that the lists $\Bar{E}^{dep}$ and $\Bar{E}^{arr}$ are still node departure and node arrival sorted respectively. The list $\Bar{E}^{dep}$ is node departure sorted, as it is a reordering of $E^{dep}$ obtained by swapping edges with same tail and departure time. On the other side, the reordering of $E^{arr}$ concerns only sets of edges within the same zero-block which all have same arrival time, thus $\Bar{E}^{arr}$ is still ordered by non-decreasing arrival time of the edges. 
%The lists $\Bar{E}^{dep}$ and $\Bar{E}^{arr}$ are still respectively node departure and node arrival sorted. The reordering of $E^{dep}$ is obtained by swapping edges with same tail and departure time. On the other side, the reordering of $E^{arr}$ concerns only sets of edges within the same block, thus $\Bar{E}^{arr}$ is still ordered by non-decreasing arrival time of the edges.

The rest of the correctness part is dedicated to proving that $(\Bar{E}^{dep},\Bar{E}^{arr})$ is $s$-optimal-respecting. Suppose for the sake of contradiction that there exists an $s$-reachable edge $e$ such that there exists no minimum cost walk from $s$ ending with $e$ that is $(\Bar{E}^{dep},\Bar{E}^{arr})$-respected. Without loss of generality, let $e$ be the first edge in $\Bar{E}^{arr}$ such that this happens, and let $Q=\langle e_1, \dots, e_k \rangle$ be a walk with minimum cost among the walks from $s$ ending with $e=e_k=(u_k,v_k,\tau_k,\lambda_k)$. As our assumption implies that $Q$ itself is not $(\Bar{E}^{dep},\Bar{E}^{arr})$-respected, it must have at least two edges, we thus assume $k\ge 2$.
\paragraph{Note.} In the following, whenever we have to verify if $e <_{\Bar{E}^{arr}} f$, for some edges $e$ and $f$, it is sufficient to show that the block containing $e$ precedes the block that contains $f$ according to the ordering of $E^{arr}$ computed at Line~\ref{ling:ord}. This comes from the fact that $\Bar{E}^{arr}$ follows the same sequence of blocks and differs only by swaps within the zero-blocks.

\medskip

\paragraph{Case A: edge $e_{k-1}$ does not belong to a zero-block or $arr(e_k-1) < dep(e_k)$.} 
%
%\fb{We cannot start with the case $e_{k-1} <_{\Bar{E}^{arr}} e_k$, as to prove $Q'.e_k$ is $(\Bar{E}^{dep},\Bar{E}^{arr})$-respected we exploit $e_{k-1}$ being not in the zero-block (I think I tried starting the proof this way, but I couldn't make it work). But I don't think this is a problem.}
%

Let us first consider the case in which $e_{k-1}$ does not belong to a zero-block. This implies $e_{k-1} <_{\Bar{E}^{arr}} e_k$. The reason is that we have $arr(e_{k-1})\le dep(e_k)$ since $e_k$ extends $e_{k-1}$ and $\Bar{E}^{arr}$ is ordered by non-decreasing arrival time. Hence, $e_{k} <_{\Bar{E}^{arr}} e_{k-1}$ would imply $arr(e_{k})\le arr(e_{k-1})$ and thus $arr(e_{k-1})=dep(e_k)=arr(e_k)$. As $e_k$ extends $e_{k-1}$, we then must have $\alpha_{u_k}=0$. If $e_{k-1}$ does not belong to a zero-block we then have either $\lambda_{k-1}>0$ or $\alpha_{u_{k-1}}>0$, and thus $e_{k-1}$ belongs to a block before the block of $e_k$ according to the ordering of $E^{arr}$ computed at Line~\ref{ling:ord}.
%\lv{Begin with this case and then exclude it as it implies (justification needed) that $e_k$ and $e_{k-1}$ are both in the same zero-block.}
As $e_{k-1} <_{\Bar{E}^{arr}} e_k$, our choice of $e=e_k$ implies that
there exists a walk $Q'$ from $s$ ending with $e_{k-1}$ that is $(\Bar{E}^{dep},\Bar{E}^{arr})$-respected and has minimum cost among the walks from $s$ ending with $e_{k-1}$. Because of isotonocity we obtain that $\gamma(Q'.e_{k}) \preceq \gamma(Q)$, and proving that $Q'.e_k$ is $(\Bar{E}^{dep},\Bar{E}^{arr})$-respected would raise a contradiction. Since $Q'$ is $(\Bar{E}^{dep},\Bar{E}^{arr})$-respected, we just need to check that for each edge $e'$ with tail $u_k$ such that $e_k \leq_{\Bar{E}^{dep}} e'$, then $e_{k-1} <_{\Bar{E}^{arr}} e'$. Due to the fact that $e_k$ extends $e_{k-1}$ and $e_k \leq_{\Bar{E}^{dep}} e'$, we have $arr(e_{k-1}) \leq dep(e_k) \leq dep(e') \leq arr(e')$. If $arr(e_{k-1}) < arr(e')$ we can conclude that $e_{k-1} <_{\Bar{E}^{arr}} e'$, since $\Bar{E}^{arr}$ is non-decreasing arrival time sorted. Otherwise, $arr(e_{k-1}) = arr(e')$ implies that the travel time of $e'$ is zero and that $\alpha_{v_{k-1}} = 0$ as wet get $arr(e_{k-1}) = dep(e_k)$ and $e_k$ extends $e_{k-1}$. 
If $e_{k-1}$ has positive travel time, we can again conclude, as $e'$ has zero travel time and the two edges have same arrival time, the block of $e_{k-1}$ precedes the block of $e'$. Finally, suppose that $e_{k-1}$ has also zero travel time. 
As $\alpha_{v_{k-1}} = 0$ and $e_{k-1}$ does not belong to a zero-block, we must have $\alpha_{u_{k-1}} > 0$. On the other hand,
$v_{k-1}$ is the tail of $e'$ and we have $\alpha_{v_{k-1}} = 0$. Also in this case  the block of $e_{k-1}$ precedes the block of $e'$.

We consider now the case in which $arr(e_{k-1}) < dep(e_k)$. This implies $arr(e_{k-1}) < arr(e_k)$, and thus $e_{k-1} <_{\Bar{E}^{arr}} e_k$. We can choose $Q'$ as above: a walk from $s$ ending with $e_{k-1}$ that is $(\Bar{E}^{dep},\Bar{E}^{arr})$-respected and has minimum cost among the walks from $s$ ending with $e_{k-1}$. If we prove that $Q'.e_k$ is $(\Bar{E}^{dep},\Bar{E}^{arr})$-respected we can conclude by isotonicity. In particular, we just need to check that for each edge $e'$ with tail $u_k$ such that $e_k \leq_{\Bar{E}^{dep}} e'$, then $e_{k-1} <_{\Bar{E}^{arr}} e'$. Due to the fact that $arr(e_{k-1}) < dep(e_k)$ and $e_k \leq_{\Bar{E}^{dep}} e'$,  we have $arr(e_{k-1}) < dep(e_k) \leq dep(e') \leq arr(e')$. As $\Bar{E}^{arr}$ is ordered by non-decreasing arrival time we obtain $e_{k-1} <_{\Bar{E}^{arr}} e'$.

\medskip
\paragraph{Case B: edge $e_{k-1}$ belongs to a zero-block.}
We can now focus on the case where $e_{k-1}$ belongs to a zero-block, namely $e_{k-1}=(u_{k-1},v_{k-1},\tau,0) \in E^{arr}_{\tau,\lambda=0,\alpha=0}$ and $\alpha_{u_{k-1}}=\alpha_{v_{k-1}}=0$.
%\fb{If $e_{k-1}$ is not in a zero-block $->$ ... ok. We can now focus on the case $e_{k-1}=(u_{k-1},v_{k-1},\tau,0) \in E^{arr}_{\tau,\lambda=0,\alpha=0}$. }
Consider the call to $\reorder(E^{arr}_{\tau,\lambda=0,\alpha=0})$ from Algorithm~\ref{alg:0del}. Let $D=(V'_{tail}\cup V'_{head},A',A')$ be the digraph constructed at Lines~\ref{linr:begdigr}-\ref{linr:enddigr} in Algorithm~\ref{alg:order}, and let $S\subseteq V'_{head},A'$ be the set of source nodes computed according to Line~\ref{linr:set}.
The proof of correctness follows from the last of the three following claims.

\begin{claim}\label{cla:pathwalk}
Any path $P$ in $D$ from a node in $S$ corresponds to a walk $Q_P$ in $G$ with cost $\gamma^D(B',P)$ arriving at time $\tau$. %\lv{Using $\gamma^D(B',P)$ would be more clear than $\gammaB{P}$!}
\end{claim}

The reason is twofold. First, each node $v\in S$ is associated to the cost $B'[v]$ of an $sv$-walk $Q_v$ which ends with edge $P'[v]$ as set in Lines~\ref{linr:if}-\ref{linr:endS} in Algorithm~\ref{alg:order}. More precisely, suppose $P'[v]=e=(u,v,\tau,0)\in E^{arr}_{\tau,\lambda=0,\alpha=0}$. In the case where $u=s$ and the cost $c=B'[v]$ is  indeed $\gamma(e)$, we define $Q_v=\langle e\rangle$.  Otherwise, $e$ must have an associated cost $c'=B_{tail}[u]$ computed using Algorithm~\ref{alg:gal} and we have $c=B_{tail}[u]\oplus\gamma(e)$. The correctness of Algorithm~\ref{alg:gal} implies that $c'$ is the cost of an $su$-walk $Q_e$ that $e$ can extend. We then define $Q_v=Q_e.e$ whose cost is precisely $\gamma_{Q_e.e}=B_{tail}[u]\oplus\gamma(e)=c$.
%where $e$ is the edge with head $v$ providing the smallest cost $c$. We denote by $P'[v]$ such $e$. % notation P'[v] is already defined !!!

Second, as edge $P'[v]$ has arrival time $\tau$, $Q_v$ also has arrival time $\tau$ and any edge $(v,w,\tau,0)$ can extend it as long as $\alpha_v=0$. More generally each arc $(x,y,c)$ in a path $P$ from $v$ in $D$ corresponds to an edge $f=(x,y,\tau,0)$ with cost $\gamma(f)=c$ and such that $\alpha_x=\alpha_y=0$ according to the construction of $A'$ at Line~\ref{linr:begdigr}. The condition $\alpha_x=0$ implies that $f$ extends the edge associated to the arc preceding $(x,y,c)$ in $P$ or extends $P'[v]=e$ if it is the first arc. The path $P$ is thus associated to a walk $Q_P$ of edges in $E^{arr}_{\tau,\lambda=0,\alpha=0}$ such that $Q_v.Q_P$ is a walk.
Moreover, the cost of $P$ in $D$ is obtained as $\gamma^D(B',P)=(\cdots(B'[v]\oplus c_1)\cdots\oplus c_{k-1})\oplus c_k$ where $c_1,\ldots,c_k$ denote the respective costs of arcs in $P$. As $B'[v]=\gamma_{Q_e}$, we get $\gamma^D(B',P)=\gamma_{Q_v.Q_P}$.

\begin{claim}\label{cla:pathQ}
There exists a path $P_Q$ from $S$ to $v_{k-1}=u_k$ in $D$ that has cost $\gamma^D(B',P_Q)\preceq\gamma_{\langle e_1, \dots, e_{k-1} \rangle}$.
\end{claim}

To prove this, we decompose $\langle e_1, \dots, e_{k-1}\rangle=Q^1.Q^2$ where $Q^2=\langle e_i,\ldots, e_{k-1}\rangle$ is its longest suffix of edges in the zero-block $E^{arr}_{\tau,\lambda=0,\alpha=0}$ containing $e_{k-1}$. Note that all edges of $\langle e_1, \dots, e_{k-1}\rangle$ belonging to the zero-block must be consecutive as edges are sorted according to non-decreasing arrival time in a walk and no edge of the zero-block can be extend by and edge having positive travel time or positive minimum waiting-time. 

First, consider the moment when $e_i=(u_i,v_i,\tau,0)$ is considered at Line~\ref{linr:if} in Algorithm~\ref{alg:order} for possibly updating $B'[v_i]$. 
If its associated cost is not $\bottom$, we have $B_{tail}[u_i]\preceq \gamma_{Q^1}$. The reason is that $B_{tail}[u_i]$ is the cost of a walk $R$ that $e_i$ extends, which is composed of edges scanned so far, and such that $R.e_i$ is  $(\Bar{E}^{dep},\Bar{E}^{arr})$-respected. Moreover, it has minimum cost among such walks according to Lemma~\ref{lem:kthiter}. We thus have $B_{tail}[u_i]\preceq \gamma_{Q^1}$ since $Q^1$ is $(\Bar{E}^{dep},\Bar{E}^{arr})$-respected by our choice of $e_k$. Otherwise, $Q^1$ is empty and we must have $u_i=s$ and $i=1$. In this latter case, we let $Q^1.e_i$ denote the walk $\langle e_i\rangle$. In both cases, $c$ is set at Lines~\ref{linr:c1}-\ref{linr:c2} to a value such that $c\preceq \gamma_{Q^1.e_i}$. The update of $B'[v_i]$ according to Lines~\ref{linr:ifbprime}-\ref{linr:bprime} then ensures $B'[v_i]\preceq \gamma_{Q^1.e_i}$.
%
%\lv{Why do we say this?:
%Notice that, thanks to Lemma \ref{lem:paths}, without loss of generality on the choice of $Q$, since it has minimum-cost, we can assume that walk $\langle e_{i+1},\ldots, e_{k-1}\rangle$ is indeed a path, because if it was a walk we could just remove any cycle and $Q$ we would still be a minimum-cost walk ending with $e_k$.}
Second, each edge $e_j=(u_j,v_j,\tau,0)$ for $j>i$, is associated to an arc $e'_j=(u_j,v_j,\gamma(e_j))$ in $D$ according to the construction of $A'$ at Line~\ref{linr:begdigr}. Let $W_Q$ denote the walk $\langle e'_{i+1},\ldots, e'_{k-1}\rangle$ in $D$ which has cost $\gamma^D(B',W_Q)=(\cdots(B'[v_i]\oplus \gamma(e_{i+1}))\cdots)\oplus \gamma(e_{k-1})$. As $B'[v_i]\preceq \gamma_{Q^1.e_i}$, we get $\gamma^D(B',W_Q)\preceq \gamma_Q$ according to isotonicity. According to Lemma~\ref{lem:paths}, there exists a path $P_Q$ in $D$ from $v_i$ to $v_{k-1}$ satisfying $\gamma^D(B',P_Q)\preceq \gamma^D(B',W_Q)\preceq \gamma_Q$. $P_Q$ is thus a path from $v_i\in S$ to $v_{k-1}$ with cost $\gamma^D(B',P_Q)\preceq \gamma_{\langle e_1, \dots, e_{k-1} \rangle}$ as claimed.

%\fb{questo cammino non contiene $e_k$, ma attenzione, quando prendiamo quello con costo minimo potrebbe contenerlo? Verificare bene}

%\fb{First rule out the case in which $e_k$ is $P'[v_k]$}
\begin{claim}
There exists a walk $\Tilde{Q}$ such that $\Tilde{Q}.e_k$ is a $(\Bar{E}^{dep}, \Bar{E}^{arr})$-respected minimum-cost walk among the walks from $s$ that end with $e_k$.
\end{claim}

This claim will clearly conclude the proof of correctness. We first note that $u_k$ is reachable from $S$ according to Claim~\ref{cla:pathQ} through a path $P_Q$ with cost $\gamma^D(B',P_Q)\preceq\gamma_{\langle e_1, \dots, e_{k-1} \rangle}$. Lemma~\ref{lem:dijkstra} then ensures that $F$ contains a minimum-cost path $P$ in $D$. Its cost must thus satisfy $\gamma^D(B',P)\preceq \gamma^D(B',P_Q)$. %\fb{rephrase these two sentences}. 
%Notice that both $P$ and $P_Q$ correspond according to Claim~\ref{cla:pathwalk} to some walks in $G$ that can be extended by edge $e_k$, since $e_k$ extends $e_{k-1}$.
Isotonicity then implies $\gamma^D(B',P)\oplus \gamma(e_k)\preceq \gamma^D(B',P_Q)\oplus \gamma(e_k)\preceq \gamma_{\langle e_1, \dots, e_{k-1} \rangle}\oplus \gamma(e_k) = \gamma_Q$.

%\lv{Treat the case $e_k$ is a $P'[v]$ edge for some $v$. As $v$ is then the head of $e_k$, this is equivalent to $e_k=P'[v_k]$.}

Let us denote with $\Tilde{Q} = \langle \Tilde{e}_1, \dots, %\Tilde{e}_h,\Tilde{e}_{h+1}, \dots,
\Tilde{e}_l \rangle$ the walk corresponding to $P$ according to Claim~\ref{cla:pathwalk}, and $\Tilde{e}_j = (\Tilde{u}_j,\Tilde{v}_j, \Tilde{\tau}_j, \Tilde{\lambda}_j)$ for $j=1, \dots,l$. According to the construction of $\Tilde{Q}$ in Claim~\ref{cla:pathwalk}, let $h$ be the index of the edge $P'[v]$, that is $h$ is the (only) index satisfying $\Tilde{e}_{h}=P'[\Tilde{v}_{h}]$. Note that the subsequent edges $\Tilde{e}_{h+1}, \dots, \Tilde{e}_l$ correspond to the arcs of $P$. On the other hand, edges $\Tilde{e}_{1}, \dots, \Tilde{e}_{h-1}$ precede the zero-block $E^{arr}_{\tau,\lambda=0,\alpha=0}$ in $E^{arr}$. Note that they thus also precede $e_k$ in $\Bar{E}^{arr}$ as $e_k$ extends an edge of the zero-block and thus satisfy $arr(e_k)\ge \tau$ and $\alpha_{u_k}=0$.

\paragraph{We can assume $e_k$ is not a $P'[v]$ edge for some $v$.} If this was the case, as $v$ is then the head of $e_k$, this is equivalent to $e_k=P'[v_k]$. Then the walk $Q''=\langle \Tilde{e}_1, \dots, \Tilde{e}_{h} \rangle$ ends with edge $e_k = \Tilde{e}_{h}$ and has cost $\gamma_{Q''}\preceq \gamma_{\Tilde{Q}}$ by $C$-non-negativity of the edges $\Tilde{e}_{h+1}, \dots, \Tilde{e}_{l}$. As Claim~\ref{cla:pathwalk} guarantees  $\gamma_{\Tilde{Q}}=\gamma^D(B',P)$, we then have $\gamma_{Q''}\preceq \gamma^D(B',P)\oplus \gamma(e_k)$ by $C$-non-negativity of $e_k$, implying $\gamma_{Q''}\preceq \gamma_Q$. As $\Tilde{e}_{h-1}$ is not in the zero-block at time $\tau$, then either $\Tilde{e}_{h-1}$ is not in a zero-block or $arr(\Tilde{e}_{h-1}) < dep(e_k)$, and we can conclude as case A.

\medskip

We now prove that $\Tilde{Q}.e_k$ is a walk: first $e_k$ is not an edge of $\Tilde{Q}$ and second, $e_k$ extends $\Tilde{Q}$. 
The only case where $e_k$ could be in $\Tilde{Q}$ is when the edge $e_k$ itself belongs to the zero-block $E^{arr}_{\tau,\lambda=0,\alpha=0}$. Let us rule out this eventuality. 
As $P$ is a path and not a walk, $\Tilde{v}_l=u_k$ cannot be the tail of $\Tilde{e}_j$ for any $j\in [h+1,l]$. We also just proved we can assume $e_k \neq \Tilde{e}_{h}$.
%If $e_k = \Tilde{e}_{h}$, then the walk $Q''=\langle \Tilde{e}_1, \dots, \Tilde{e}_{h} \rangle$ ends with edge $e_k = \Tilde{e}_{h}$ and has cost $\gamma_{Q''}\preceq \gamma_{\Tilde{Q}}$ by $C$-non-negativity of the edges $\Tilde{e}_{h+1}, \dots, \Tilde{e}_{l}$. As Claim~\ref{cla:pathwalk} guarantees  $\gamma_{\Tilde{Q}}=\gamma^D(B',P)$, we then have $\gamma_{Q''}\preceq \gamma^D(B',P)\oplus \gamma(e_k)$ by $C$-non-negativity of $e_k$, implying $\gamma_{Q''}\preceq \gamma_Q$. As $\Tilde{e}_{h-1}$ is not in the zero-block,
%we can conclude as in the first case considered in which $e_{k-1}$ does not belong to a zero-block. In the following we will thus assume $e_k \neq P'[v_{k}]$, and $e_k$ cannot be an edge of $\Tilde{Q}$.
%
Second, $e_k$ extends $\Tilde{Q}$. The reason is that it extends $e_{k-1}$ which belongs to the zero-block $E^{arr}_{\tau,\lambda=0,\alpha=0}$. We thus have $\tau +\alpha_{u_{k}}=\tau\le dep(e_k)\le \tau+\beta_{u_{k}}$. As $\Tilde{e}_l$ also belongs to the zero-block, it has also arrival time $\tau$ and $e_k$ also extends $\Tilde{e}_l$ since $\Tilde{v}_l=u_k$.

It just remains to prove that $\Tilde{Q}.e_k$ is $(\Bar{E}^{dep},\Bar{E}^{arr})$-respected.
Since $\Tilde{e}_{h-1} <_{\Bar{E}^{arr}} e_k$, our choice of $e_k$ implies that we can assume without of loss of generality that $\langle \Tilde{e}_1, \dots, \Tilde{e}_{h-1} \rangle$ is $(\Bar{E}^{dep},\Bar{E}^{arr})$-respected. Thus, we have to consider the $(\Bar{E}^{dep},\Bar{E}^{arr})$-respected property with respect to the following three types of pairs of consecutive edges in $\Tilde{Q}.e_k$:
\begin{itemize}
    \item[1)] $\Tilde{e}_{h-1}, \Tilde{e}_{h}$,
    \item[2)] $\Tilde{e}_j, \Tilde{e}_{j+1}$ for $j = h, \dots, l-1$,
    \item[3)] $\Tilde{e}_l, e_k$.
\end{itemize}

In Case~1, we have to prove that for each edge $e'$ such that $tail(e')=tail(\Tilde{e}_{h})=\Tilde{u}_{h}$ and $\Tilde{e}_{h} \leq_{\Bar{E}^{dep}} e'$ we have $\Tilde{e}_{h-1} <_{\Bar{E}^{arr}} e'$. As $\Tilde{e}_{h}$ extends $\Tilde{e}_{h-1}$ and $\Tilde{e}_{h} \leq_{\Bar{E}^{dep}} e'$, we have $arr(\Tilde{e}_{h-1}) \leq dep(\Tilde{e}_{h}) \leq dep(e') \leq arr(e')$. If $arr(\Tilde{e}_{h-1}) < arr(e')$ we can conclude since $\Bar{E}^{arr}$ is sorted by non-decreasing arrival time. Thus let us suppose that $arr(\Tilde{e}_{h-1}) = arr(e')$, which implies that the travel time of $e'$ is zero and $arr(\tilde{e}_{h-1})=dep(\tilde{e}_{h})=\tau$. Moreover, we have $\alpha_{tail(\Tilde{e}_{h})} = \alpha_{tail(e')}= 0$ since $e_{h} \in E^{arr}_{\tau,\lambda=0,\alpha=0}$. Since $\tilde{e}_{h-1}$ is not in the zero-block $E^{arr}_{\tau,\lambda=0,\alpha=0}$ and $arr(\tilde{e}_{h-1})=\tau$, we must have $\alpha_{tail(\Tilde{e}_{h-1})} >0$. This implies $\Tilde{e}_{h-1} <_{\Bar{E}^{arr}} e'$ since in $\Bar{E}^{arr}$ the block of edges in $E^{arr}_{a,\lambda=0,\alpha_{tail}>0}$ precedes block $E^{arr}_{a,\lambda=0,\alpha=0}$ and $E^{arr}_{a,\lambda=0,\alpha_{head}>0}$.
%\lv{FB:due to the scan order of the lists with same arrival time $E^{arr}$; LV: we should say something like that very early}.

In Case~2, we have to prove that for each edge $e'$ such that $tail(e')=tail(\Tilde{e}_{j+1})=\Tilde{u}_{j+1}$ and $\Tilde{e}_{j+1} \leq_{\Bar{E}^{dep}} e'$ we have $\Tilde{e}_j <_{\Bar{E}^{arr}} e'$. Again, we have $arr(\Tilde{e}_{j}) \leq dep(\Tilde{e}_{j+1}) \leq dep(e') \leq arr(e')$ and if $arr(\Tilde{e}_{j}) < arr(e')$ we can conclude. So let us suppose $arr(\Tilde{e}_j) = arr(e')$, implying that $e'$ has zero travel time and departure $\tau$. If $\alpha_{head(e')>0}$, we then have $\Tilde{e}_j <_{\Bar{E}^{arr}} e'$ since $\Tilde{e}_j$ belongs to the zero-block and $e'$ is scanned after the zero-block. Otherwise, we have $\alpha_{head(e')=0}$ and $e'$ is in the zero-block $E^{arr}_{\tau,\lambda=0,\alpha=0}$. Having $e' = P'[head(e')]$ would contradict $\Tilde{e}_{j+1} \leq_{\Bar{E}^{dep}} e'$. The reason is that $\Tilde{e}_{j+1}$ is in $F$ while the swaps performed at Line~\ref{linr:swap} ensures that alls edges $P'[v]$ with tail $tail(\Tilde{e}_{j+1})$ for some $v\in V'_{head}$ precede other edges with same tail and same departure time in $\Bar{E}^{dep}$. We can thus assume $e' \neq P'[head(e')]$.
If $e' \in F$, notice that either $\Tilde{e}_j = P'[\Tilde{v}_j]$ or $\Tilde{e}_j\in F$ and precede $e'$ in the BFS ordering, since $tail(e') = head(\Tilde{e}_j)$, and thus in both cases we have $\Tilde{e}_j <_{\Bar{E}^{arr}} e'$. Finally, if $e'$ is not in $F$, it is added at the end of the reordered zero-block at Line~\ref{linr:remain} and we again have $\Tilde{e}_j <_{\Bar{E}^{arr}} e'$.

In Case~3, we have to prove that for each edge $e'$ such that $tail(e')=tail(e_k)=u_k$ and $e_k \leq_{\Bar{E}^{dep}} e'$ we have $\Tilde{e}_l <_{\Bar{E}^{arr}} e'$. %Both in the case that $e_k$ does and does not belong to the zero-block $E^{arr}_{\tau,\lambda=0,\alpha=0}$, 
Since $e_k \leq_{\Bar{E}^{dep}} e'$ and $e_k \neq P'[v_k]$, we deduce that $e'$, as $e_k$ is ordered after edges with tail $u_k$ of the form $P'[v]$ for some $v\in V'_{head}$, which implies $e' \neq P'[head(e')]$. The proof is now similar to the previous case: if $arr(\Tilde{e}_l) < arr(e')$, we can directly conclude. Otherwise, $e'$ has zero travel time and arrival time $\tau$. In the three cases $\alpha_{head(e')}>0$, $e'\in F$ and $e'\notin F$, we conclude similarly.
This achieves the proof of correctness

\bigskip
\noindent
\textit{Complexity analysis.}

We finally analyze the complexity of Algorithm~\ref{alg:0del}. As discussed previously, sorting $E^{arr}$ by non-decreasing arrival time can be done in $O(|E|\log |V|)$ time by computing the lists $(E_v^{arr})_{v\in V}$ and then merging them using a priority queue. Once $E^{arr}$ is sorted by non decreasing arrival time we can partition the list of edges with same arrival time into $E^{arr}_{a,\lambda>0}$, $E^{arr}_{a,\lambda=0,\alpha_{tail}>0}$, $E^{arr}_{a,\lambda=0,\alpha=0}$ and $E^{arr}_{a,\lambda=0,\alpha_{head}>0}$ by bucket sorting into the four lists.

The calls to $\reorder(E^{arr}_{\tau,\lambda=0,\alpha=0},\tau)$ incur the only other additional costs compared to Algorithm~\ref{alg:gal}. The nodes in $V'_{tail}$ and $V'_{head}$ are identified in $O(|E^{arr}_{\tau,\lambda=0,\alpha=0}|)$, and both sets cardinality is bounded by $|E^{arr}_{\tau,\lambda=0,\alpha=0}|$. Thus the construction of digraph $D$ is linear in $|E^{arr}_{\tau,\lambda=0,\alpha=0}|$. In order to identify interval $I$ at Line~\ref{linr:findint} we might scan up to $|\intervs{u}|$ intervals. However, all interval scanned contain edges with departure time less or equal to $\tau$. This means that, at the moment of the next call to $\reorder$() these intervals will have already been removed by the calls to $\processcosts()$ at Line~\ref{ling:callonu} of Algorithm~\ref{alg:gal}, and thus will not be scanned again. Overall, we can bound the time complexity of this operation by the total number of intervals created during the execution of the algorithm, which is bounded by $|E|$. Similarly, the computation of index $i$ at Line~\ref{linr:findindex} requires to scan all the edges from $E_u^{dep}[l_u]$ until an edge with departure time $\tau$ is found. We know that such an edge exists in interval $I$ as some edges in the zero-block have tail $u$, and again, at the moment of the next call to $\reorder$() these edges will have already been processed, and thus will not be scanned again. Overall, we can bound the time complexity of this operation by the total number of edges $|E|$.
Computing set $S$ and computing for each node $v$ in such set the cost $B'[v]$ and the edge $P'[v]$ is linear in $|E^{arr}_{\tau,\lambda=0,\alpha=0}|$, since it requires  a single scan of the edges in the zero-block and few constant time operations per edge. 

The time complexity of Algorithm~\ref{alg:order} is thus dominated by the Dijkstra call which costs time $O(|A'|+n'\log n')$ where $n'=|V'_{tail}\cup V'_{head}|\le |V|$. As $|A'|=|E^{arr}_{\tau,\lambda=0,\alpha=0}|$ and $n'=O(|E^{arr}_{\tau,\lambda=0,\alpha=0}|)$, the overall time complexity of the calls to Algorithm~\ref{alg:order} is $O(\sum_{a\in A^{arr}}|E^{arr}_{\tau,\lambda=0,\alpha=0}|(1+\log |V|))=O(|E|\log |V|)$.
The space complexity is clearly linear.

%--------------------------------------------------------------

% As the number $n'=|V'|$ of nodes considered in such a call  is at most $2|E^{arr}_{\tau,\lambda=0}|$, the construction of the graph is linear in $|E^{arr}_{\tau,\lambda=0}|$. The time complexity is thus dominated by the Dijkstra call and is $O(|A'|+n'\log n') = O(|E^{arr}_{\tau,\lambda=0}|\log |V|)$ as $|A'\le |E^{arr}_{\tau,\lambda=0}|$, $n'=O(|E^{arr}_{\tau,\lambda=0}|)$ and $n'\le |V|$. Hence, the overall time complexity of the calls to Algorithm~\ref{alg:order} is $\sum_{a\in A^{arr}}O(|E^{arr}_{\tau,\lambda=0}|\log |V|)=O(|E|\log |V|)$.
\end{proof}

\paragraph*{Optimizing a linear combination of classical criteria.}
A temporal graph with the cost structure defined in \Cref{sec:solving} for optimizing a linear combination of classical criteria also satisfies the absorption property under the following assumption: for any edge $e=(u,v,\tau,\lambda)$ such that $\lambda=0$, $\alpha_u=0$ and $\alpha_v=0$, we require $\delta(e)=\delta_5 c(e) +\delta_6 \ge 0$. This implies that for any cost $(\tau,\Delta)\in \mathbb{R}\times\mathbb{R}$, we indeed have $(\tau,\Delta)\preceq (\tau,\Delta)\oplus \gamma(e)$ as $(\tau,\Delta)\oplus\gamma(e)=(\tau,\Delta+\delta(e))$ and $\Delta+\delta(e)\ge \Delta$.
%Note that both isotonicity and absorption hold more generally for any real values of $\delta_1,\ldots,\delta_7\in\mathbb{R}$ and costs $(c(e))_{e\in E}$ as long as $\delta_5 c(e) +\delta_6 \ge 0$ for all edges $e\in E$ with zero travel time.
If we assume $\delta_5,\delta_6 \ge 0$, non-negativity of costs is required only for edges with zero travel time, and negative values are allowed for $\delta_1,\ldots,\delta_4$ and $\delta_7$.  \Cref{th:zero} thus extends the result of~\cite{BentertHNN2020} to a wider range of linear combinations, and has a slightly better complexity.

\begtodolater

\lv{If we need a neutral element...}

\paragraph*{Optimizing a linear combination of classical criteria (zero travel time part).}
We now define an extension of the cost structure defined in \Cref{sec:solving} for optimizing a linear combination of classical criteria in order to have a neutral element. We consider an undefined departure time $\bottom$ and add a neutral cost $0_C=(\bottom,0)$ so that we now consider $C=(\mathbb{R}\times\mathbb{R})\cup\{0_C\}$. We define the cost of an edge similarly as before and extend $\oplus$ and $\preceq$ as follows. For any $(\tau,\Delta)\in \mathbb{R}\times\mathbb{R}$ we define $(\bottom,0)\oplus(\tau,\Delta)=(\tau,\Delta)$, $(\tau,\Delta)\oplus(\bottom,0)=(\tau,\Delta)$ and $(\bottom,0)\oplus(\bottom,0)=(\bottom,0)$, implying that $0_C=(\bottom,0)$ is indeed a neutral element. We also define $(\bottom,0)\preceq (\tau,\Delta)$ (resp. $ (\tau,\Delta)\preceq (\bottom,0)$) as the condition $0\le \Delta$ (resp. $\Delta\le 0$).

We now show that the isotonicity property is still satisfied, that is: for any $c,c_1,c_1\in C$ such hat $c1\preceq c2$, we have $c\oplus c_1\preceq c \oplus c_2$. We have already proved the case where $c,c_1,c_2\in \mathbb{R}\times\mathbb{R}$. If $c=0_C$, we have $0_C\oplus c_1=c_1$ and $0_C\oplus c_2=c_2$, implying $0_C\oplus c_1 \preceq 0_C\oplus c_2$. Now suppose that we have $c=(\tau,\Delta)\in \mathbb{R}\times\mathbb{R}$. The case $c_1=c_2=0_C$ is also immediate as we then have $c\oplus c_1=c=c\oplus c_2$.  If $c_1=0_C$ and $c_2=(\tau_2,\Delta_2)\in \mathbb{R}\times\mathbb{R}$, the condition $c_1\preceq c_2$ is then equivalent to $0\le \Delta_2$. We then have $-(\delta_2+\delta_3+\delta_7)\tau+\Delta\le -(\delta_2+\delta_3+\delta_7)\tau+\Delta+\Delta_2$ which is equivalent to $c\preceq c\oplus c_2$ as $c\oplus c_2=(\tau,\Delta+\Delta_2)$. Since $c\oplus c_1=c$, we get $c\oplus c_1\preceq c\oplus c_2$. The case where $c_2=0_C$ and $c_1\in \mathbb{R}\times\mathbb{R}$ is similar.

Finally, the absorption property is also satisfied under the following assumption: for any edge $e=(u,v,\tau,\lambda)$ such that $\lambda=0$ and $\alpha_u=0$, we require $\delta(e) \ge 0$ or equivalently $\delta_5 c(e) +\delta_6\ge 0$. We show that $\gamma(e)$ is then $C$-non-negative, that is: $c\preceq c\oplus \gamma(e)$ for any $c\in C$. For any cost $c=(\tau,\Delta)\in \mathbb{R}\times\mathbb{R}$, we indeed have $(\tau,\Delta)\oplus\gamma(e)=(\tau,\Delta+\delta(e))$,implying $(\tau,\Delta)\preceq (\tau,\Delta)\oplus \gamma(e)$ as we have $\Delta+\delta(e)\ge \Delta$ for such an edge $e$. For cost $c=0_C=(\bottom,0)$, we have $0_C\oplus \gamma(e)=\gamma(e)=(\tau,\delta(e))$ which indeed satisfies $(\bottom,0)\preceq (\tau,\delta(e))$ as long as $0\le \delta(e)$.
In particular, non-negativity of costs is required only for edges with zero travel time if we assume $\delta_5,\delta_6 \ge 0$, and negative values are allowed for $\delta_1,\ldots,\delta_4$ and $\delta_7$.  \Cref{th:zero} thus extends the result of~\cite{BentertHNN2020} to a wider range of linear combinations, with a slightly better time complexity.

\endtodolater

\bibliographystyle{plainurl}
\bibliography{biblio}

\begin{thebibliography}{10}

\bibitem{BentertHNN2020}
Matthias Bentert, Anne{-}Sophie Himmel, Andr{\'{e}} Nichterlein, and Rolf
  Niedermeier.
\newblock Efficient computation of optimal temporal walks under waiting-time
  constraints.
\newblock {\em Appl. Netw. Sci.}, 5(1):73, 2020.

\bibitem{Berman1996}
Kenneth~A. Berman.
\newblock Vulnerability of scheduled networks and a generalization of menger's
  theorem.
\newblock {\em Networks}, 28(3):125--134, 1996.

\bibitem{BrodalJ04}
Gerth~St{\o}lting Brodal and Riko Jacob.
\newblock Time-dependent networks as models to achieve fast exact time-table
  queries.
\newblock {\em Electron. Notes Theor. Comput. Sci.}, 92:3--15, 2004.

\bibitem{BrunelliCV2021}
Filippo Brunelli, Pierluigi Crescenzi, and Laurent Viennot.
\newblock On computing pareto optimal paths in weighted time-dependent
  networks.
\newblock {\em Inf. Process. Lett.}, 168:106086, 2021.

\bibitem{BuiXuanFJ2003}
Binh-Minh Bui-Xuan, Afonso Ferreira, and Aubin Jarry.
\newblock Computing shortest, fastest, and foremost journeys in dynamic
  networks.
\newblock {\em International Journal of Foundations of Computer Science},
  14(02):267--285, 2003.

\bibitem{Bumby1981}
Richard~T. Bumby.
\newblock A problem with telephones.
\newblock {\em SIAM. J. on Algebraic and Discrete Methods}, 2(1):13--18, 1981.

\bibitem{BussMNR2020}
Sebastian Bu{\ss}, Hendrik Molter, Rolf Niedermeier, and Maciej Rymar.
\newblock Algorithmic aspects of temporal betweenness.
\newblock In Rajesh Gupta, Yan Liu, Jiliang Tang, and B.~Aditya Prakash,
  editors, {\em {KDD} '20: The 26th {ACM} {SIGKDD} Conference on Knowledge
  Discovery and Data Mining, Virtual Event, CA, USA, August 23-27, 2020}, pages
  2084--2092. {ACM}, 2020.

\bibitem{Casteigts2012}
Arnaud Casteigts, Paola Flocchini, Walter Quattrociocchi, and Nicola Santoro.
\newblock Time-varying graphs and dynamic networks.
\newblock {\em {IJPEDS}}, 27(5):387--408, 2012.

\bibitem{CasteigtsHMZ2021}
Arnaud Casteigts, Anne{-}Sophie Himmel, Hendrik Molter, and Philipp Zschoche.
\newblock Finding temporal paths under waiting time constraints.
\newblock {\em Algorithmica}, 83(9):2754--2802, 2021.

\bibitem{CookeH1966}
Kenneth~L. Cooke and Eric Halsey.
\newblock The shortest route through a network with time-dependent internodal
  transit times.
\newblock {\em Journal of Mathematical Analysis and Applications},
  14(3):493--498, 1966.

\bibitem{CormenLRS2001}
Thomas~H. Cormen, Charles~E. Leiserson, Ronald~L. Rivest, and Clifford Stein.
\newblock {\em Introduction to Algorithms}, chapter Single-Source Shortest
  Paths and All-Pairs Shortest Paths, page 580–642.
\newblock MIT Press and McGraw-Hill, 2001.

\bibitem{DehneOS2012}
Frank Dehne, Masoud~T. Omran, and Jörg-Rüdiger Sack.
\newblock Shortest {Paths} in {Time}-{Dependent} {FIFO} {Networks}.
\newblock {\em Algorithmica}, 62(1-2):416--435, 2012.

\bibitem{Dibbelt2013}
Julian Dibbelt, Thomas Pajor, Ben Strasser, and Dorothea Wagner.
\newblock Intriguingly {Simple} and {Fast} {Transit} {Routing}.
\newblock In {\em Experimental {Algorithms}}, Lecture {Notes} in {Computer}
  {Science}, pages 43--54. Springer, 2013.

\bibitem{Dibbelt2018}
Julian Dibbelt, Thomas Pajor, Ben Strasser, and Dorothea Wagner.
\newblock Connection scan algorithm.
\newblock {\em {ACM} Journal of Experimental Algorithmics}, 23:1.7:1--1.7:56,
  2018.

\bibitem{Dijkstra59}
Edsger~W. Dijkstra.
\newblock A note on two problems in connexion with graphs.
\newblock {\em Numerische Mathematik}, 1:269--271, 1959.

\bibitem{FuchsleMNR22}
Eugen F{\"{u}}chsle, Hendrik Molter, Rolf Niedermeier, and Malte Renken.
\newblock Delay-robust routes in temporal graphs.
\newblock In Petra Berenbrink and Benjamin Monmege, editors, {\em 39th
  International Symposium on Theoretical Aspects of Computer Science, {STACS}
  2022, March 15-18, 2022, Marseille, France (Virtual Conference)}, volume 219
  of {\em LIPIcs}, pages 30:1--30:15. Schloss Dagstuhl - Leibniz-Zentrum
  f{\"{u}}r Informatik, 2022.

\bibitem{Kahn62}
Arthur~B. Kahn.
\newblock Topological sorting of large networks.
\newblock {\em Commun. {ACM}}, 5(11):558--562, 1962.

\bibitem{Latapy2018}
Matthieu Latapy, Tiphaine Viard, and Cl{\'{e}}mence Magnien.
\newblock Stream graphs and link streams for the modeling of interactions over
  time.
\newblock {\em Social Netw. Analys. Mining}, 8(1):61:1--61:29, 2018.

\bibitem{Michail2016}
Othon Michail.
\newblock An introduction to temporal graphs: An algorithmic perspective.
\newblock {\em Internet Mathematics}, 12(4):239--280, 2016.

\bibitem{MullerHSWZ04}
Matthias M{\"{u}}ller{-}Hannemann, Frank Schulz, Dorothea Wagner, and
  Christos~D. Zaroliagis.
\newblock Timetable information: Models and algorithms.
\newblock In {\em {ATMOS}}, volume 4359 of {\em Lecture Notes in Computer
  Science}, pages 67--90. Springer, 2004.

\bibitem{Nachtigall1995}
Karl Nachtigall.
\newblock Time depending shortest-path problems with applications to railway
  networks.
\newblock {\em European Journal of Operational Research}, 83(1):154--166, 1995.

\bibitem{PallottinoS1997}
Stefano Pallottino and Maria~Grazia Scutell{\`a}.
\newblock Shortest path algorithms in transportation models: classical and
  innovative aspects.
\newblock Technical Report TR-97-06, University of Pisa, 1997.

\bibitem{PallottinoS1998}
Stefano Pallottino and Maria~Grazia Scutell{\`a}.
\newblock {\em Equilibrium and Advanced Transportation Modelling}, chapter
  Shortest path algorithms in transportation models: classical and innovative
  aspects, pages 245--281.
\newblock Kluwer Academic Publishers, 1998.

\bibitem{SchulzWW2000}
Frank Schulz, Dorothea Wagner, and Karsten Weihe.
\newblock Dijkstra's algorithm on-line: An empirical case study from public
  railroad transport.
\newblock {\em {ACM} J. Exp. Algorithmics}, 5:12, 2000.

\bibitem{Simard21}
Fr{\'{e}}d{\'{e}}ric Simard.
\newblock Evaluating metrics in link streams.
\newblock {\em Soc. Netw. Anal. Min.}, 11(1):51, 2021.

\bibitem{Sobrinho2005}
Jo{\~{a}}o~L. Sobrinho.
\newblock An algebraic theory of dynamic network routing.
\newblock {\em {IEEE/ACM} Trans. Netw.}, 13(5):1160--1173, 2005.

\bibitem{Griff2010}
Jo{\~{a}}o~L. Sobrinho and Timothy~G. Griffin.
\newblock Routing in equilibrium.
\newblock {\em 19th International Symposium on Mathematical Theory of Networks
  and System}, pages 941--947, 2010.

\bibitem{Wu2014}
Huanhuan Wu, James Cheng, Silu Huang, Yiping Ke, Yi~Lu, and Yanyan Xu.
\newblock {Path Problems in Temporal Graphs}.
\newblock {\em VLDB Endowment}, 7(9):721--732, 2014.

\bibitem{Wu2016}
Huanhuan Wu, James Cheng, Yiping Ke, Silu Huang, Yuzhen Huang, and Hejun Wu.
\newblock Efficient {Algorithms} for {Temporal} {Path} {Computation}.
\newblock {\em IEEE Transactions on Knowledge and Data Engineering},
  28:2927--2942, 2016.

\end{thebibliography}

%\newpage
%\appendix

%\section{Appendix: full paper at next page}

\end{document}